\documentclass[11pt]{article}
\usepackage{amsmath,amssymb,amsfonts,latexsym,graphicx,amsthm}
\usepackage{fullpage,color}
\usepackage{url,hyperref}
\usepackage{comment}
\usepackage[linesnumbered,boxed,ruled,vlined]{algorithm2e}
\usepackage{framed}
\usepackage[shortlabels]{enumitem}
\usepackage{cleveref}
\usepackage{setspace}

\usepackage{tikz}
\usetikzlibrary{backgrounds}
\usetikzlibrary{decorations.pathreplacing,calligraphy}
\usetikzlibrary{decorations.pathmorphing}
\usepackage{subfigure}
\usepackage{pgfplots}
\usetikzlibrary{shapes.misc}
\tikzset{cross/.style={cross out, draw=black, fill=none, minimum size=2*(#1-\pgflinewidth), inner sep=0pt, outer sep=0pt}, cross/.default={10pt}}
\usepackage[margin=1in]{geometry}

\newtheorem{theorem}{Theorem}[section]
\newtheorem{lemma}{Lemma}[section]
\newtheorem{definition}{Definition}[section]

\newtheorem{claim}{Claim}[section]

        {\medskip}

\newcommand{\ceil}[1]{\lceil #1 \rceil}

\newcommand{\dist}{\mathsf{dist}}
\newcommand{\est}{\mathsf{est}}

\newcommand{\wts}{\omega}
\newcommand{\brac}[1]{\left(#1\right)}

\newcommand{\bw}{\mathsf{bw}}

\def\shiri#1{} 

\begin{document}

\begin{titlepage}
	\title{Faster Algorithms for Dual-Failure Replacement Paths}
	\author{Shiri Chechik \thanks{Tel Aviv University, 		\href{}{shiri.chechik@gmail.com}} \and
		Tianyi Zhang \thanks{Tel Aviv University, \href{}{tianyiz21@tauex.tau.ac.il}}}
	\date{}
	
	\maketitle
	\thispagestyle{empty}
	
	\begin{abstract}
		Given a simple weighted directed graph $G = (V, E, \wts)$ on $n$ vertices as well as two designated terminals $s, t\in V$, our goal is to compute the shortest path from $s$ to $t$ avoiding any pair of presumably failed edges $f_1, f_2\in E$, which is a natural generalization of the classical replacement path problem which considers single edge failures only.
		
		This dual failure replacement paths problem was recently studied by Vassilevska Williams, Woldeghebriel and Xu [FOCS 2022] who designed a cubic time algorithm for general weighted digraphs which is conditionally optimal; in the same paper, for unweighted graphs where $\wts \equiv 1$, the authors presented an algebraic algorithm with runtime $\tilde{O}(n^{2.9146})$, as well as a conditional lower bound of $n^{8/3-o(1)}$ against combinatorial algorithms. However, it was unknown in their work whether fast matrix multiplication is necessary for a subcubic runtime in unweighted digraphs.
		
		As our primary result, we present the first truly subcubic combinatorial algorithm for dual failure replacement paths in unweighted digraphs. Our runtime is $\tilde{O}(n^{3-1/18})$. Besides, we also study algebraic algorithms for digraphs with small integer edge weights from $\{-M, -M+1, \cdots, M-1, M\}$. As our secondary result, we obtained a runtime of $\tilde{O}(Mn^{2.8716})$, which is faster than the previous bound of $\tilde{O}(M^{2/3}n^{2.9144} + Mn^{2.8716})$ from [Vassilevska Williams, Woldeghebriela and Xu, 2022].
	\end{abstract}

\end{titlepage}

\begin{spacing}{1.3}
	\tableofcontents
\end{spacing}

\thispagestyle{empty}
\clearpage
\pagenumbering{arabic}
\setcounter{page}{1}

\newpage

\section{Introduction}
In the replacement path problem, we want to understand shortest paths in a directed graph that avoid presumably failed edges. More specifically, let $G = (V, E, \wts)$ be an edge-weighted simple digraph on $n$ vertices and $m$ edges. Fix a pair of source and terminal vertices $s, t\in V$, we want to compute the shortest path from $s$ to $t$ that avoids any designated set $F\subseteq E$ of failed edges.

The most classical setting is when the number of failures is at most one; namely, we want to compute all the values of $\dist(s, t, G\setminus \{f\})$ when $f$ ranges over all edges on the shortest path from $s$ to $t$ in $G$. The complexity of single-failure replacement path is now well-understood. On the hardness side, it was proved that computing all single-failure replacement paths in weighted graphs requires at least $n^{3-o(1)}$ time \cite{williams2010subcubic} assuming the APSP conjecture. To breach the cubic barrier, we need to assume the input digraph has small integer edge weights, or allow approximation errors in the algorithm output. When the edge weights are integers in the range $\{-M, -M+1, \cdots, M-1, M\}$, there is an algorithm with runtime $\tilde{O}(Mn^\omega)$\footnote{$\omega\in [2, 2.371552]$ is the fast matrix multiplication exponent \cite{williams2024new,duan2023faster,alman2021refined,le2014powers,williams2012multiplying}.} \cite{chechik2020simplifying,williams2011faster}. For the special case when the input digraph is unweighted ($\wts\equiv 1$), there is a combinatorial algorithm (algorithms not using fast matrix multiplication) with runtime $\tilde{O}(mn^{1/2})$ \cite{roditty2005replacement}, which is optimal under the hardness of combinatorial boolean matrix multiplication \cite{williams2010subcubic}.

A natural extension is to study replacement paths when there are two edge failures. We are interested in fast algorithms that compute for all pairs of edges $f_1, f_2\in E$ the value of $\dist(s, t, G\setminus\{f_1, f_2\})$. This problem was first studied in \cite{bhosle2004replacement} and recently revisited in \cite{williams2022algorithms}. For general weighted digraphs, the authors of \cite{williams2022algorithms} designed an algorithm with runtime $\tilde{O}(n^3)$, which is the same as the easier single failure replacement paths problem. When the graph has small edge weights from range $\{-M, -M+1, \cdots, M-1, M\}$, in the same paper the authors have shown subcubic runtime upper bound of $\tilde{O}(M^{2/3}n^{2.9144} + Mn^{2.8716})$ using fast matrix multiplication. Finally, as complementary to their algorithms, the authors showed a conditional lower bound of $n^{8/3-o(1)}$ against combinatorial algorithms for unweighted digraphs assuming the hardness of boolean matrix multiplication.

According to the results in \cite{williams2022algorithms}, there is a gap in their understanding about dual-failure replacement paths in unweighted graphs. On the one hand, their algebraic algorithm computes dual-failure replacement paths in $\tilde{O}(n^{2.9144})$ by setting $M = 1$; on the other hand, their conditional lower bound against combinatorial algorithms is also subcubic. So, it is not clear whether combinatorial algorithm can achieve subcubic runtime as well, or the conditional lower bound can be improved to cubic.

\subsection{Our results}
In this paper, we first study fast combinatorial algorithms for unweighted digraphs and show that subcubic runtime can indeed be achieved without using fast matrix multiplication.
\begin{theorem}\label{subcubic}
	Given a simple unweighted directed graph $G = (V, E)$ on $n$ vertices, and fix any pair of vertices $s, t\in V$, the values of all dual-failure replacement path distances $\dist(s, t, G\setminus\{f_1, f_2\}), \forall f_1, f_2\in E$ can be computed in  $\tilde{O}(n^{3-1/18})$ time with high probability; most importantly, the algorithm does not use fast matrix multiplication.
\end{theorem}

Here, a digraph is simple if it does not contain two edges between the same pair of vertices with the same direction. Secondly, we also study fast algebraic algorithms for dual-failure replacement paths when the edge weights are from the set $\{-M, -M+1, \cdots, M-1, M\}$.

\begin{theorem}\label{faster-alge}
	Given a simple weighted directed graph $G = (V, E, \wts)$ on $n$ vertices along with integer edge weights $\wts : E\rightarrow \{-M, -M+1, \cdots, M-1, M\}$ without negative cycles, and fix any pair of vertices $s, t\in V$, the values of all dual-failure replacement path distances $\dist(s, t, G\setminus\{f_1, f_2\}), \forall f_1, f_2\in E$ can be computed in time $\tilde{O}(Mn^{2.8716})$.
\end{theorem}

\subsection{Other related works}
The replacement paths problem has also been studied in other settings, including the single-source setting \cite{gu2021faster,chechik2020near,grandoni2019faster,grandoni2012improved} and the approximation setting \cite{chechik2024nearly,bernstein2010nearly,roditty2007k}.

\subsection{Subcubic combinatorial algorithm for unweighted digraphs}
\subsubsection{One failure on a long $st$-path}
Let us first consider the case where one edge failure $f_1$ lies on the shortest $st$-path $\pi$, while the other one $f_2$ does not. This case would be easy when we use fast matrix multiplication as did by \cite{williams2022algorithms}, but it becomes  complicated when we are restricted to purely combinatorial algorithms.

As a preliminary step, we first show how to deal with the case where $|\pi| > L$ for some parameter $L$ slightly larger than $n^{0.5}$, say $L = n^{0.55}$. Partition the $st$-path $\pi$ into sub-paths of length exactly $5L$ as $\pi = \gamma_1\circ\gamma_2\circ\cdots\circ \gamma_h, h\leq \ceil{n/5L}$, and let $s_i, t_i$ be the endpoints of sub-path $\gamma_i$. Assume only one edge failure $f_1$ falls on sub-path $\gamma_i$, and let $\rho$ be the optimal replacement path from $s$ to $t$ avoiding $\{f_1, f_2\}$. 

Take a random set of vertices $U$ of size $O(n\log n / L)$.  If $|\rho\setminus \pi| > L$, then with high probability $\rho\setminus \pi$ would contain a vertex $p$ from $U$. Then, we can compute single-source single-failure replacement paths to and from $p$ in graph $G\setminus E(\pi)$ that takes runtime $\tilde{O}(n^{3.5}/L)$ \cite{chechik2020near}, so that for each vertex $z\in V$, we know the shortest path between $z, p$ in graph $G\setminus (E(\pi)\cup \{f_2\})$. Using this information, we will be able to recover $\rho$.

Now suppose $|\rho\setminus \pi| < L$. Then in this case, we will show that the detour parts of different dual-failure replacement paths $\rho$ are vertex-disjoint from each other when the first edge failure $f_1$ comes from different choice of sub-paths $\gamma_i$. Then, we can partition $G$ into vertex-disjoint subgraphs $G_1, G_2, \cdots, G_h$, such that the dual-failure replacement paths for failures on $\gamma_i$ belong to subgragh $G_i$, and solve the dual-failure replacement paths problem for source-terminal pair $(s_i, t_i)$ in graph $G_i$. See \Cref{overview-1fail-long} for an illustration.

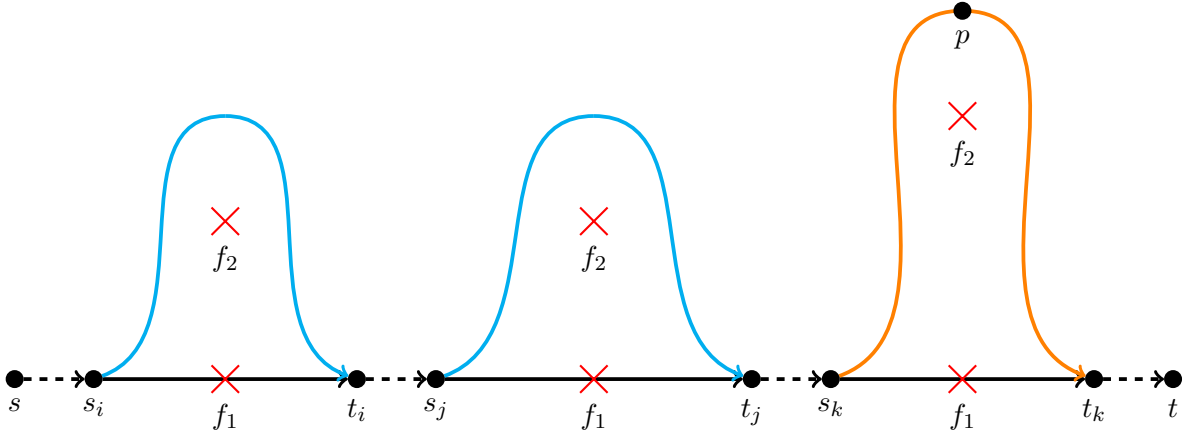
\begin{figure}
	\centering
	\begin{tikzpicture}[thick,scale=0.7]
	\draw (-11, 0) node(1)[circle, draw, fill=black, inner sep=0pt, minimum width=6pt, label=-90: {$s$}] {};
	\draw (11, 0) node(2)[circle, draw, fill=black, inner sep=0pt, minimum width=6pt, label=-90: {$t$}] {};
	
	\draw (-9.5, 0) node(3)[circle, draw, fill=black, inner sep=0pt, minimum width=6pt, label=-90: {$s_i$}] {};
	\draw (-4.5, 0) node(4)[circle, draw, fill=black, inner sep=0pt, minimum width=6pt, label=-90: {$t_i$}] {};
	
	\draw (-3, 0) node(5)[circle, draw, fill=black, inner sep=0pt, minimum width=6pt, label=-90: {$s_j$}] {};
	\draw (3, 0) node(6)[circle, draw, fill=black, inner sep=0pt, minimum width=6pt, label=-90: {$t_j$}] {};
	
	\draw (4.5, 0) node(7)[circle, draw, fill=black, inner sep=0pt, minimum width=6pt, label=-90: {$s_k$}] {};
	\draw (9.5, 0) node(8)[circle, draw, fill=black, inner sep=0pt, minimum width=6pt, label=-90: {$t_k$}] {};
	
	\draw (7, 7) node(9)[circle, draw, fill=black, inner sep=0pt, minimum width=6pt, label=-90: {$p$}] {};
	

	
	\draw (0, 0) node[cross=6, red, label=-90: {$f_1$}] {};
	\draw (0, 3) node[cross=6, red, label=-90: {$f_2$}] {};
	
	\draw (-7, 0) node[cross=6, red, label=-90: {$f_1$}] {};
	\draw (-7, 3) node[cross=6, red, label=-90: {$f_2$}] {};
	
	\draw (7, 0) node[cross=6, red, label=-90: {$f_1$}] {};
	\draw (7, 5) node[cross=6, red, label=-90: {$f_2$}] {};

\begin{scope}[on background layer]
	\draw [->, line width = 0.5mm, dashed] (1) to (3);
	\draw [->, line width = 0.5mm, dashed] (4) to (5);
	\draw [->, line width = 0.5mm, dashed] (6) to (7);
	\draw [->, line width = 0.5mm, dashed] (8) to (2);
	
	\draw [->, line width = 0.5mm] (3) to (4);
	\draw [->, line width = 0.5mm] (5) to (6);
	\draw [->, line width = 0.5mm] (7) to (8);
	
	\draw [line width = 0.5mm, color = cyan] (5) to[out=20, in=180] (0, 5);
	\draw [->, line width = 0.5mm, color = cyan] (0, 5) to[out=0, in=160] (6);
	
	\draw [line width = 0.5mm, color = cyan] (3) to[out=20, in=180] (-7, 5);
	\draw [->, line width = 0.5mm, color = cyan] (-7, 5) to[out=0, in=160] (4);
	
	\draw [line width = 0.5mm, color = orange] (7) to[out=20, in=180] (7, 7);
	\draw [->, line width = 0.5mm, color = orange] (7, 7) to[out=0, in=160] (8);
\end{scope}

\end{tikzpicture}
	\caption{For simplicity, let us assume that when $f_1$ falls on the sub-paths $\pi[s_i, t_i]$, the dual-failure replacement path also passes through $s_i, t_i$. In this picture, the two cyan dual-failure replacement paths have length less than $L$, and we will show that they are vertex-disjoint. The orange dual-failure replacement path has length at least $L$, and so it hits a vertex $p\in U$ with high probability; in this case, we will compute single-source single-failure replacement paths to and from $p$ in graph $G\setminus E(\pi)$ to help us compute dual-failure replacement paths.}
	\label{overview-1fail-long}
\end{figure}

\subsubsection{One failure on a short $st$-path}
By the previous subsection, we have reduced to the case that the $st$-path has length at most $L$. So, for the rest, let us rename the problem instance and assume that $|\pi|\leq L$. For the $i$-th edge $e_i$ on $\pi$ ($0\leq i<L$), we can compute the optimal replacement path from $s$ to $t$ avoiding $e_i$ and let $\alpha_i$ be the corresponding detour whose endpoints are $a_i$ and $b_i$. Again, through some case analysis, we can assume that the span of each detour $\alpha_i$ which is $|\pi[a_i, b_i]|$ is at most $g$ for some parameter $g$ slightly larger than $n^{1/3}$ (say $g = n^{0.35}$). Consider any edge failure $f_i\in E(\alpha_i)$, as a simplification let us assume that the optimal replacement path from $s$ to $t$ avoiding $\{e_i, f_i\}$ is a concatenation: $$\rho = \pi[s, a_i]\circ \alpha_i[a_i, x_i]\circ \gamma_i\circ \alpha_i[y_i, b_i]\circ \pi[b_i, t]$$
where $\gamma_i$ is a detour with respect to $\alpha_i$ in $G\setminus E(\pi)$. Via some case analysis, we can assume that $|\alpha_i|<L$. See \Cref{overview-1fail-short1} for an illustration.

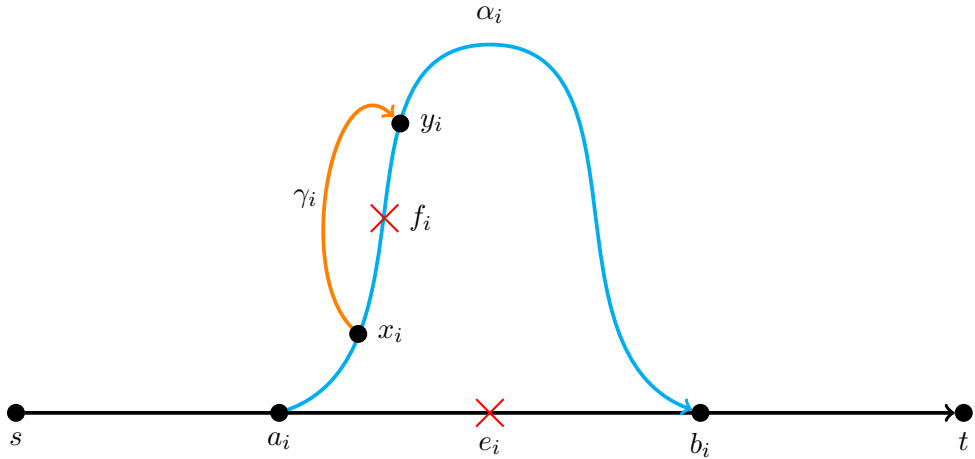
\begin{figure}
	\centering
	\begin{tikzpicture}[thick,scale=0.7]
	\draw (-9, 0) node(1)[circle, draw, fill=black, inner sep=0pt, minimum width=6pt, label=-90: {$s$}] {};
	\draw (9, 0) node(2)[circle, draw, fill=black, inner sep=0pt, minimum width=6pt, label=-90: {$t$}] {};
	
	\draw (-4, 0) node(3)[circle, draw, fill=black, inner sep=0pt, minimum width=6pt, label=-90: {$a_i$}] {};
	\draw (4, 0) node(4)[circle, draw, fill=black, inner sep=0pt, minimum width=6pt, label=-90: {$b_i$}] {};
	
	\draw (-2.5, 1.5) node(5)[circle, draw, fill=black, inner sep=0pt, minimum width=6pt, label=0: {$x_{i}$}] {};
	\draw (-1.7, 5.5) node(6)[circle, draw, fill=black, inner sep=0pt, minimum width=6pt, label=0: {$y_{i}$}] {};
	
	\draw (0, 0) node[cross=6, red, label=-90: {$e_i$}] {};
	\draw (-2, 3.7) node[cross=6, red, label=0: {$f_i$}] {};
		
	\draw (0, 7) node[label=90: {$\alpha_i$}] {};
	\draw (-3.5, 3.5) node[label=90: {$\gamma_i$}] {};
	
	\begin{scope}[on background layer]
		\draw [->, line width = 0.5mm] (1) to (2);
		\draw [line width = 0.5mm, color = cyan] (3) to[out=20, in=180] (0, 7);
		\draw [->, line width = 0.5mm, color = cyan] (0, 7) to[out=0, in=160] (4);
		\draw [->, line width = 0.5mm, color = orange] (5) to[out=135, in=135] (6);
	\end{scope}
\end{tikzpicture}
	\caption{The cyan path is the detour $\alpha_i$ that avoids $e_i$, and the orange path is the detour $\gamma_i$ that avoids $f_i$ which lies on $\alpha_i$. Via some case analysis, we will show the difficult case is that $|\alpha_i| < L$.}
	\label{overview-1fail-short1}
\end{figure}

Let us first consider the case when $|\gamma_i| < g$. A wishful thought is that for two different choice of dual failures $\{e_i, f_i\}$ and $\{e_j, f_j\}$ where $e_i$ and $e_j$ are well-separated on the shortest path $\pi$ ($|\pi(e_i, e_j)| \geq 10g$), we are guaranteed that the two dual-failure detours $\gamma_i$ and $\gamma_j$ are vertex-disjoint; if this is the case, then we can partition the graph $G$ into $O(L/g)$ vertex-disjoint subgraphs and compute dual-failure replacement paths separately. However, this is generally not true. The key observation is that when $f_i$ and $f_j$ are roughly at the same height on the detour, namely $|\alpha_i[a_i, f_i)|\approx |\alpha_j[a_j, f_j)|$ up to an additive error of at most $g$, such a disjointness condition indeed holds. Therefore, our algorithm will further partition each detour $\alpha_i$ into sub-paths of length $g$ as $\alpha_i = \beta^i_1\circ \beta^i_2\circ\cdots\circ \beta^i_l$. Then, fix any height index $1\leq h\leq l$, we will deal with all the dual failures $\{e_i, f_i\}, \forall 1\leq i \leq L/g, \forall f_i\in E(\beta^i_h)$ at the same time. See \Cref{overview-1fail-short2} for an illustration.

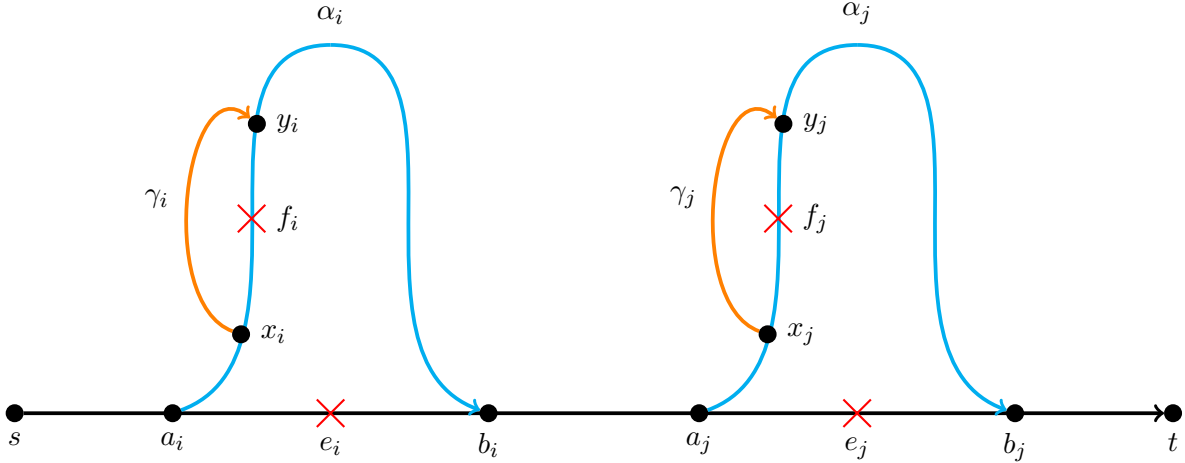
\begin{figure}
	\centering
	\begin{tikzpicture}[thick,scale=0.7]
	\draw (-11, 0) node(1)[circle, draw, fill=black, inner sep=0pt, minimum width=6pt, label=-90: {$s$}] {};
	\draw (11, 0) node(2)[circle, draw, fill=black, inner sep=0pt, minimum width=6pt, label=-90: {$t$}] {};
	
	\draw (-8, 0) node(3)[circle, draw, fill=black, inner sep=0pt, minimum width=6pt, label=-90: {$a_i$}] {};
	\draw (-2, 0) node(4)[circle, draw, fill=black, inner sep=0pt, minimum width=6pt, label=-90: {$b_i$}] {};
	
	\draw (-6.7, 1.5) node(5)[circle, draw, fill=black, inner sep=0pt, minimum width=6pt, label=0: {$x_{i}$}] {};
	\draw (-6.4, 5.5) node(6)[circle, draw, fill=black, inner sep=0pt, minimum width=6pt, label=0: {$y_{i}$}] {};
	
	\draw (-5, 0) node[cross=6, red, label=-90: {$e_i$}] {};
	\draw (-6.5, 3.7) node[cross=6, red, label=0: {$f_i$}] {};
		
	\draw (-5, 7) node[label=90: {$\alpha_i$}] {};
	\draw (-8.3, 3.5) node[label=90: {$\gamma_i$}] {};

	\draw (2, 0) node(7)[circle, draw, fill=black, inner sep=0pt, minimum width=6pt, label=-90: {$a_j$}] {};
	\draw (8, 0) node(8)[circle, draw, fill=black, inner sep=0pt, minimum width=6pt, label=-90: {$b_j$}] {};
	
	\draw (3.3, 1.5) node(9)[circle, draw, fill=black, inner sep=0pt, minimum width=6pt, label=0: {$x_{j}$}] {};
	\draw (3.6, 5.5) node(10)[circle, draw, fill=black, inner sep=0pt, minimum width=6pt, label=0: {$y_{j}$}] {};
	
	\draw (5, 0) node[cross=6, red, label=-90: {$e_j$}] {};
	\draw (3.5, 3.7) node[cross=6, red, label=0: {$f_j$}] {};
	
	\draw (5, 7) node[label=90: {$\alpha_j$}] {};
	\draw (1.7, 3.5) node[label=90: {$\gamma_j$}] {};
	
	\begin{scope}[on background layer]
		\draw [->, line width = 0.5mm] (1) to (2);
		\draw [line width = 0.5mm, color = cyan] (3) to[out=20, in=180] (-5, 7);
		\draw [->, line width = 0.5mm, color = cyan] (-5, 7) to[out=0, in=160] (4);
		\draw [->, line width = 0.5mm, color = orange] (5) to[out=160, in=140] (6);
		
		\draw [line width = 0.5mm, color = cyan] (7) to[out=20, in=180] (5, 7);
		\draw [->, line width = 0.5mm, color = cyan] (5, 7) to[out=0, in=160] (8);
		\draw [->, line width = 0.5mm, color = orange] (9) to[out=160, in=140] (10);
	\end{scope}
\end{tikzpicture}
	\caption{For two well-separated edges $e_i, e_j$ such that $|\pi(e_i, e_j)| \geq 10g$, if both $|\gamma_i|, |\gamma_j|$ are less than $g$, and $f_i, f_j$ are roughly at the same height (i.e., $|\alpha_i[a_i, f_i)|\approx |\alpha_j[a_j, f_j)|$), then we can show that the two dual-failure detours $\gamma_i, \gamma_j$ are vertex-disjoint.}
	\label{overview-1fail-short2}
\end{figure}

Now, what happens if $|\gamma_i| \geq g$? We can take a random sample $U$ of pivot vertices of size $O(n\log n / g)$. Then, with high probability, $\gamma_i$ contains a vertex in $U$. For simplicity, assume both endpoints of the detour $\gamma_i$ are lying within the sub-path $\beta^i_h$ which contains the second edge failure $f_i$, then we could compute single-source shortest paths to and from each vertex $p\in U$ in the subgraph $G\setminus \brac{E(\pi)\cup E(\beta^i_h)}$ which takes time $\tilde{O}(n^3/g)$, then we can compute each detour $\gamma_i$ for each choice of $f_i$ on $\beta^i_h$ in time $\tilde{O}(g^2)$ by guessing the positions of $x_i, y_i$.

Unfortunately, there are $L^2/g$ different choices for the subgraph $G\setminus \brac{E(\pi)\cup E(\beta^i_h)}$, and thus we do not have enough time to compute single-source shortest paths for each vertex in $U$ in all these subgraphs. In fact, we can only afford to compute single-source shortest paths for vertices in $U$ in the subgraph $G_h = G\setminus \brac{E(\pi)\cup \bigcup_{i=0}^{L-1} E(\beta^i_h)}$; that is, for each index $h$, remove all sub-paths $\beta^i_h, \forall 0\leq i < L$ from $G\setminus E(\pi)$ simultaneously (which becomes $G_h$), and compute multi-source shortest paths to and from $U$ in $G_h$. The key observation is that the detour $\gamma_i$ cannot touch vertices on $\beta^j_h$ if $|i-j|> 10g$. Therefore, to compute $\gamma_i$, we can build a small shortcut graph consisting of all vertices in $U\cup\bigcup_{j=i-10g}^{i+10g}V(\beta^j_h)$ which contains all edges in $\bigcup_{j=i-10g}^{i+10g}E(\beta^j_h)\setminus E(\beta^i_h)$ and all shortcut edges to and from $p\in U$ weighted by the single-source distances we have computed in $G_h$. See \Cref{overview-1fail-short3} for an illustration.

\begin{figure}
	\centering
	\begin{tikzpicture}[thick,scale=0.7]
	\draw (-11, 0) node(1)[circle, draw, fill=black, inner sep=0pt, minimum width=6pt, label=180: {$s$}] {};
	\draw (11, 0) node(2)[circle, draw, fill=black, inner sep=0pt, minimum width=6pt, label=0: {$t$}] {};
	
	\draw (-2, 0) node(3)[circle, draw, fill=black, inner sep=0pt, minimum width=6pt, label=-90: {$a_i$}] {};
	\draw (2, 0) node(4)[circle, draw, fill=black, inner sep=0pt, minimum width=6pt, label=-90: {$b_i$}] {};
	
	\draw (-3, 0) node(5)[circle, draw, fill=black, inner sep=0pt, minimum width=6pt, label=-90: {$a_{i-1}$}] {};
	\draw (1, 0) node(6)[circle, draw, fill=black, inner sep=0pt, minimum width=6pt, label=-90: {$b_{i-1}$}] {};
	
	\draw (-1, 0) node(7)[circle, draw, fill=black, inner sep=0pt, minimum width=6pt, label=-90: {$a_{i+1}$}] {};
	\draw (3, 0) node(8)[circle, draw, fill=black, inner sep=0pt, minimum width=6pt, label=-90: {$b_{i+1}$}] {};
	
	\draw (-9, 0) node(9)[circle, draw, fill=black, inner sep=0pt, minimum width=6pt, label=-90: {$a_{i-10g}$}] {};
	\draw (-5, 0) node(10)[circle, draw, fill=black, inner sep=0pt, minimum width=6pt, label=-90: {$b_{i-10g}$}] {};
	
	\draw (5.5, 0) node(11)[circle, draw, fill=black, inner sep=0pt, minimum width=6pt, label=-90: {$a_{i+10g}$}] {};
	\draw (9.5, 0) node(12)[circle, draw, fill=black, inner sep=0pt, minimum width=6pt, label=-90: {$b_{i+10g}$}] {};

	\draw (-4, 9) node(13)[circle, draw, fill=black, inner sep=0pt, minimum width=6pt] {};	
	\draw (-2, 9) node(14)[circle, draw, fill=black, inner sep=0pt, minimum width=6pt] {};
	\draw (0, 9) node(15)[circle, draw, fill=black, inner sep=0pt, minimum width=6pt] {};
	\draw (2, 9) node(16)[circle, draw, fill=black, inner sep=0pt, minimum width=6pt] {};
	\draw (4, 9) node(17)[circle, draw, fill=black, inner sep=0pt, minimum width=6pt] {};

	\draw (-8.7, 4) node(18)[circle, draw, fill=black, inner sep=0pt, minimum width=6pt] {};	
	\draw (5.8, 4) node(20)[circle, draw, fill=black, inner sep=0pt, minimum width=6pt] {};
	
	\draw (0, 0) node[cross=6, red, label=-90: {$e_i$}] {};
	\draw (-1.5, 4.3) node[cross=6, red, label=-90: {$f_i$}] {};
	
	
	
	
	\begin{scope}[on background layer]
		\draw [->, line width = 0.5mm] (1) to (2);
		
		\draw [->, line width = 0.5mm, color = cyan, dotted] (3) to[out=90, in=-90] (-2, 3);
		\draw [->, line width = 0.5mm, color = cyan] (-2, 3) to[out=90, in=180] (0, 5);
		\draw [->, line width = 0.5mm, color = cyan, dotted] (0, 5) to[out=0, in=90] (4);
		
		\draw [->, line width = 0.5mm, color = cyan, dotted] (5) to[out=90, in=-90] (-3, 3);
		\draw [->, line width = 0.5mm, color = cyan] (-3, 3) to[out=90, in=180] (-1, 5);
		\draw [->, line width = 0.5mm, color = cyan, dotted] (-1, 5) to[out=0, in=90] (6);
		
		\draw [->, line width = 0.5mm, color = cyan, dotted] (7) to[out=90, in=-90] (-1, 3);
		\draw [->, line width = 0.5mm, color = cyan] (-1, 3) to[out=90, in=180] (1, 5);
		\draw [->, line width = 0.5mm, color = cyan, dotted] (1, 5) to[out=0, in=90] (8);
		
		\draw [->, line width = 0.5mm, color = cyan, dotted] (9) to[out=90, in=-90] (-9, 3);
		\draw [->, line width = 0.5mm, color = cyan] (-9, 3) to[out=90, in=180] (-7, 5);
		\draw [->, line width = 0.5mm, color = cyan, dotted] (-7, 5) to[out=0, in=90] (10);
		
		\draw [->, line width = 0.5mm, color = cyan, dotted] (11) to[out=90, in=-90] (5.5, 3);
		\draw [->, line width = 0.5mm, color = cyan] (5.5, 3) to[out=90, in=180] (7.5, 5);
		\draw [->, line width = 0.5mm, color = cyan, dotted] (7.5, 5) to[out=0, in=90] (12);
		
		\draw [->, line width = 0.5mm, color = orange] (18) to (13);
		\draw [->, line width = 0.5mm, color = orange] (14) to (20);
		\draw [->, line width = 0.5mm, color = orange] (18) to[out=-30, in=210] (20);
		\draw [->, line width = 0.5mm, color = orange] (14) to[out=60, in=120] (16);
		\draw [->, line width = 0.5mm, color = orange] (13) to[out=60, in=120] (17);
		
	\end{scope}
\end{tikzpicture}
	\caption{A shortcut graph that helps computing the dual-failure detour $\gamma_i$ avoiding $\{f_i, g_i\}$, which consists of vertices in $U\cup\bigcup_{j=i-10g}^{i+10g}V(\beta^j_h)$. This shortcut graph includes all edges in $\bigcup_{j=i-10g}^{i+10g}E(\beta^j_h)\setminus E(\beta^i_h)$, and some shortcut edges representing distances to and from $U$ in $G_h$; actually, it will also contain some shortcut edges between vertices in $\bigcup_{j=i-10g}^{i+10g}V(\beta^j_h)$ which we have not discussed in the overview.}
	\label{overview-1fail-short3}
\end{figure}
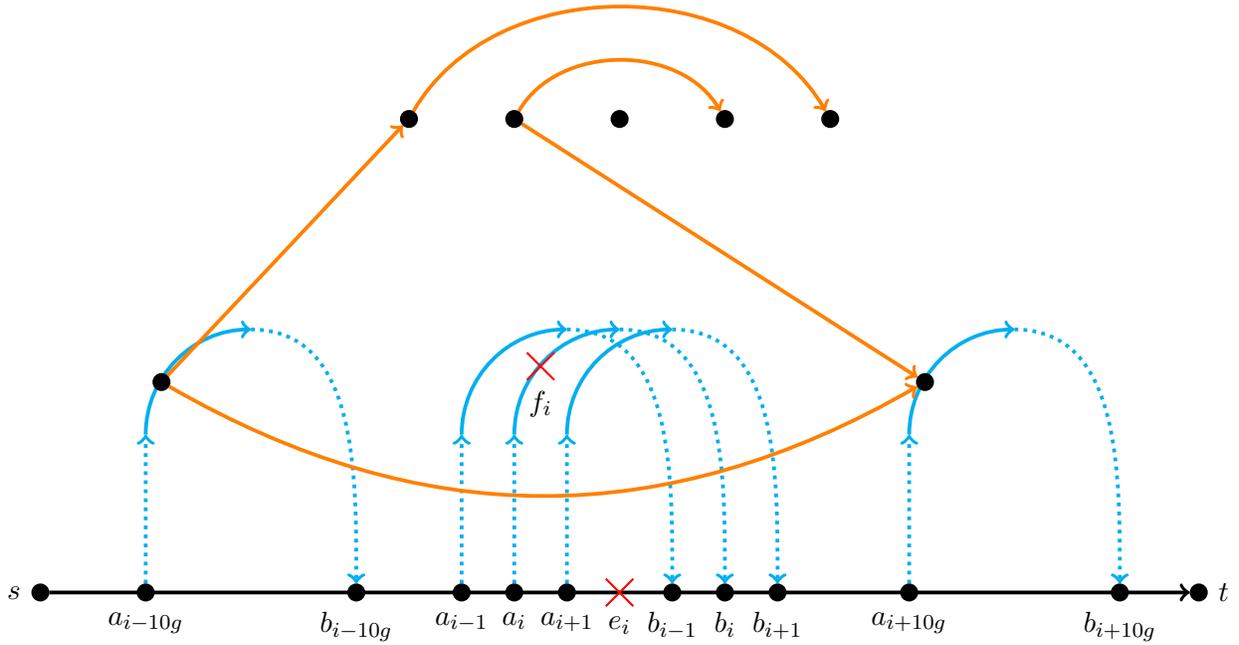

\subsubsection{Both failures on the $st$-path}
Now let us assume both edge failures lie on $\pi$. As the same in \cite{williams2022algorithms}, the main difficulty of dual-failure replacement path for this case comes from the backward paths. More specifically, given two edge failures $\{f_1, f_2\}$, in general the optimal dual replacement path $\rho$ has three parts.
\begin{itemize}[leftmargin=*]
	\item A prefix of $\rho$ that diverges from $\pi$ before the first edge failure $f_1$ on $\pi[s, f_1)$, and then meets somewhere in the middle $y\in V(\pi(f_1, f_2))$ using edges in $G\setminus E(\pi)$.
	
	\item A sub-path of $\rho$ that travels from $y$ to another vertex $x\in V(\pi(f_1, y])$ using edges in $(G\setminus E(\pi))\cup E(\pi(f_1, f_2))$; this sub-path is the so-called \emph{backward} path of $\rho$ which may converge and diverge multiple times with $\pi(f_1, f_2)$.
	
	\item A suffix of $\rho$ that converges with $\pi$ after the second edge failure $f_2$ on $\pi(f_2, t]$ starting from $x$ using edges in $G\setminus E(\pi)$, and then reach $t$ using the rest of $\pi$.
\end{itemize}
To compute the backward path, let us divide $\pi$ into sub-paths of length $L$ for some parameter $L$ as $\pi = \gamma_1\circ \gamma_2\circ\cdots\circ \gamma_{n/L}$. Take a random pivot vertex set $U$ of size $O(\frac{n\log n}{L})$. The main observation is that if $f_1$ and $f_2$ are in different sub-paths $\gamma_i, \gamma_j, j-i > 1$, plus that $x\in V(\gamma_i), y\in V(\gamma_j)$, then the prefix $\rho[s, y]$ must have length at least $L$ and thus contain a pivot $p\in U$ with high probability. Therefore, if we compute single-source replacement paths \cite{chechik2020near} from $p$ in graph $G\setminus E(\gamma_i)$, then it would provide useful information about the backward path of $\rho$ from $y$ to $x$. See \Cref{overview-2fail} for an illustration.

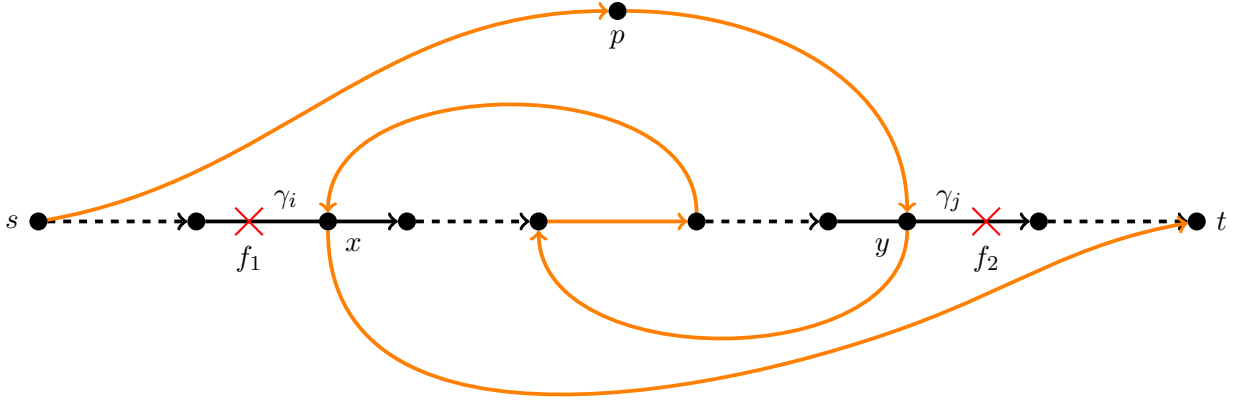
\begin{figure}
	\centering
	\begin{tikzpicture}[thick,scale=0.7]
	\draw (-11, 0) node(1)[circle, draw, fill=black, inner sep=0pt, minimum width=6pt, label=180: {$s$}] {};
	\draw (11, 0) node(2)[circle, draw, fill=black, inner sep=0pt, minimum width=6pt, label=0: {$t$}] {};
	
	\draw (-8, 0) node(3)[circle, draw, fill=black, inner sep=0pt, minimum width=6pt] {};
	\draw (-4, 0) node(4)[circle, draw, fill=black, inner sep=0pt, minimum width=6pt] {};
	
	\draw (-1.5, 0) node(5)[circle, draw, fill=black, inner sep=0pt, minimum width=6pt] {};
	\draw (1.5, 0) node(6)[circle, draw, fill=black, inner sep=0pt, minimum width=6pt] {};
	
	\draw (4, 0) node(7)[circle, draw, fill=black, inner sep=0pt, minimum width=6pt] {};
	\draw (8, 0) node(8)[circle, draw, fill=black, inner sep=0pt, minimum width=6pt] {};
	
	\draw (0, 4) node(9)[circle, draw, fill=black, inner sep=0pt, minimum width=6pt, label=-90: {$p$}] {};
	
	\draw (-5.5, 0) node(10)[circle, draw, fill=black, inner sep=0pt, minimum width=6pt, label=-45:{$x$}] {};
	\draw (5.5, 0) node(11)[circle, draw, fill=black, inner sep=0pt, minimum width=6pt, label=-135:{$y$}] {};
	
	\draw (-7, 0) node[cross=6, red, label=-90: {$f_1$}] {};
	\draw (7, 0) node[cross=6, red, label=-90: {$f_2$}] {};
	
	\draw (-6.3, 1) node[label=-90: {$\gamma_i$}] {};
	\draw (6.3, 1) node[label=-90: {$\gamma_j$}] {};
	
	\begin{scope}[on background layer]
		\draw [->, line width = 0.5mm] (3) to (4);
		\draw [->, line width = 0.5mm] (7) to (8);
		\draw [->, line width = 0.5mm, color=orange] (5) to (6);
		
		\draw [->, line width = 0.5mm, dashed] (1) to (3);
		\draw [->, line width = 0.5mm, dashed] (4) to (5);
		\draw [->, line width = 0.5mm, dashed] (6) to (7);
		\draw [->, line width = 0.5mm, dashed] (8) to (2);
		
		\draw [->, line width = 0.5mm, color=orange] (1) to[out=10, in=180] (9);
		\draw [->, line width = 0.5mm, color=orange] (9) to[out=0, in=90] (11);
		\draw [->, line width = 0.5mm, color=orange] (11) to[out=-90, in=-90] (5);
		\draw [->, line width = 0.5mm, color=orange] (6) to[out=90, in=90] (10);
		\draw [line width = 0.5mm, color=orange] (10) to[out=-90, in=200] (5.5, -2);
		\draw [->, line width = 0.5mm, color=orange] (5.5, -2) to[out=20, in=190] (2);

	\end{scope}
\end{tikzpicture}
	\caption{When $f_1, f_2$ lie in different sub-paths $\gamma_i, \gamma_j$ such that $j-i>1$, we can show that $\rho$ must contain a vertex in $U$ with high probability. Then we can apply the algorithm from \cite{chechik2020near} to compute single-source replacement path from $p$ in graph $G\setminus E(\gamma_i)$ to learn about the sub-path $\rho[p, x]$ which contains the backward path in the middle.}
	\label{overview-2fail}
\end{figure}

\subsection{Faster algebraic algorithm for weighted digraphs}
Let us divide the shortest path $\pi$ from $s$ to $t$ into sub-paths of hops at most $L$; that is, $\pi = \gamma_1\circ\gamma_2\circ\cdots\circ\gamma_{\ceil{n/L}}$. Given a pair of edge failures $\{f_1, f_2\}$, suppose $f_1$ and $f_2$ belong to $\gamma_{l_1}$ and $\gamma_{l_2}$ respectively. In the previous paper \cite{williams2022algorithms}, the difficult case is when $f_1, f_2$ come from different sub-paths, say $l_1 < l_2$. To compute $\dist(s, t, G\setminus \{f_1, f_2\})$ efficiently, their approach was to build a sketch graph $H_{f_1, f_2}$ on vertices $\{s ,t\}\cup V(\gamma_{l_1})\cup V(\gamma_{l_2})$ which encodes an optimal replacement path, and then run $s$-$t$ shortest path in $H_{f_1, f_2}$ which takes $\tilde{O}(L^2)$ time for each $\{f_1, f_2\}$, leading to a total runtime overhead of $\tilde{O}(n^2L^2)$.

Divide each sub-path $\gamma_{l}$ into segments $\gamma_l = \alpha^l_1\circ \alpha^l_2\circ \cdots \circ \alpha^l_{\ceil{L/g}}$ of hops at most $g$ (for some parameter $g<L$). Assume $f_1$ lies on $\alpha^{l_1}_{h_1}$ and $f_2$ lies on $\alpha^{l_2}_{h_2}$, and let $\rho$ be the optimal replacement path from $s$ to $t$ avoiding $\{f_1, f_2\}$. Our main observation is that, if $\rho$ intersects both $\alpha^{l_1}_{h_1}, \alpha^{l_2}_{h_2}$, then we can build a smaller sketch graph $H_{f_1, f_2}$ on $\{s, t\}\cup V\brac{\alpha^{l_1}_{h_1}}\cup V\brac{\alpha^{l_2}_{h_2}}$ which still encodes $\rho$, and so the runtime would be reduced to $\tilde{O}(g^2)$. Otherwise, if $\rho$ skips over $\alpha^{l_2}_{h_2}$ entirely, we will build a sketch graph $H_{f_1, (l_2, h_2)}$ on the vertex set $\{s ,t\}\cup V(\gamma_{l_1})\cup V(\gamma_{l_2})$ which only depends on $f_1$ and $\alpha^{l_2}_{h_2}$ and not on $f_2$. As the number of such graphs $H_{f_1, (l_2, h_2)}$ is at most $O(n^2/g)$, the runtime can be bounded by $\tilde{O}(n^2L^2/g)$. Overall, the runtime would be $\tilde{O}(n^2g^2 + n^2L^2/g)$ which is always better than the previous bound of $\tilde{O}(n^2L^2)$.

\subsection{Organization}
In Section \ref{1fail-short} \ref{1fail-long} \ref{2fail-unweighted} we design fast combinatorial algorithms for unweighted digraphs; in Section \ref{alge} we design fast algebraic algorithms for weighted digraphs.

\section{Preliminaries}
Throughout the paper, logarithm $\log(*)$ will have base $2$, and we assume the number of vertices $n$ in the input graph is an integral power of $2$. In any weighted digraph $G = (V, E, \wts)$, for any vertex $u\in V$, let $\deg(u, G)$ be its vertex degree (counting both in-edges and out-edges); for any pair of vertices $u, v\in V$, let $\dist(u, v, G)$ be the weighted length of the shortest path from $u$ to $v$ in $G$. Throughout the algorithm, we will maintain a distance estimation function $\est(*, *, *)$, such that the value $\est(u, v, G)\geq \dist(u, v, G)$ is always an upper bound. All values of $\est(*, *, *)$ are initially infinity and are non-increasing throughout the algorithms. When we update the value of an estimation $\est(u, v, G)$ with a distance value $D$, we mean $\est(u, v, G)\leftarrow \{D, \est(u, v, G) \}$.

Given any directed path or walk $\rho$ in $G$, let $|\rho|$ be the number of edges on $\rho$, and let $\wts(\rho)$ be its total edge weight. For any two vertices $u, v\in V(\rho)$ where $u$ comes before $v$ on $\rho$, let $\rho[u, v]\subseteq \rho$ be the sub-path of $\rho$ from $u$ to $v$. In addition, let $\rho[*, v]$ and $\rho[u, *]$ be the prefix and suffix sub-path of $\rho$; this notation will be useful when we don't have variable names for the endpoints of path $\rho$.

For any edge $f\in E(\rho)$ and vertex $u\in V(\rho)$ which comes before edge $f$, let $\rho[u, f)$ be the sub-path from $u$ to $f$ (excluding edge $f$); similarly we can define notations $\rho(f, v]$ and $\rho(f_1, f_2)$ in the natural way.

Borrowing a terminology from \cite{williams2022algorithms}, let us define the notion of canonical paths.
\begin{definition}[\cite{williams2022algorithms}]
	Let $s, t\in V$ be two vertices, and let $\pi$ be a shortest path from $s$ to $t$. For any edge set $F\subseteq E$, a path $\rho$ from $s$ to $t$ in $G\setminus F$ is called \emph{canonical} with respect to $\pi$ and $F$, if for any $u, v\in V(\pi)\cap V(\rho)$ such that $u$ appears before $v$ in both $\pi, \rho$ and $E(\pi[u, v])\cap F=\emptyset$, then $\rho[u, v]$ is the same as $\pi[u, v]$.
\end{definition}

\subsection{Unweighted digraphs}
For tie-breaking among shortest paths, we can randomly perturb all unit edge weights slightly so that all replacement shortest paths for at most two edge failures are unique under the perturbed weights. 
We can show that, under the weight perturbation, all replacement shortest paths are canonical. Next, let us state a basic property regarding shortest replacement paths for one edge failure.
\begin{lemma}\label{detour}
	Consider any edge $f\in E(\pi)$, any canonical shortest path $\rho$ from $s$ to $t$ avoiding $f$ can be decomposed as $\rho = \pi[s, a]\circ \alpha\circ\pi[b, t]$, where $a, b\in V(\pi)$ and $\alpha$ is a shortest path from $a$ to $b$ in $G\setminus E(\pi)$. For convenience, $\alpha$ will be called the \emph{detour} of the replacement path, and $a, b$ are called the \emph{divergence} and \emph{convergence} vertex, respectively.
\end{lemma}

It is known that shortest replacement paths for one failure can be computed efficiently.

\begin{lemma}[\cite{roditty2005replacement}]\label{rp}
	Given an unweighted digraph $G = (V, E)$ on $n$ vertices and $m$ edges. Fixing any source-terminal pair $s, t\in V$, we can compute all canonical shortest replacement paths from $s$ to $t$ avoiding $f$ with high probability in time $\tilde{O}(mn^{1/2} + n^2)$, where $f$ ranges over all edges in $E$.
\end{lemma}

We also need a basic property about replacement paths for dual edge failures where only one edge falls on the $st$-path $\pi$.
\begin{lemma}
	Consider any edge $f_1\in E(\pi)$ and let $\alpha$ be the detour of the shortest replacement path that avoids $f_1$. Consider any edge $f_2\in E(\alpha)$. Then, there is a shortest replacement path $\rho$ avoiding $\{f_1, f_2\}$ such that:
	\begin{itemize}[leftmargin=*]
		\item $\rho$ is a canonical shortest path avoiding $f_1$ in graph $G\setminus\{f_1\}$.
		\item $\rho$ diverges and converges with $\pi$ for once.
	\end{itemize}
\end{lemma}

We will be applying the following algorithm for single-source replacement paths algorithm from \cite{chechik2020near} as a black-box.

\begin{lemma}[\cite{chechik2020near}]\label{ssrp}
	Given an unweighted digraph $G = (V, E)$ on $n$ vertices and $m$ edges. Fixing any vertex $s\in V$, we can compute all shortest replacement paths from $s$ to $t$ avoiding $f$ with high probability in time $\tilde{O}(mn^{1/2} + n^2)$, where $t$ ranges over all vertices in $V\setminus \{s\}$ and $f$ ranges over all edges in $E$.
\end{lemma}

To be consistent with our perturbation-based tie-braking, we should also impose the same edge weight perturbation on the graph on which \Cref{ssrp} is applied; although \Cref{ssrp} is only stated for unweighted digraphs, it also works with edge weight perturbation (for example, all hitting set arguments still work, since perturbation only changes path lengths negligibly).

In the original paper \cite{chechik2020near}, they only claim to compute all the length of shortest replacement paths, but here we need the actual shortest paths tree in each graph $G\setminus \{f\}$. To achieve such an augmentation, during the execution of the algorithm in \cite{chechik2020near}, we can keep track of the last edge of each replacement path, and by uniqueness of shortest paths under weight perturbation, the set of these last edges form the shortest paths tree.

Finally, one of our basic tools is a \emph{truncated} version of Dijkstra's algorithm, which is stated below.
\begin{lemma}[\cite{dijkstra2022note}]
	Given an unweighted digraph $G = (V, E)$ on $n$ vertices. For any vertex $s\in V$ and any threshold $L$, let $U = \{u\mid \dist(s, u, G)\leq L\}$. Then we can compute shortest paths from $s$ to all vertices in $U$ in time $O(\sum_{u\in U}\deg(u, G) + n\log n)$.
\end{lemma}

\subsection{Weighted digraphs}
When the input graph contains negative edge weights, we assume it does not contain any negative cycles. For weighted graphs, since fast matrix multiplication algorithms only work with small integer edge weights, we will not assume unique shortest paths by perturbing the edge weights.

Our algorithm relies on fast algorithms which computes shortest paths in digraphs with negative edge weights.
\begin{lemma}[\cite{bernstein2022negative}]\label{sssp-neg}
	Given a weighted digraph $G = (V, E, \wts)$ with edge weights $\wts: E\rightarrow \{-M, -M+1, \cdots, M-1, M\}$ without negative cycles, and fix any source vertex $s\in V$. Then, single-source shortest path from $s$ can be computed in time $\tilde{O}(m\log M)$ with high probability.
\end{lemma}

\begin{lemma}[\cite{zwick2002all}]\label{apsp-neg}
		Given a weighted digraph $G = (V, E, \wts)$ with edge weights $\wts: E\rightarrow \{-M, -M+1, \cdots, M-1, M\}$ without negative cycles, all-pairs shortest paths in $G$ can be computed in time $\tilde{O}(M^{\frac{1}{4-\omega}}n^{2+\frac{1}{4-\omega}})$ with high probability.
\end{lemma}

We will apply single-source multi-terminal replacement paths algorithms as a black-box; the currently best known upper bound is stated below.
\begin{lemma}[\cite{gu2021faster}]\label{ssrp-alge}
	Given a weighted digraph $G = (V, E, \wts)$ with edge weights $\wts: E\rightarrow \{-M, -M+1, \cdots, M-1, M\}$  without negative cycles, and fix any source vertex $s\in V$ and a terminal set $T\subseteq V$. Then, the value of all $\dist(s, t, G\setminus \{f\}), \forall t\in T, f\in E$ can be computed in time $\tilde{O}(Mn^\omega + M^{\frac{1}{4-\omega}}n^{1+\frac{1}{4-\omega}}\cdot |T|)$ with high probability.
\end{lemma}

We will use the following fast algorithm for distance sensitivity oracles in weighted digraphs. In a weighted digraph $G = (V, E, \wts)$ with edge weights $\wts: E\rightarrow \{-M, -M+1, \cdots, M-1, M\}$, a distance sensitivity oracle is an efficient data structure that answers $\dist(u, v, G\setminus \{f\})$ for any $u, v\in V, f\in E$. 

\begin{lemma}[\cite{chechik2020distance}]\label{dso}
	Given a weighted digraph $G = (V, E, \wts)$ with edge weights $\wts: E\rightarrow \{-M, -M+1, \cdots, M-1, M\}$  without negative cycles, a distance sensitivity oracle with $\tilde{O}(1)$ query time can be constructed in time $\tilde{O}(Mn^{2.8719})$.
\end{lemma}

\section{One failure on a short $st$-path}\label{1fail-short}

We use two parameters $g$ and $L$ which will be determined in the end such that $g < n^{1/2} < L$.  In this section, we study the case where only one edge failure lies on a short $st$-path, plus that the input graph is sparse and $\dist(s, t, G)\leq L$. More specifically, let $G = (V, E)$ be an unweighted digraph with $n$ vertices and $m$ edges, and consider a pair of vertices $s, t$ with a shortest $st$-path $E(\pi)$ of length at most $L$. The task is to compute for any pairs of edges $f_1, f_2$ the value of $\dist(s, t, G\setminus \{f_1, f_2\})$, where $f_1\in E(\pi), f_2\notin E(\pi)$. The purpose of this section is to prove the following lemma.

\begin{lemma}\label{1fail-special}
	Let $G = (V, E)$ be an unweighted digraph with $n$ vertices and $m$ edges. Fix a pair of source and terminal $s, t\in V$ such that $\dist(s, t, G)\leq L$. Then, all values of $\dist(s, t, G\setminus \{f_1, f_2\})$ can be computed with high probability in time:
	$$\tilde{O}\left(\frac{mn^{1.5} + n^3}{L} + mn^{1/2}L/g + n^2L/g + mL^2/g + mnL^2/g^3 + mLg + L^2g^4 + n^2L^2/g^2\right)$$
	when $f_1\in E(\pi)$ while $f_2\notin E(\pi)$; here $g$ is an arbitrary parameter to be determined later.
\end{lemma}

Let $\pi = \langle s = u_0, u_1, \cdots, u_{|\pi|} =t\rangle$.  First, we use \Cref{rp} to compute for each $(u_i, u_{i+1})\in E(\pi)$ the replacement path from $s$ to $t$ that avoids $e_i = (u_i, u_{i+1})$, and let $\alpha_i$ be the corresponding detour avoiding $e_i$ which starts at $a_i$ and ends at $b_i$.

For each detour $\alpha_i$, divide $\alpha_i$ into sub-paths of length $g$ (except for the last sub-path) and list them as $\beta^i_1, \cdots, \beta^i_{l_i}$, and assume $f\in E(\beta^i_l)$; later on, when we refer to $\beta^i_l$,  if $l\leq 0$ or $l>l_i$, then $\beta^i_l$ would simply be an empty path.

Throughout the algorithm, for each $0\leq i<|\pi|$ and every edge $f\in E(\alpha_i)$, we will maintain a distance estimation $\est(s, t, G\setminus \{e_i, f\})\geq \dist(s, t, G\setminus \{e_i, f\})$, and in the end it will be guaranteed that $\est(s, t, G\setminus \{e_i, f\}) = \dist(s, t, G\setminus \{e_i, f\})$.

Let $\rho_{i, f}$ be a canonical shortest replacement path for $\{e_i, f\}$. Suppose $\rho_{i, f}$ departs from $\pi[s, a_i]\circ \alpha_i\circ \pi[b_i, t]$ at vertex $x_{i, f}$ and converges with $\pi[s, a_i]\circ \alpha_i\circ \pi[b_i, t]$ at vertex $y_{i, f}$. For the rest, we address different cases depending on the properties regarding $\alpha_i, \rho_{i, f}$; note that our algorithm does not need to know which case it is in advance.

\subsection*{Case 1: $x_{i, f}\in V(\pi[s, a_i])$ or $y_{i, f}\in V(\pi[b_i, t])$}
This is an easy case of our algorithm, and the runtime of this part can be bounded by $O(mL)$. See \Cref{1fail-short-case1} for an illustration. 

\paragraph{Algorithm.} Without loss of generality, assume $x_{i, f}\in V(\pi[s, a_i])$; if $y_{i, f}\in V(\pi[b_i, t])$, then we can handle it symmetrically. Compute single-source shortest paths from $s$ in $G\setminus \left(E(\alpha_i)\cup \{e_i\}\right)$. Then, for each $f\in E(\alpha_i)$, update the distance estimation $\est(s, t, G\setminus \{e_i, f\})$ as the minimum between $\dist\left(s, t, G\setminus \left(E(\alpha_i)\cup \{e_i\}\right)\right)$ and the following term:
$$\min_{v\in V(\alpha_i(f, b_i])}\{\dist\left(s, v, G\setminus \left(E(\alpha_i)\cup \{e_i\}\right)\right) + |\alpha_i[v, b_i]| + |\pi[b_i, t]|\}$$

\paragraph{Runtime.} The runtime of the single-source shortest paths computation is bounded by $O(mL)$ summing over all $1\leq i\leq L$. To implement the updates to our distance estimates $\est(s, t, G\setminus \{e_i, f\})$, we can first compute all values:
$$\dist\left(s, v, G\setminus \left(E(\alpha_i)\cup \{e_i\}\right)\right) + |\alpha_i[v, b_i]| + |\pi[b_i, t]|$$
for each $v\in V(\alpha_i)$, and then calculate the minimum clause for each $f\in E(\alpha_i)$ in $O(\log n)$ time using standard interval data structures. The overall runtime is bounded by $O(mL)$.

\paragraph{Correctness.} Since $\rho_{i, f}$ diverges and converges with $\pi[s, a_i]\circ \alpha_i \circ \pi[b_i, t]$ for only once, the prefix-path $\rho_{i, f}[s, y_{i, f}]$ does not use any edge on the sub-path $\alpha_i$. If $\rho_{i, f}$ converges with $\alpha_i$, then this case is covered by the calculation of the minimum clause:
$$\min_{v\in V(\alpha_i(f, b_i])}\{\dist\left(s, v, G\setminus \left(E(\alpha_i)\cup \{e_i\}\right)\right) + |\alpha_i[v, b_i]| + |\pi[b_i, t]|\}$$
Otherwise, it must be $\dist\left(s, t, G\setminus \left(E(\alpha_i)\cup \{e_i\}\right)\right)$.

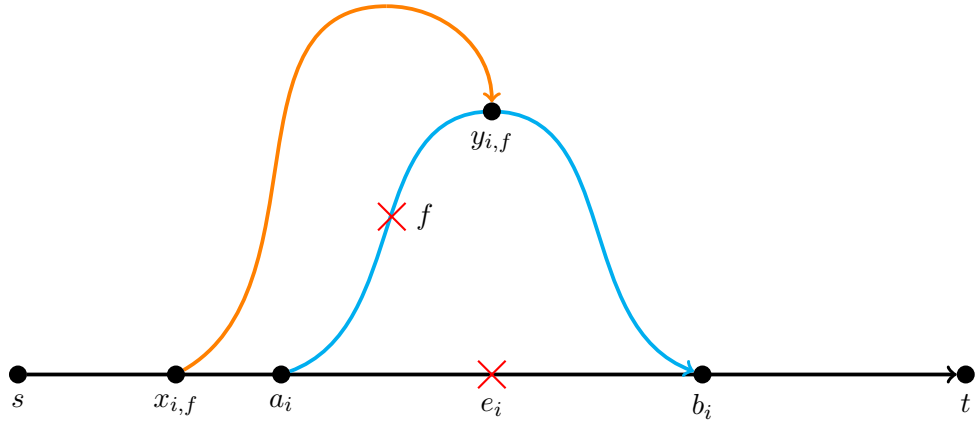
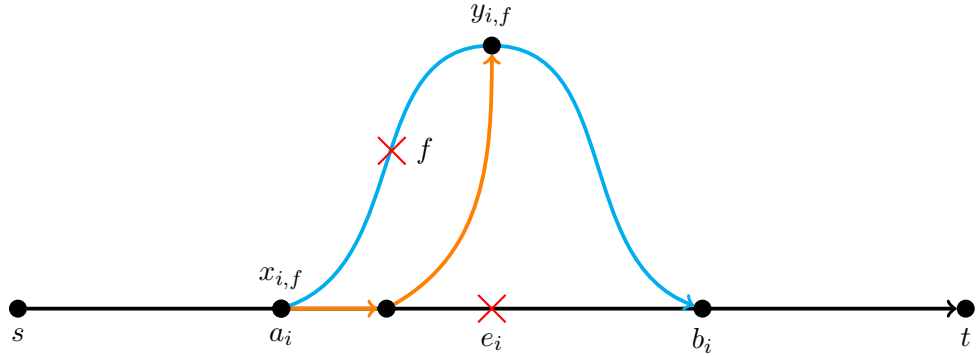
\begin{figure}
	\centering
	\subfigure[In this case, ${x_{i, f}\neq a_i}$ and ${\rho_{i, f}[x_{i, f}, *]}$ diverges from $\pi$ immediately.]{\begin{tikzpicture}[thick,scale=0.7]
	\draw (-9, 0) node(1)[circle, draw, fill=black, inner sep=0pt, minimum width=6pt, label=-90: {$s$}] {};
	\draw (9, 0) node(2)[circle, draw, fill=black, inner sep=0pt, minimum width=6pt, label=-90: {$t$}] {};
	
	\draw (-4, 0) node(3)[circle, draw, fill=black, inner sep=0pt, minimum width=6pt, label=-90: {$a_i$}] {};
	\draw (4, 0) node(4)[circle, draw, fill=black, inner sep=0pt, minimum width=6pt, label=-90: {$b_i$}] {};
	
	\draw (-6, 0) node(5)[circle, draw, fill=black, inner sep=0pt, minimum width=6pt, label=-90: {$x_{i, f}$}] {};
	\draw (0, 5) node(6)[circle, draw, fill=black, inner sep=0pt, minimum width=6pt, label=-90: {$y_{i, f}$}] {};
	
	\draw (0, 0) node[cross=6, red, label=-90: {$e_i$}] {};
	\draw (-1.9, 3) node[cross=6, red, label=0: {$f$}] {};
		
	\begin{scope}[on background layer]
		\draw [->, line width = 0.5mm] (1) to (2);
		\draw [line width = 0.5mm, color = cyan] (3) to[out=20, in=180] (0, 5);
		\draw [->, line width = 0.5mm, color = cyan] (0, 5) to[out=0, in=160] (4);
		\draw [line width = 0.5mm, color = orange] (5) to[out=30, in=180] (-2, 7);
		\draw [->, line width = 0.5mm, color = orange] (-2, 7) to[out=0, in=90] (6);
	\end{scope}
\end{tikzpicture}}
	\vspace{0.3cm}
	\subfigure[In this case, ${x_{i, f} = a_i}$ and ${\rho_{i, f}[x_{i, f}, *]}$ uses some edges on ${\pi[a_i, b_i]}$ before it diverges.]{\begin{tikzpicture}[thick,scale=0.7]
	\draw (-9, 0) node(1)[circle, draw, fill=black, inner sep=0pt, minimum width=6pt, label=-90: {$s$}] {};
	\draw (9, 0) node(2)[circle, draw, fill=black, inner sep=0pt, minimum width=6pt, label=-90: {$t$}] {};
	
	\draw (-4, 0) node(3)[circle, draw, fill=black, inner sep=0pt, minimum width=6pt, label=-90: {$a_i$}] {};
	\draw (4, 0) node(4)[circle, draw, fill=black, inner sep=0pt, minimum width=6pt, label=-90: {$b_i$}] {};
	
	\draw (-4, 0) node(5)[circle, draw, fill=black, inner sep=0pt, minimum width=6pt, label=90: {$x_{i, f}$}] {};
	\draw (0, 5) node(6)[circle, draw, fill=black, inner sep=0pt, minimum width=6pt, label=90: {$y_{i, f}$}] {};
	\draw (-2, 0) node(7)[circle, draw, fill=black, inner sep=0pt, minimum width=6pt] {};
	
	\draw (0, 0) node[cross=6, red, label=-90: {$e_i$}] {};
	\draw (-1.9, 3) node[cross=6, red, label=0: {$f$}] {};
		
	\begin{scope}[on background layer]
		\draw [->, line width = 0.5mm] (1) to (2);
		\draw [line width = 0.5mm, color = cyan] (3) to[out=20, in=180] (0, 5);
		\draw [->, line width = 0.5mm, color = cyan] (0, 5) to[out=0, in=160] (4);
		\draw [->, line width = 0.5mm, color = orange] (7) to[out=30, in=-90] (6);
		\draw [->, line width = 0.5mm, color = orange] (5) to (7);
	\end{scope}
\end{tikzpicture}}
	\caption{In this picture, the cyan path is detour $\alpha_i$, and the orange path is $\rho_{i, f}[x_{i, f}, y_{i, f}]$.}
	\label{1fail-short-case1}
\end{figure}

\subsection*{Case 2: $x_{i, f}, y_{i, f}\in V(\alpha_i)$ and $|\alpha_i|\geq L$}
This is an easy case of our algorithm, and the runtime of this part can be bounded by $\tilde{O}(\frac{mn^{1.5} + n^3}{L})$. See \Cref{1fail-short-case2} for an illustration. \paragraph{Algorithm.} Take a uniformly random vertex subset $P$ of size $O(\frac{n\log n}{L})$. For each $p\in P$, apply \Cref{ssrp} to and from $p$ in graph $G\setminus E(\pi)$. Then, for each index $1\leq i\leq l$ and $f\in E(\alpha_i)$, update $\est(s, t, G\setminus \{e_i, f\})$ as the value of
$$\min_{p\in P}\{\dist(s, a_i, G) + \dist(a_i, p, G\setminus (E(\pi)\cup \{f\})) + \dist(p, b_i, G\setminus (E(\pi)\cup \{f\})) + \dist(b_i, t, G)\}$$

According to \Cref{ssrp}, this procedure takes time $\tilde{O}(\frac{mn^{1.5} + n^3}{L})$.

\paragraph{Correctness.} Since $x_{i, f}, y_{i, f}\in \alpha_i$, we know that $a_i, b_i\in V(\rho_{i, f})$. Consider the sub-path $\rho_{i, f}[a_i, b_i]$. As $|\alpha_i|\geq L$, we know that $|\rho_{i, f}[a_i, b_i]|\geq L$. Since $P$ is a random vertex subset of size $O(\frac{n\log n}{L})$, $P\cap V(\rho_{i, f}[a_i, b_i])\neq \emptyset$. Take any vertex $p\in P\cap V(\rho_{i, f}[a_i, b_i])$, we know that:
$$|\rho_{i, f}[a_i, b_i]| = \dist(a_i, p, G\setminus (E(\pi)\cup \{f\})) + \dist(p, b_i, G\setminus (E(\pi)\cup \{f\}))$$

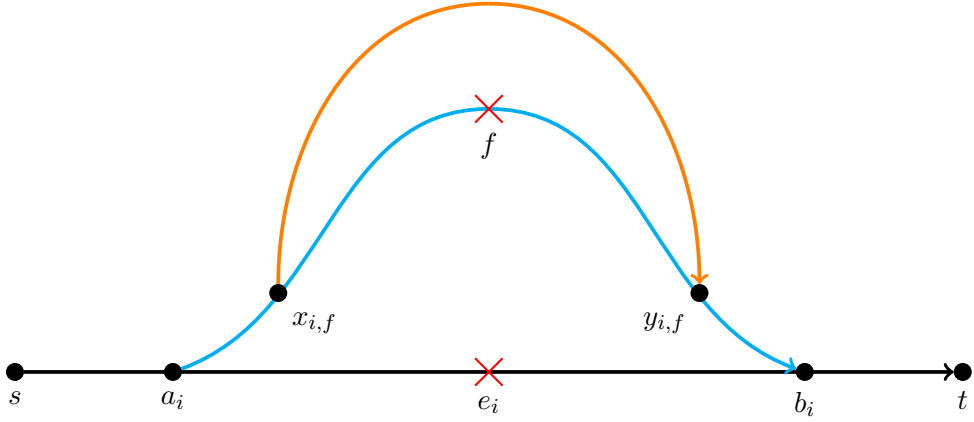
\begin{figure}
	\centering
	\begin{tikzpicture}[thick,scale=0.7]
	\draw (-9, 0) node(1)[circle, draw, fill=black, inner sep=0pt, minimum width=6pt, label=-90: {$s$}] {};
	\draw (9, 0) node(2)[circle, draw, fill=black, inner sep=0pt, minimum width=6pt, label=-90: {$t$}] {};
	
	\draw (-6, 0) node(3)[circle, draw, fill=black, inner sep=0pt, minimum width=6pt, label=-90: {$a_i$}] {};
	\draw (6, 0) node(4)[circle, draw, fill=black, inner sep=0pt, minimum width=6pt, label=-90: {$b_i$}] {};
	
	\draw (-4, 1.5) node(5)[circle, draw, fill=black, inner sep=0pt, minimum width=6pt, label=-70: {$x_{i, f}$}] {};
	\draw (4, 1.5) node(6)[circle, draw, fill=black, inner sep=0pt, minimum width=6pt, label=-110: {$y_{i, f}$}] {};

	
	\draw (0, 0) node[cross=6, red, label=-90: {$e_i$}] {};
	\draw (0, 5) node[cross=6, red, label=-90: {$f$}] {};

	\begin{scope}[on background layer]
		\draw [->, line width = 0.5mm] (1) to (2);
		\draw [line width = 0.5mm, color = cyan] (3) to[out=20, in=180] (0, 5);
		\draw [->, line width = 0.5mm, color = cyan] (0, 5) to[out=0, in=160] (4);
		\draw [line width = 0.5mm, color = orange] (5) to[out=90, in=180] (0, 7);
		\draw [->, line width = 0.5mm, color = orange] (0, 7) to[out=0, in=90] (6);
	\end{scope}
\end{tikzpicture}
	\caption{The sub-path $\rho_{i, f}[a_i, b_i]$ has length greater than $L$.}
	\label{1fail-short-case2}
\end{figure}

\subsection*{Case 3: $x_{i, f}, y_{i, f}\in V(\alpha_i)$, plus that $|\pi[a_i, u_i]| > g$ or $|\pi[u_{i+1}, b_i]| > g$}
This is an easy case of our algorithm, and the runtime of this part can be bounded by $$\tilde{O}\brac{m\sqrt{n}L/g + n^2L/g + L^2}$$ See \Cref{1fail-short-case3} for an illustration.

\paragraph{Algorithm.} Without loss of generality, assume that $|\pi[a_i, u_i]| > g$. Then, there exists an integer $k$ and vertex $v\in \pi[a_i, u_i]$ such that $|\pi[s, v]| = k\cdot g$. Apply \Cref{ssrp} on source $s$ in graph $G\setminus E(\pi[v, t])$. After that, update the estimation value $\est(s, t, G\setminus \{e_i, f\})$ as
$$\est(s, t, G\setminus \{e_i, f\})\leftarrow\min\{\dist(s, b_i, G\setminus \left(E(\pi[v, t])\cup \{f\}\right)) + |\pi[b_i, t]|, \est(s, t, G\setminus \{e_i, f\})\}$$

\paragraph{Runtime.} The runtime of all instances of single-source replacement paths is bounded by: $$\tilde{O}(m\sqrt{n}L/g + n^2L/g)$$
Calculating the values of $\est(s, t, G\setminus \{e_i, f\})$ can be done in total time $\tilde{O}(L^2)$.

\paragraph{Correctness.} Since $a_i$ lies between $s$ and $v$ on $E(\pi)$, plus that $x_{i, f}\in V(\alpha_i)$, we know that the prefix path $\rho_{i, f}[s, b_i]$ does not use any edges on sub-path $\pi[v, t]$. Therefore, the formula calculates the correct length of $\rho_{i, f}$.

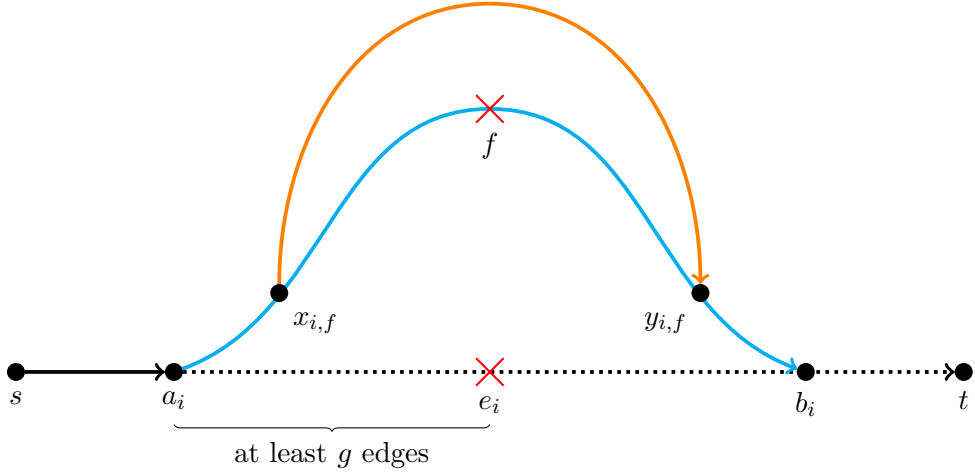
\begin{figure}
	\centering
	\begin{tikzpicture}[thick,scale=0.7]
	\draw (-9, 0) node(1)[circle, draw, fill=black, inner sep=0pt, minimum width=6pt, label=-90: {$s$}] {};
	\draw (9, 0) node(2)[circle, draw, fill=black, inner sep=0pt, minimum width=6pt, label=-90: {$t$}] {};
	
	\draw (-6, 0) node(3)[circle, draw, fill=black, inner sep=0pt, minimum width=6pt, label=-90: {$a_i$}] {};
	\draw (6, 0) node(4)[circle, draw, fill=black, inner sep=0pt, minimum width=6pt, label=-90: {$b_i$}] {};
	
	\draw (-4, 1.5) node(5)[circle, draw, fill=black, inner sep=0pt, minimum width=6pt, label=-70: {$x_{i, f}$}] {};
	\draw (4, 1.5) node(6)[circle, draw, fill=black, inner sep=0pt, minimum width=6pt, label=-110: {$y_{i, f}$}] {};
	
	
	\draw (0, 0) node[cross=6, red, label=-90: {$e_i$}] {};
	\draw (0, 5) node[cross=6, red, label=-90: {$f$}] {};

	\begin{scope}[on background layer]
		\draw [->, line width = 0.5mm] (1) to (3);
		\draw [->, line width = 0.5mm, dotted] (3) to (2);
		\draw [line width = 0.5mm, color = cyan] (3) to[out=20, in=180] (0, 5);
		\draw [->, line width = 0.5mm, color = cyan] (0, 5) to[out=0, in=160] (4);
		\draw [line width = 0.5mm, color = orange] (5) to[out=90, in=180] (0, 7);
		\draw [->, line width = 0.5mm, color = orange] (0, 7) to[out=0, in=90] (6);
		
		\draw [decorate,
		decoration = {brace}] (0,-1) -- (-6,-1);
		\draw (-3, -2.2) node[label={at least $g$ edges}]{};
	\end{scope}
\end{tikzpicture}
	\caption{In this picture, $x_{i, f}, y_{i, f}\in V(\alpha_i)$ plus that $|\pi[a_i, u_i]| > g$.}
	\label{1fail-short-case3}
\end{figure}

\subsection*{Case 4: $x_{i, f}, y_{i, f}\in V(\alpha_i)$, plus that $x_{i, f}\notin V(\beta^i_l)$ or $y_{i, f}\notin V(\beta^i_l)$, and $|\alpha_i| < L$}

This is an easy case of our algorithm, and the runtime of this part can be bounded by $\tilde{O}(mL^2/g)$. See \Cref{1fail-short-case4} for an illustration.

\paragraph{Algorithm.} Without loss of generality, assume $x_{i, f}\notin V(\beta^i_l)$. Compute single-source shortest paths from $s$ in graph $G\setminus \left(E(\beta^i_l)\cup\{e_i\}\right)$. Then, for each $f\in E(\beta^i_l)$, update the estimation $\est(s, t, G\setminus \{e_i, f\})$ with the minimum clause:
$$\min_{v\in V(\alpha_i(f, b_i])}\left\{\dist\left(s, v, G\setminus \left(E(\beta^i_l)\cup\{e_i\}\right)\right) +  |\alpha_i[v, b_i]| + \dist(b_i, t, G)\right\}$$

\paragraph{Runtime.} The runtime of all instances of single-source shortest paths from $s$ takes time $O(mL^2/g)$. To evaluate all the values $\est(s, t, G\setminus \{e_i, f\})$, we can preprocess each sub-path $\beta^i_j$ in $\tilde{O}(g)$ time so that each value of the minimum can be evaluated in $O(\log L)$ time. Hence, the total runtime of this piece of algorithm is $O(mL^2/g + L^2\log L) = \tilde{O}(mL^2/g)$.

\paragraph{Correctness.} Since $x_{i, f}\notin \beta^i_j$, the shortest path $\rho_{i, f}[s, y_{i, f}]$ must have skipped the entire sub-path $\beta^i_j$. Therefore, when we pick the right guess $v = y_{i, f}$, we have:
$$|\rho_{i, f}| = \dist\left(s, v, G\setminus \left(E(\beta^i_j)\cup\{e_i\}\right)\right) +  |\alpha_i[v, b_i]| + \dist(b_i, t, G)$$

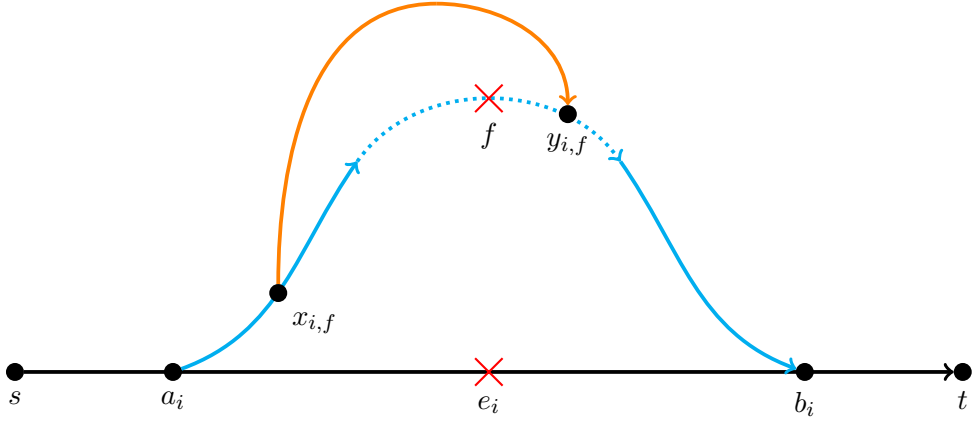
\begin{figure}
	\centering
	\begin{tikzpicture}[thick,scale=0.7]
	\draw (-9, 0) node(1)[circle, draw, fill=black, inner sep=0pt, minimum width=6pt, label=-90: {$s$}] {};
	\draw (9, 0) node(2)[circle, draw, fill=black, inner sep=0pt, minimum width=6pt, label=-90: {$t$}] {};
	
	\draw (-6, 0) node(3)[circle, draw, fill=black, inner sep=0pt, minimum width=6pt, label=-90: {$a_i$}] {};
	\draw (6, 0) node(4)[circle, draw, fill=black, inner sep=0pt, minimum width=6pt, label=-90: {$b_i$}] {};
	
	\draw (-4, 1.5) node(5)[circle, draw, fill=black, inner sep=0pt, minimum width=6pt, label=-70: {$x_{i, f}$}] {};
	\draw (1.5, 4.9) node(6)[circle, draw, fill=black, inner sep=0pt, minimum width=6pt, label=-90: {$y_{i, f}$}] {};
		
	\draw (0, 0) node[cross=6, red, label=-90: {$e_i$}] {};
	\draw (0, 5.2) node[cross=6, red, label=-90: {$f$}] {};
	
	\begin{scope}[on background layer]
		\draw [->, line width = 0.5mm] (1) to (2);
		\draw [->, line width = 0.5mm, color = cyan] (3) to[out=20, in=235] (-2.5, 4);
		\draw [->, line width = 0.5mm, color = cyan] (2.5, 4) to[out=-55, in=160] (4);
		\draw [->, line width = 0.5mm, color = cyan, dotted] (-2.5, 4) to[out=55, in=125] (2.5, 4);
		\draw [line width = 0.5mm, color = orange] (5) to[out=90, in=180] (-1, 7);
		\draw [->, line width = 0.5mm, color = orange] (-1, 7) to[out=0, in=90] (6);
	\end{scope}
\end{tikzpicture}
	\caption{The dotted cyan path is $\beta^i_l$.}
	\label{1fail-short-case4}
\end{figure}

\subsection*{Main case: $|\pi[a_i, u_i]|, |\pi[u_{i+1}, b_i]| \leq g$, and $x_{i, f}, y_{i, f}\in V(\beta^i_l)$, and $|\alpha_i| < L$}
\paragraph{Algorithm.} Let $I\subseteq [L]$ be the set of all indices such that $|\pi[a_i, u_i]|, |\pi[u_{i+1}, b_i]| \leq g$ and $|\alpha_i|< L$. For each pair of indices $(p, q)\in [L/g]\times[L/g]$, define the set of sub-paths:
$$\mathcal{P}_{p, q} = \{\beta^i_l\mid(l, l_i - l+1) =  (p, q), i\in I\}$$

\begin{enumerate}[(1),leftmargin=*]
	\item Let $U\subseteq V$ be the uniformly random subset of size $\frac{10n\log n}{g}$. Then, for each pair of indices $p, q\in [L/g]$, define the graph:
	$$G_{p, q} = G\setminus \left(E(\pi)\cup\bigcup_{\beta \in \mathcal{P}_{p, q}}E(\beta)\right)$$
	Then, for each vertex $u\in U$, compute single-source shortest paths to and from $u$ in $G_{p, q}$.

	\item For any index $1\leq l\leq \ceil{L/g}$, and for any offset $1\leq b \leq 10g$, initialize two sets of vertices $A_{b, l}, B_{b, l}\leftarrow V$. Then, go over the sequence of all sub-paths $\beta^{b}_l, \beta^{b+10g}_l, \cdots, \beta^{b+10hg}_l$, where $h \leq \ceil{\frac{L}{10g}}$.
	
	Next, for each sub-path $\beta^{b+10kg}_l$, define the following two graphs
	$$X^{b+10kg}_l = G\setminus \left(E(\pi)\cup \bigcup_{j\in [b+10kg, b+(10k+5)g]\cap I}E(\beta^{j}_l)\right)$$
	$$Y^{b+10kg}_l = G\setminus \left(E(\pi)\cup \bigcup_{j\in [b+(10k-5)g, b+10kg]\cap I}^{}E(\beta^{j}_l)\right)$$
	See \Cref{1fail-short-case5-2} for an illustration.
	
	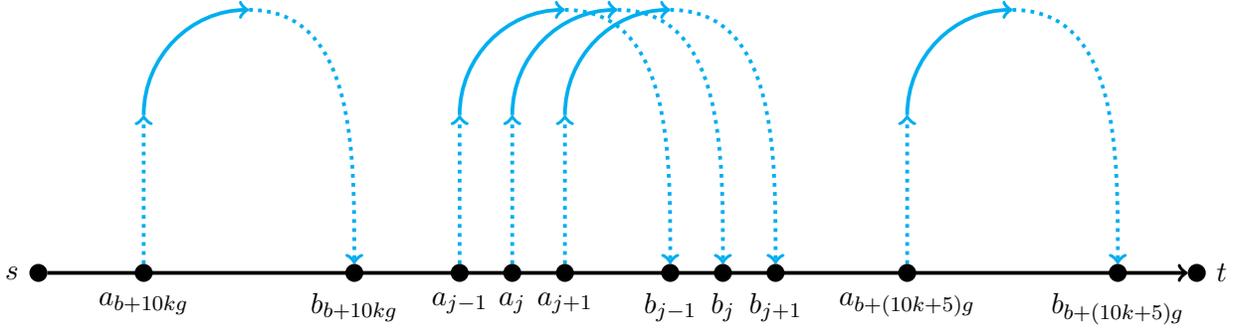
\begin{figure}
		\centering
		\begin{tikzpicture}[thick,scale=0.7]
	\draw (-11, 0) node(1)[circle, draw, fill=black, inner sep=0pt, minimum width=6pt, label=180: {$s$}] {};
	\draw (11, 0) node(2)[circle, draw, fill=black, inner sep=0pt, minimum width=6pt, label=0: {$t$}] {};
	
	\draw (-2, 0) node(3)[circle, draw, fill=black, inner sep=0pt, minimum width=6pt, label=-90: {$a_j$}] {};
	\draw (2, 0) node(4)[circle, draw, fill=black, inner sep=0pt, minimum width=6pt, label=-90: {$b_j$}] {};
		
	\draw (-3, 0) node(5)[circle, draw, fill=black, inner sep=0pt, minimum width=6pt, label=-90: {$a_{j-1}$}] {};
	\draw (1, 0) node(6)[circle, draw, fill=black, inner sep=0pt, minimum width=6pt, label=-90: {$b_{j-1}$}] {};
	
	\draw (-1, 0) node(7)[circle, draw, fill=black, inner sep=0pt, minimum width=6pt, label=-90: {$a_{j+1}$}] {};
	\draw (3, 0) node(8)[circle, draw, fill=black, inner sep=0pt, minimum width=6pt, label=-90: {$b_{j+1}$}] {};
	
	\draw (-9, 0) node(9)[circle, draw, fill=black, inner sep=0pt, minimum width=6pt, label=-90: {$a_{b+10kg}$}] {};
	\draw (-5, 0) node(10)[circle, draw, fill=black, inner sep=0pt, minimum width=6pt, label=-90: {$b_{b+10kg}$}] {};
	
	\draw (5.5, 0) node(11)[circle, draw, fill=black, inner sep=0pt, minimum width=6pt, label=-90: {$a_{b+(10k+5)g}$}] {};
	\draw (9.5, 0) node(12)[circle, draw, fill=black, inner sep=0pt, minimum width=6pt, label=-90: {$b_{b+(10k+5)g}$}] {};
	
	
	
	
	\begin{scope}[on background layer]
		\draw [->, line width = 0.5mm] (1) to (2);
		
		\draw [->, line width = 0.5mm, color = cyan, dotted] (3) to[out=90, in=-90] (-2, 3);
		\draw [->, line width = 0.5mm, color = cyan] (-2, 3) to[out=90, in=180] (0, 5);
		\draw [->, line width = 0.5mm, color = cyan, dotted] (0, 5) to[out=0, in=90] (4);
		
		\draw [->, line width = 0.5mm, color = cyan, dotted] (5) to[out=90, in=-90] (-3, 3);
		\draw [->, line width = 0.5mm, color = cyan] (-3, 3) to[out=90, in=180] (-1, 5);
		\draw [->, line width = 0.5mm, color = cyan, dotted] (-1, 5) to[out=0, in=90] (6);

		\draw [->, line width = 0.5mm, color = cyan, dotted] (7) to[out=90, in=-90] (-1, 3);
		\draw [->, line width = 0.5mm, color = cyan] (-1, 3) to[out=90, in=180] (1, 5);
		\draw [->, line width = 0.5mm, color = cyan, dotted] (1, 5) to[out=0, in=90] (8);
		
		\draw [->, line width = 0.5mm, color = cyan, dotted] (9) to[out=90, in=-90] (-9, 3);
		\draw [->, line width = 0.5mm, color = cyan] (-9, 3) to[out=90, in=180] (-7, 5);
		\draw [->, line width = 0.5mm, color = cyan, dotted] (-7, 5) to[out=0, in=90] (10);
		
		\draw [->, line width = 0.5mm, color = cyan, dotted] (11) to[out=90, in=-90] (5.5, 3);
		\draw [->, line width = 0.5mm, color = cyan] (5.5, 3) to[out=90, in=180] (7.5, 5);
		\draw [->, line width = 0.5mm, color = cyan, dotted] (7.5, 5) to[out=0, in=90] (12);
		
	\end{scope}
\end{tikzpicture}
		\caption{Graph $X_{l}^{b+10kg}$ excludes all edges in $E(\pi)\cup \bigcup_{j=b+10kg, j\in I}^{b+(10k+5)g}E(\beta^{j}_l)$ which are drawn and black and cyan solid paths.}
		\label{1fail-short-case5-2}
	\end{figure}
	
	\begin{enumerate}[(a),leftmargin=*]
		\item Then, for each vertex $v\in V\brac{\beta^{b+10kg}_l}$, perform a truncated Dijkstra at $v$ in the induced subgraph $X^{b+10kg}_l[A_{b, l}]$ up to depth $g$. For each $z\in A_{b, l}$, record the distance value
		$$\mu_X(v, z)\leftarrow \dist\brac{v, z, X^{b+10kg}_l[A_{b, l}]}$$
		if $\dist\brac{v, z, X^{b+10kg}_l[A_{b, l}]}\leq g$.
		
		After we have visited all vertices $v\in V\brac{\beta^{b+10kg}_l}$, let $P^k_{b, l}\subseteq A_{b, l}$ be the set of all vertices searched by truncated Dijkstra of any $v\in V\brac{\beta^{b+10kg}_l}$. Then, prune the vertex set
		$$A_{b, l}\leftarrow A_{b, l}\setminus P^k_{b, l}$$
		
		\item Symmetrically, for each vertex $v\in V\brac{\beta^{b+10kg}_l}$, perform a truncated Dijkstra at $v$ in the induced subgraph $Y^{b+10kg}_l[B_{b, l}]$ up to depth $g$. For each $z\in B_{b, l}$, record the distance value
		$$\mu_Y(v, z)\leftarrow \dist\brac{v, z, Y^{b+10kg}_l[B_{b, l}]}$$
		if $\dist\brac{v, z, Y^{b+10kg}_l[B_{b, l}]}\leq g$.
		
		After we have visited all vertices $v\in V\brac{\beta^{b+10kg}_l}$, let $Q^k_{b, l}\subseteq B_{b, l}$ be the set of all vertices searched by truncated Dijkstra of any $v\in V\brac{\beta^{b+10kg}_l}$. Then, prune the vertex set
		$$B_{b, l}\leftarrow B_{b, l}\setminus Q^k_{b, l}$$
	\end{enumerate}
	
	\item For any index $i\in I$ and index $1\leq l\leq l_i/g$, let us build a shortcut digraph $H^i_l$ with edge weight function $\wts$ in the following manner. See \Cref{1fail-short-case5-3} for an illustration.
	
	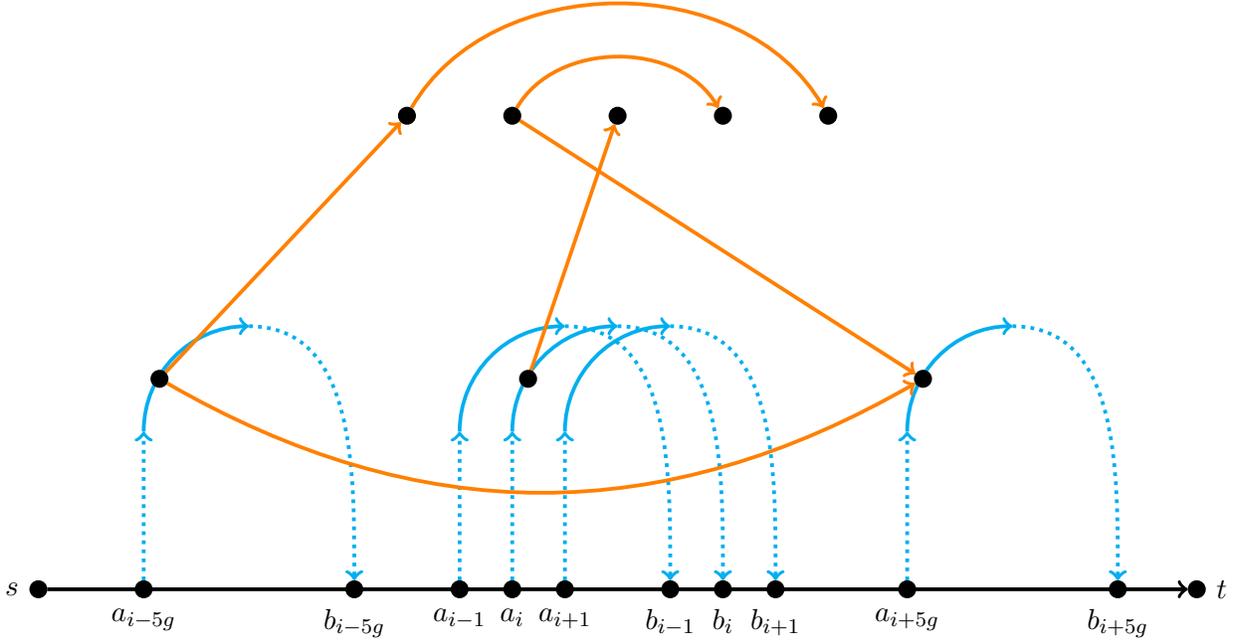
\begin{figure}
		\centering
		\begin{tikzpicture}[thick,scale=0.7]
	\draw (-11, 0) node(1)[circle, draw, fill=black, inner sep=0pt, minimum width=6pt, label=180: {$s$}] {};
	\draw (11, 0) node(2)[circle, draw, fill=black, inner sep=0pt, minimum width=6pt, label=0: {$t$}] {};
	
	\draw (-2, 0) node(3)[circle, draw, fill=black, inner sep=0pt, minimum width=6pt, label=-90: {$a_i$}] {};
	\draw (2, 0) node(4)[circle, draw, fill=black, inner sep=0pt, minimum width=6pt, label=-90: {$b_i$}] {};
	
	\draw (-3, 0) node(5)[circle, draw, fill=black, inner sep=0pt, minimum width=6pt, label=-90: {$a_{i-1}$}] {};
	\draw (1, 0) node(6)[circle, draw, fill=black, inner sep=0pt, minimum width=6pt, label=-90: {$b_{i-1}$}] {};
	
	\draw (-1, 0) node(7)[circle, draw, fill=black, inner sep=0pt, minimum width=6pt, label=-90: {$a_{i+1}$}] {};
	\draw (3, 0) node(8)[circle, draw, fill=black, inner sep=0pt, minimum width=6pt, label=-90: {$b_{i+1}$}] {};
	
	\draw (-9, 0) node(9)[circle, draw, fill=black, inner sep=0pt, minimum width=6pt, label=-90: {$a_{i-5g}$}] {};
	\draw (-5, 0) node(10)[circle, draw, fill=black, inner sep=0pt, minimum width=6pt, label=-90: {$b_{i-5g}$}] {};
	
	\draw (5.5, 0) node(11)[circle, draw, fill=black, inner sep=0pt, minimum width=6pt, label=-90: {$a_{i+5g}$}] {};
	\draw (9.5, 0) node(12)[circle, draw, fill=black, inner sep=0pt, minimum width=6pt, label=-90: {$b_{i+5g}$}] {};

	\draw (-4, 9) node(13)[circle, draw, fill=black, inner sep=0pt, minimum width=6pt] {};	
	\draw (-2, 9) node(14)[circle, draw, fill=black, inner sep=0pt, minimum width=6pt] {};
	\draw (0, 9) node(15)[circle, draw, fill=black, inner sep=0pt, minimum width=6pt] {};
	\draw (2, 9) node(16)[circle, draw, fill=black, inner sep=0pt, minimum width=6pt] {};
	\draw (4, 9) node(17)[circle, draw, fill=black, inner sep=0pt, minimum width=6pt] {};

	\draw (-8.7, 4) node(18)[circle, draw, fill=black, inner sep=0pt, minimum width=6pt] {};	
	\draw (-1.7, 4) node(19)[circle, draw, fill=black, inner sep=0pt, minimum width=6pt] {};
	\draw (5.8, 4) node(20)[circle, draw, fill=black, inner sep=0pt, minimum width=6pt] {};
	
	
	
	
	\begin{scope}[on background layer]
		\draw [->, line width = 0.5mm] (1) to (2);
		
		\draw [->, line width = 0.5mm, color = cyan, dotted] (3) to[out=90, in=-90] (-2, 3);
		\draw [->, line width = 0.5mm, color = cyan] (-2, 3) to[out=90, in=180] (0, 5);
		\draw [->, line width = 0.5mm, color = cyan, dotted] (0, 5) to[out=0, in=90] (4);
		
		\draw [->, line width = 0.5mm, color = cyan, dotted] (5) to[out=90, in=-90] (-3, 3);
		\draw [->, line width = 0.5mm, color = cyan] (-3, 3) to[out=90, in=180] (-1, 5);
		\draw [->, line width = 0.5mm, color = cyan, dotted] (-1, 5) to[out=0, in=90] (6);
		
		\draw [->, line width = 0.5mm, color = cyan, dotted] (7) to[out=90, in=-90] (-1, 3);
		\draw [->, line width = 0.5mm, color = cyan] (-1, 3) to[out=90, in=180] (1, 5);
		\draw [->, line width = 0.5mm, color = cyan, dotted] (1, 5) to[out=0, in=90] (8);
		
		\draw [->, line width = 0.5mm, color = cyan, dotted] (9) to[out=90, in=-90] (-9, 3);
		\draw [->, line width = 0.5mm, color = cyan] (-9, 3) to[out=90, in=180] (-7, 5);
		\draw [->, line width = 0.5mm, color = cyan, dotted] (-7, 5) to[out=0, in=90] (10);
		
		\draw [->, line width = 0.5mm, color = cyan, dotted] (11) to[out=90, in=-90] (5.5, 3);
		\draw [->, line width = 0.5mm, color = cyan] (5.5, 3) to[out=90, in=180] (7.5, 5);
		\draw [->, line width = 0.5mm, color = cyan, dotted] (7.5, 5) to[out=0, in=90] (12);
		
		\draw [->, line width = 0.5mm, color = orange] (18) to (13);
		\draw [->, line width = 0.5mm, color = orange] (14) to (20);
		\draw [->, line width = 0.5mm, color = orange] (19) to (15);
		\draw [->, line width = 0.5mm, color = orange] (18) to[out=-30, in=210] (20);
		\draw [->, line width = 0.5mm, color = orange] (14) to[out=60, in=120] (16);
		\draw [->, line width = 0.5mm, color = orange] (13) to[out=60, in=120] (17);
		
	\end{scope}
\end{tikzpicture}
		\caption{Vertices on the topmost are in $U$, and orange edges represent shortcut edges in $H^i_l$.}
		\label{1fail-short-case5-3}
	\end{figure}

	\begin{itemize}[leftmargin=*]
		\item \textbf{Vertices.} Add all vertices in sub-paths $V(\beta^j_l), \forall j\in [i-5g, i+5g]\cap I$, as well as all pivot vertices in $U$ to $H^i_l$.
		
		\item \textbf{Edges.} First, add to $E(H^i_l)$ all the sub-paths:
		$$\left(\bigcup_{j\in[i-5g, i+5g]\cap I}^{}E\brac{\beta^j_l}\right)\setminus E\brac{\beta^i_l}$$ Then, add the following three types of weighted edges.
		\begin{enumerate}[(a),leftmargin=*]
			\item For any $u\in U$ and any vertex $v\in V(H^i_l)$, add an edge $(u, v)$ with edge weight: $$\wts(u, v) = \dist(u, v, G_{l, l_i-l+1})$$
			and edge $(v, u)$ with edge weight:
			$$\wts(v, u) = \dist(v, u, G_{l, l_i-l+1})$$
			
			\item For any pair of vertices $u, v\in V(H^i_l)$ where $u\in V(\beta^j_l), j\in [i-5g, i]$, add an edge $(u, v)$, and assign a weight:
			$$\wts(u, v) =  \mu_X(u, v)$$
			If $\mu_X(u, v)$ was not assigned in Step (2), then it is infinity by default.
			
			\item For any pair of vertices $u, v\in V(H^i_l)$ where $u\in V(\beta^j_l), j\in [i, i+5g]$, add an edge $(u, v)$, and assign a weight:
			$$\wts(u, v) =  \mu_Y(u, v)$$
			If $\mu_Y(u, v)$ was not assigned in Step (2), then it is infinity by default.
		\end{enumerate}
	\end{itemize}
	After that, for each vertex $z\in V\brac{\beta^i_l}$, apply single-source shortest path on $z$ in $H^i_l$. In this way, for every pair of vertices $x, y\in V\brac{\beta^i_l}$, we have computed a distance $\dist\brac{x, y, H^i_l}$.
	
	Finally, for every edge $f\in E\brac{\beta^i_l}$, update the estimation $\est(s, t, G\setminus \{e_i, f\})$ with:
	$$\min_{x\in V(\beta^i_l[*, f)), y\in V(\beta^i_l(f, *])}\left\{\dist(s, x, G\setminus \{e_i\}) + \dist\brac{x, y, H^i_l} + \dist(y, t, G\setminus \{e_i\})\right\}$$
	
\end{enumerate}

\paragraph{Runtime.} The runtime of Step (1) is bounded by $\tilde{O}(|U|\cdot mL^2/g^2) = \tilde{O}(\frac{mnL^2}{g^3})$. As for the runtime of Step (2), consider any offset $1\leq b\leq 10g$. We argue that the Dijkstra procedure for all vertices on sub-paths $\beta^{b}_l, \beta^{b+10g}_l, \cdots, \beta^{b+10hg}_l$ has runtime at most $O(mg)$; this is because every vertex in $V$ is visited by at most $O(g)$ instances of Dijkstra rooted at vertices from some sub-paths $\beta^{b+10kg}_l$. Therefore, the overall runtime of Step (2) is bounded by $O(mLg)$ summing over all $1\leq b\leq 10g$ and $1\leq l\leq L/g$.

Finally, let us estimate the runtime of Step (3). For each index $i\in I$, there are at most $L/g$ different sub-paths $\beta^i_l$ as $|\alpha_i|\leq L$. By definition, the shortcut digraph $H_i$ contains at most $O(g^2+\frac{n\log n}{g})$ vertices, and so the number of edges within is bounded by $\tilde{O}(g^4 + \frac{n^2}{g^2})$, and hence the runtime of multi-source shortest paths computation for all vertices in $V(\beta^i_l)$ in $H^i_l$ is $\tilde{O}(g^5 + \frac{n^2}{g})$. Then, for each $f\in E(\beta^i_l)$, the time of calculating the formula
$$\min_{x\in V(\beta^i_l[*, f)), y\in V(\beta^i_l(f, *])}\left\{\dist(s, x, G\setminus \{e_i\}) + \dist(x, y, H_i) + \dist(y, t, G\setminus \{e_i\})\right\}$$
is bounded by $O(g^2)$. Summing over all $i, l$ and $f\in E(\beta^i_l)$, the runtime of this part is bounded by $\tilde{O}(L^2g^4 + \frac{n^2L^2}{g^2})$.

Taking a summation, the overall runtime for this case is at most:
$$\tilde{O}\left(mnL^2/g^3 + mLg + L^2g^4 + n^2L^2/g^2\right)$$

\paragraph{Correctness.} To prove that our algorithm computes the correct value for $|\rho_{i, f}|$, it suffices to prove that $\dist\brac{x_{i, f}, y_{i, f}, H^i_l} = |\rho_{i, f}[x_{i, f}, y_{i, f}]|$. First we argue that $\dist\brac{x_{i, f}, y_{i, f}, H^i_l} \geq |\rho_{i, f}[x_{i, f}, y_{i, f}]|$ due to the following reason.

\begin{claim}
	Any weighted edge $(u, v)$ in $H^i_{l}$ corresponds to a path from $u$ to $v$ in $G$ that does not contain any edge in $E(\pi)\cup E(\beta^i_l)$.
\end{claim}
\begin{proof}
	If the weighted edge $(u, v)$ was defined on Step (3)(a), then it corresponds to a shortest path in $G_{l, l_i-l+1}$ which excludes all edges in $E(\pi)\cup E(\beta^i_l)$.
	
	If the weighted edge $(u, v)$ was defined on Step (3)(b), suppose $u\in V(\beta^j_l)$ where $j\in [i-5g, i]$, then by definition of $\wts(u, v) = \mu_X(u, v)$, which is equal to the length of a path in graph $X^j_l$ which excludes all edges in $E(\pi)\cup E(\beta^i_l)$.
	
	Symmetrically, if the weighted edge $(u, v)$ was defined on Step (3)(c), suppose $u\in V(\beta^j_l)$ where $j\in [i, i+5g]$, then by definition of $\wts(u, v) = \mu_Y(u, v)$, it is equal to the length of a path in graph $Y^j_l$ which excludes all edges in $E(\pi)\cup E(\beta^i_l)$.
\end{proof}

For the rest, let us prove that $\dist\brac{x_{i, f}, y_{i, f}, H^i_l}\leq |\rho_{i, f}[x_{i, f}, y_{i, f}]|$. To do this, we will find a path in $H^i_l$ from $x_{i, f}$ to $y_{i, f}$ with weighted length at most $|\rho_{i, f}[x_{i, f}, y_{i, f}]|$.

\begin{claim}\label{interval}
	$\rho_{i, f}[x_{i, f}, y_{i, f}]$ does not contain any vertices in the following vertex subset: $$\bigcup_{j=0}^{i-5g-1}V\brac{\beta^j_l}\cup \bigcup_{j=i+5g+1}^L V\brac{\beta^j_{l_j-l_i+l}}$$
\end{claim}
\begin{proof}
	Suppose otherwise that the sub-path $\rho_{i, f}[x_{i, f}, y_{i, f}]$ contains a vertex $v\in \bigcup_{j=0}^{i-5g-1}V(\beta^j_l)\cup \bigcup_{j=i+5g+1}^L V(\beta^j_{l_j-l_i+l})$. Let us assume that $v\in V(\beta^j_l)$ for some index $0\leq j<i-5g$; if $v\in V(\beta^j_{l_j-l_i+l})$ for some $j>i+5g$, the proof will be similar.
	
	By the assumption that $i\in I$, we know that $|\pi[a_i, u_i]|\leq g$. Since $j < i-5g$, we know that vertex $u_j$ lies between $s, a_i$, and consequently $|\pi[s, a_j]| \leq |\pi[s, u_j]|< |\pi[s, a_i]| - 4g$.
	
	Consider the path $\pi[s, a_j]\circ \alpha_j[a_j, v]$. We first argue that this path does not contain edges from $\{e_i, f\}$. Clearly, $\pi[s, a_j]\circ \alpha_j[a_j, v]$ does not contain the edge $e_i$, since $a_j$ lies between $s$ and $u_i$, and $E(\alpha_j[a_j, v])\cap E(\pi) = \emptyset$. As for the position of $f$, if $f\in E(\alpha_j[a_j, v])$, then there must be a vertex $z\in V(\alpha_j[a_j, v])\cap V(\beta^i_l)$. Then $\pi[s, a_j]\circ\alpha_j[a_j, z]\circ \alpha_i[z, b_i]\circ \alpha_i[b_i, t]$ is a replacement path that avoids edge $e_i$, with length at most:
	$$\begin{aligned}
		&|\pi[s, a_j]| + |\alpha_j[a_j, z] | + |\alpha_i[z, b_i]| +  |\alpha_i[b_i, t]|\\
		&< (|\pi[s, a_i]| - 4g) + |\alpha_j[a_j, v]| + |\alpha_i[z, b_i]| +  |\alpha_i[b_i, t]|\\
		&\leq (|\pi[s, a_i]| - 4g) + l\cdot g + |\alpha_i[z, b_i]| +  |\alpha_i[b_i, t]|\\
		&\leq (|\pi[s, a_i]| - 4g) + (|\alpha_i[a_i, z]| + g) + |\alpha_i[z, b_i]| +  |\alpha_i[b_i, t]|\\
		&\leq |\pi[s, a_i] + |\alpha_i[a_i, z]| + |\alpha_i[z, b_i]| +  |\alpha_i[b_i, t]| - 3g = |\pi[s, a_i]\circ \alpha_i\circ\pi[b_i, t]| - 3g
	\end{aligned}$$
	which is a contradiction that $\pi[s, a_i]\circ \alpha_i\circ\pi[b_i, t]$ is a shortest replacement path avoiding $e_i$.
	
	Next, we argue that $|\pi[s, a_j]\circ \alpha_j[a_j, v]| < |\rho_{i, f}[s, x_{i, f}]| < |\rho_{i, f}[s, v]|$. In fact, by $|\pi[s, a_j]| < |\pi[s, a_i]| - 4g$, we have:
	$$\begin{aligned}
		&|\pi[s, a_j]| + |\alpha_j[a_j, v]| < (|\pi[s, a_i]| - 4g )+ l\cdot g\\&\leq  (|\pi[s, a_i]| - 4g )+ (|\alpha_i[a_i, x_{i, f}]| + g) = |\rho_{i, f}[s, x_{i, f}]| -3g
	\end{aligned}$$
	
	As we have proved, $\pi[s, a_j]\circ \alpha_j[a_j, v]$ is a path avoiding $\{e_i, f\}$ of length less than $|\rho_{i, f}[s, v]|$. So, if we replace the prefix $\rho_{i, f}[s, v]$ with $\pi[s, a_j]\circ \alpha_j[a_j, v]$ and consider a new path: $$\rho^\prime = \pi[s, a_j]\circ \alpha_j[a_j, v]\circ \rho_{i, f}[v, t]$$
	then we have a new replacement path $\rho^\prime$ avoiding $\{e_i, f\}$ from $s$ to $t$ with a strictly better distance, a contradiction. See \Cref{1fail-short-case5-4} for an illustration.
	
	\begin{figure}
		\centering
		\begin{tikzpicture}[thick,scale=0.7]
	\draw (-11, 0) node(1)[circle, draw, fill=black, inner sep=0pt, minimum width=6pt, label=180: {$s$}] {};
	\draw (11, 0) node(2)[circle, draw, fill=black, inner sep=0pt, minimum width=6pt, label=0: {$t$}] {};
	
	\draw (-6, 0) node(3)[circle, draw, fill=black, inner sep=0pt, minimum width=6pt, label=-90: {$a_{j}$}] {};
	\draw (-2, 0) node(4)[circle, draw, fill=black, inner sep=0pt, minimum width=6pt, label=-90: {$b_{j}$}] {};
	
	\draw (2.5, 0) node(5)[circle, draw, fill=black, inner sep=0pt, minimum width=6pt, label=-90: {$a_{i}$}] {};
	\draw (6.5, 0) node(6)[circle, draw, fill=black, inner sep=0pt, minimum width=6pt, label=-90: {$b_{i}$}] {};
	
	\draw (2.5, 2.5) node(7)[circle, draw, fill=black, inner sep=0pt, minimum width=6pt, label=-20: {$x_{i, f}$}] {};
	\draw (3.8, 4.85) node(8)[circle, draw, fill=black, inner sep=0pt, minimum width=6pt, label=-20: {$y_{i, f}$}] {};
	
	\draw (-5.6, 4) node(9)[circle, draw, fill=black, inner sep=0pt, minimum width=6pt, label=0: {$v$}] {};
	
	\draw (4.5, 0) node[cross=6, red, label=-90: {$e_i$}] {};
	\draw (2.9, 4) node[cross=6, red, label=-45: {$f$}] {};
	
	
	
	
	\begin{scope}[on background layer]
		\draw [->, line width = 0.5mm] (1) to (2);
		
		\draw [->, line width = 0.5mm, color = cyan, dotted] (3) to[out=90, in=-90] (-6, 2);
		\draw [->, line width = 0.5mm, color = cyan] (-6, 2) to[out=90, in=180] (-4, 5);
		\draw [->, line width = 0.5mm, color = cyan, dotted] (-4, 5) to[out=0, in=90] (4);
		
		\draw [->, line width = 0.5mm, color = cyan, dotted] (5) to[out=90, in=-90] (2.5, 2);
		\draw [->, line width = 0.5mm, color = cyan] (2.5, 2) to[out=90, in=180] (4.5, 5);
		\draw [->, line width = 0.5mm, color = cyan, dotted] (4.5, 5) to[out=0, in=90] (6);
		
		\draw [->, line width = 0.5mm, color = orange] (7) to[out=180, in=-60] (9);
		\draw [->, line width = 0.5mm, color = orange] (9) to[out=100, in=130] (8);
		
	\end{scope}
\end{tikzpicture}
		\caption{The sub-path $\rho_{i, f}[x_{i, f}, y_{i, f}]$ is drawn as the orange curve. If $\rho_{i, f}[x_{i, f}, y_{i, f}]$ contains a vertex $v\in V\brac{\beta^j_l}$, then the alternate path $\pi[s, a_j]\circ \alpha_j[*, v]\circ \rho_{i, f}[v, y_{i, f}]\circ \alpha_i[y_{i, f}, b_i]\circ\pi[b_i, t]$ be a shorter replacement path than $\rho_{i, f}$ avoiding $\{e_i, f\}$.}
		\label{1fail-short-case5-4}
	\end{figure}
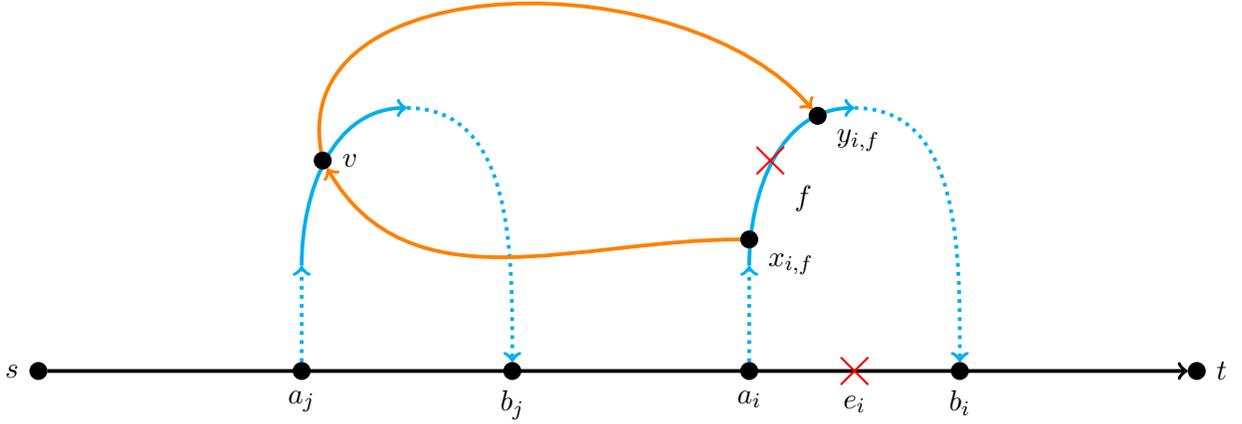
\end{proof}

Decompose the path $\rho_{i, f}[x_{i, f}, y_{i, f}]$ into minimal sub-paths whose endpoints are belonging to $V(H^i_l)$; this is achievable because both $x_{i, f}, y_{i, f}\in V(H^i_l)$. To prove the upper bound:
$$\dist(x_{i, f}, y_{i, f}, H^i_l)\leq |\rho_{i, f}[x_{i, f}, y_{i, f}]|$$
it suffices to show that for any two consecutive vertices $u, v\in V(H^i_l)$ on path $\rho_{i, f}[x_{i, f}, y_{i, f}]$, we have $\wts(u, v)\leq |\rho_{i, f}[u, v]|$. Consider several cases below.
\begin{itemize}[leftmargin=*]
	\item One of $u, v$ belongs to the pivot set $U$.
	
	Without loss of generality, assume that $u\in U$. It suffices to show that the shortest path from $u$ to $v$ in graph $G_{l, l_i-l+1}$ is has the same length as $\rho_{i, f}[u, v]$.
	
	Since $v$ is the next vertex in $V(H^i_l)$ after $u$ on the path $\rho_{i, f}$, the sub-path $\rho_{i, f}[u, v]$ does not contain any vertices from $V(H^i_l)$ except for its endpoints; in other words, $\rho_{i, f}[u, v]$ does not contain vertices from $U\cup \bigcup_{j=i-5g}^{i+5g}V(\beta^j_l)$. Additionally, according to \Cref{interval}, $\rho_{i, f}[u, v]$ does not contain (internally) any vertices in $$\bigcup_{j=0}^{i-5g-1}V\brac{\beta^j_l}\cup \bigcup_{j=i+5g+1}^L V\brac{\beta^j_{l_j-l_i+l}}$$
	Therefore, $\rho_{i, f}[u, v]$ does not contain (internally) any vertices from the set:
	$$U\cup \bigcup_{j\in [i-5g, i+5g]\cap I}^{}V\brac{\beta^j_l}\cup \bigcup_{j=0}^{i-5g-1}V\brac{\beta^j_l}\cup \bigcup_{j=i+5g+1}^L V\brac{\beta^j_{l_j-l_i+l}}\supseteq U\cup \bigcup_{\beta \in \mathcal{P}_{l, l_i-l+1}}V(\beta)$$
	
	By definition of graph $G_{l, l_i-l+1}$, we know that $\dist(u, v, G_{l, l_i-l+1})\leq |\rho_{i, f}[u, v]|$.
	
	\item Both of $u, v$ are not in $U$.
	
	Assume that $u\in V(\beta^c_l)$ for some $c\in [i-5g, i+5g]$, and $v\in V(\beta^d_l)$ for some $d\in [i-5g, i+5g]$. Without loss of generality, assume that $c\leq i$; if $c\geq i$, a symmetric argument would work.
	
	Since $v$ is the next vertex in $V(H^i_l)$ after $u$ on the path $\rho_{i, f}$ and that $U$ is a uniformly random vertex set of size $\frac{10n\log n}{g}$, with high probability over the randomness of $U$, we must have $|\rho_{i, f}[u, v]|\leq g$.
	
	Similar to the previous case, we know that $\rho_{i, f}[u, v]$ does not contain (internally) vertices from $\bigcup_{j=i-5g}^{i+5g}V\brac{\beta^j_l}$. Therefore, by definition of $X^c_l$, we know that $\rho_{i, f}[u, v]\subseteq X^c_l$; namely, the sub-path $\rho_{i, f}[u, v]$ belongs to graph $X^c_l$. Therefore, to prove that $\wts(u, v) = \mu_X(u, v)\leq |\rho_{i, f}[u, v]|$, it suffices to show that the truncated Dijkstra procedure correctly computes the value $\mu_X(u, v) = \dist(u, v, X^c_l)$.
	
	Decompose the integer $c = b+10kg$, where $1\leq b\leq 10g$, $k\geq 0$. To prove that the truncated Dijkstra procedure correctly computes the value $\mu_X(u, v) = \dist(u, v, X^c_l)$, it suffices to show that the vertex set $A_{b, l}$ contains all vertices on $\rho_{i, f}[u, v]$ when $u$ is performing a truncated Dijkstra in the induced subgraph $X^c_l[A_{b, l}]$; in other words, we need to show that any vertex on $\rho_{i, f}[u, v]$ has not been pruned from $A_{b, l}$ by truncated Dijkstraes from vertices on previous sub-paths $\beta^b_l, \beta^{b+10g}_l, \cdots, \beta^{b+10(k-1)g}_l$.
	
	Assume otherwise there is a vertex $z\in V(\rho_{i, f}[u, v])$ which was also visited by the truncated Dijkstra of some vertices $w\in V(\beta^{b + 10jg}_l)$ for some $j<k$. As all Dijkstra searches are truncated up to depth $g$, we know that there is a path $\gamma$ from $w$ to $z$ in $X^{b+10jg}_l$ of length at most $g$. Consider the path $$\theta = \pi[s, a_{b+10jg}]\circ \alpha_{b+10jg}[*, w]\circ \gamma\circ \rho_{i, f}[z, v]\circ \alpha_d[v, b_d]\circ \pi[b_d, t]$$ and claim two properties of it.
	\begin{claim}
		$\theta$ is a path from $s$ to $t$ that avoids the edge $e_d$.
	\end{claim}
	\begin{proof}
		It is clear that path $\theta$ departs from $\pi$ at vertex $a_{b+10jg}$ and converges with $\pi$ at vertex $b_d$. So it suffices to show that $a_{b+10jg}$ lies between $s$ and vertex $u_d$. This is straightforward since $a_{b+10jg}$ lies on path $\pi[s, u_{b+10jg}]$ which is strictly a prefix of $\pi[s, u_d]$, as $d\geq i -5g\geq b+(10k-5)g > b+10jg$.
	\end{proof}

	To reach a contradiction, we show that $|\theta|$ is a strictly better replacement path from $s$ to $t$ that avoids $e_d$ than path $\pi[s, a_d]\circ \alpha_d\circ\pi[b_d, t]$. In fact, on the one hand, since $d\in I$, we know that $|\pi[a_d, u_d]|\leq g$. Therefore,
	$$|\pi[s, a_{b+10jg}]| \leq |\pi[s, u_{b+10jg}]| = b+10jg \leq 10(k-1)g \leq i-10g \leq d -5g \leq |\pi[s, a_d]| -4g$$
	On the other hand, since $v\in V(\beta^d_l)$ and $w\in V(\beta^{b+10jg}_l)$, we know that
	$$|\alpha_d[*, v]|\geq (l-1)g+1 > |\alpha_{b+10jg}[*, w]| - g$$
	Finally, as $|\gamma|, |\rho_{i, f}[z, v]|\leq g$, we have:
	$$\begin{aligned}
		|\theta| &\leq (|\pi[s, a_d]| - 4g) + (|\alpha_d[*, v]|+g) + |\gamma| + |\rho_{i, f}[z, v]| + |\alpha_d[v, b_d]| + |\pi[b_d, t]|\\
		&\leq |\pi[s, a_d]| + |\alpha_d[*, v]|+|\alpha_d[v, b_d]| + |\pi[b_d, t]|-g\\
		&= |\pi[s, a_d]\circ \alpha_d\circ\pi[b_d, t]| - g
	\end{aligned}$$
	which contradicts the fact that $\pi[s, a_d]\circ \alpha_d\circ\pi[b_d, t]$ is the shortest replacement path avoiding $e_d$. See \Cref{1fail-short-case5-5} for an illustration.
	
	\begin{figure}
		\centering
		\begin{tikzpicture}[thick,scale=0.7]
	\draw (-11, 0) node(1)[circle, draw, fill=black, inner sep=0pt, minimum width=6pt, label=180: {$s$}] {};
	\draw (11, 0) node(2)[circle, draw, fill=black, inner sep=0pt, minimum width=6pt, label=0: {$t$}] {};
	
	\draw (-9, 0) node(3)[circle, draw, fill=black, inner sep=0pt, minimum width=6pt, label=-90: {$a_{b+10jg}$}] {};
	\draw (-5, 0) node(4)[circle, draw, fill=black, inner sep=0pt, minimum width=6pt, label=-90: {$b_{b+10jg}$}] {};
	
	\draw (-2, 0) node(5)[circle, draw, fill=black, inner sep=0pt, minimum width=6pt, label=-90: {$a_c$}] {};
	\draw (2, 0) node(6)[circle, draw, fill=black, inner sep=0pt, minimum width=6pt, label=-90: {$b_c$}] {};
	
	\draw (5.5, 0) node(7)[circle, draw, fill=black, inner sep=0pt, minimum width=6pt, label=-90: {$a_{d}$}] {};
	\draw (9.5, 0) node(8)[circle, draw, fill=black, inner sep=0pt, minimum width=6pt, label=-90: {$b_{d}$}] {};
	
	\draw (-8.6, 4) node(9)[circle, draw, fill=black, inner sep=0pt, minimum width=6pt, label=-30: {$w$}] {};
	\draw (-1.6, 4) node(10)[circle, draw, fill=black, inner sep=0pt, minimum width=6pt, label=-30: {$u$}] {};
	\draw (5.9, 4) node(11)[circle, draw, fill=black, inner sep=0pt, minimum width=6pt, label=-30: {$v$}] {};
	\draw (0, 8) node(12)[circle, draw, fill=black, inner sep=0pt, minimum width=6pt, label=$z$] {};
	
	\draw (7.5, 0) node[cross=6, red, label=-90: {$e_d$}] {};
	
	\begin{scope}[on background layer]
		\draw [->, line width = 0.5mm] (1) to (2);
		
		\draw [->, line width = 0.5mm, color = cyan, dotted] (3) to[out=90, in=-90] (-9, 2);
		\draw [->, line width = 0.5mm, color = cyan] (-9, 2) to[out=90, in=180] (-7, 5);
		\draw [->, line width = 0.5mm, color = cyan, dotted] (-7, 5) to[out=0, in=90] (4);
				
		\draw [->, line width = 0.5mm, color = cyan, dotted] (5) to[out=90, in=-90] (-2, 2);
		\draw [->, line width = 0.5mm, color = cyan] (-2, 2) to[out=90, in=180] (0, 5);
		\draw [->, line width = 0.5mm, color = cyan, dotted] (0, 5) to[out=0, in=90] (6);
				
		\draw [->, line width = 0.5mm, color = cyan, dotted] (7) to[out=90, in=-90] (5.5, 2);
		\draw [->, line width = 0.5mm, color = cyan] (5.5, 2) to[out=90, in=180] (7.5, 5);
		\draw [->, line width = 0.5mm, color = cyan, dotted] (7.5, 5) to[out=0, in=90] (8);
		
		\draw [->, line width = 0.5mm, color = red] (9) to[out=120, in=180] (12);
		\draw [->, line width = 0.5mm, color = orange] (10) to[out=90, in=260] (12);
		\draw [->, line width = 0.5mm, color = orange] (12) to[out=0, in=130] (11);
	\end{scope}
\end{tikzpicture}
		\caption{If $\mu_X(u, v)$ does not capture the orange path $\rho_{i, f}[u, v]$, then a previous Dijkstra search from vertex $w$ must have intercepted $\rho_{i, f}[u, v]$ at a vertex $z$ through a path $\gamma$ drawn as the red curve. In this case, $\alpha_{b+10jg}[*, w]\circ \gamma\circ \rho_{i, f}[z, v]\circ \alpha_d[v, b_d]$ would be a better detour than $\alpha_d$ for avoiding $e_d$.}
		\label{1fail-short-case5-5}
	\end{figure}
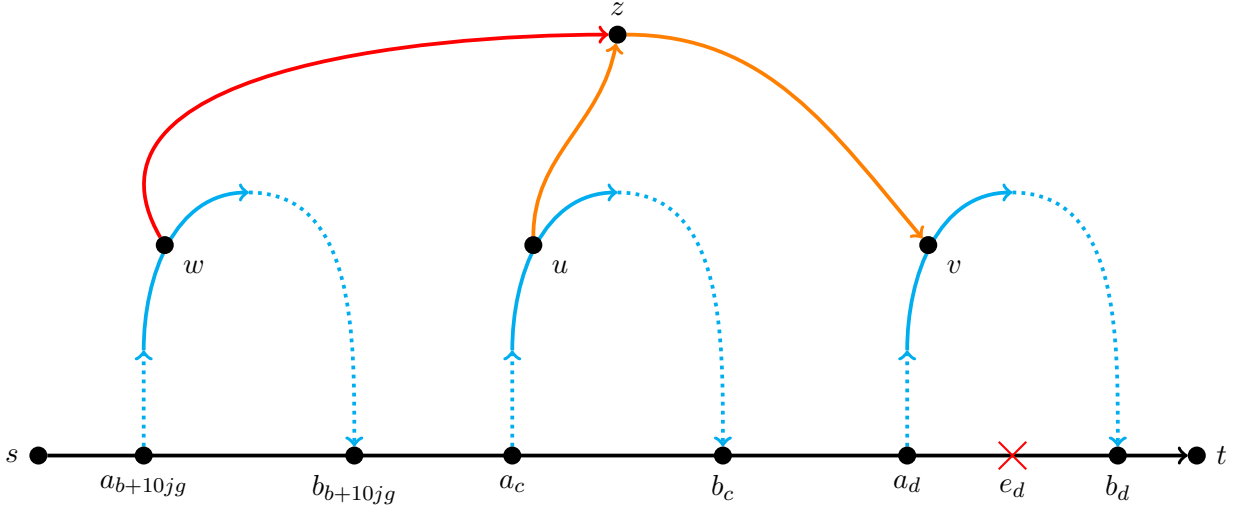
\end{itemize}

\subsection*{Proof of \Cref{1fail-special}}
Summarizing the total runtime of all five cases, the overall runtime is bounded by as following:
$$\tilde{O}\left(\frac{mn^{1.5} + n^3}{L} + mn^{1/2}L/g + n^2L/g + mL^2/g + mnL^2/g^3 + mLg + L^2g^4 + n^2L^2/g^2\right)$$

\section{One failure on a long $st$-path}\label{1fail-long}

In this section, we study the case where only one edge failure lies on the shortest path, plus that the input graph is dense and $\dist(s, t, G)$ could be as large as $O(n)$. Let $G = (V, E)$ be a digraph with $n$ vertices, and consider a pair of vertices $s, t$ as well as an $st$-shortest path $\pi = \langle s=u_0, u_1, u_2, \cdots, u_{|\pi|} = t\rangle$. The task is to compute for any pairs of edges $f_1, f_2$, the value of $\dist(s, t, G\setminus \{f_1, f_2\})$ where $f_1\in E(\pi), f_2\notin E(\pi)$. We will prove the following statement.

\begin{lemma}\label{1fail-exact}
	All values of $\dist(s, t, G\setminus \{f_1, f_2\})$ can be computed with high probability in time $\tilde{O}\left(n^{3-1/18}\right)$ when $f_1\in E(\pi)$ while $f_2\notin E(\pi)$.
\end{lemma}

Let $g<n^{1/2} < L$ be two integer parameters to be chosen later. Divide the shortest path $\pi$ into sub-paths of length equal to $L$ (expect for the last sub-path) as $\pi = \gamma_1\circ\gamma_2\circ\cdots \circ\gamma_h, h \leq \ceil{n/L}$; for each $1\leq i\leq h$, let $s_i, t_i$ be the source and terminal of the sub-path $\gamma_i$.

Similar to the previous section, for any $0\leq i <|\pi|$, define $e_i = (u_i, u_{i+1})$, and let $\alpha_i$ be the detour for $e_i$ with divergence and convergence vertex $a_i, b_i$ respectively; all the detours $\alpha_i$ can be computed in time $\tilde{O}(n^{2.5})$ by \Cref{rp}. Throughout the algorithm, for each $0\leq i<|\pi|$ and every edge $f\in E(\alpha_i)$, we will maintain a distance estimation $\est(s, t, G\setminus \{e_i, f\})\geq \dist(s, t, G\setminus \{e_i, f\})$, and in the end it will be guaranteed that $\est(s, t, G\setminus \{e_i, f\}) = \dist(s, t, G\setminus \{e_i, f\})$.

For each $f\in E(\alpha_i)$, let $\rho_{i, f}$ be a canonical shortest replacement path avoiding $\{e_i, f\}$ from $s$ to $t$. Let $z_{i, f}, w_{i, f}$ be the divergence and convergence vertex on $\pi$ of $\rho_{i, f}$. Consider any index $0\leq i<|\pi|$, and assume $e_i\in E(\gamma_l)$. There are several cases to study.

\subsection*{Case 1: $z_{i, f}\notin V(\gamma_l)$ or $w_{i, f}\notin V(\gamma_l)$}
\paragraph{Algorithm.} Without loss of generality, assume that $z_{i, f}\notin V(\gamma_l)$. Apply \Cref{ssrp} on graph $H_l = G\setminus E(\gamma_l)$ at source vertex $s$, and initialize $\est(s, t, G\setminus \{e_i, f\})\leftarrow \dist(s, t, H_l\setminus \{f\})$ for any $f$. Then, update the estimate:
$$\est(s, t, G\setminus \{e_i, f\})\leftarrow \min_{v\in V(\gamma_l(f, *])}\{\est(s, t, G\setminus \{e_i, f\}), \dist(s, v, H_l\setminus \{f\}) + |\pi[v, t]|\}$$

\paragraph{Runtime.} The runtime of applying \Cref{ssrp} is $\tilde{O}(\frac{n^{3.5}}{L})$. As for the value $\est(s, t, G\setminus \{e_i, f\})$, if $f\in E(\alpha_i)$, then calculating the value $\min_{v\in V(\gamma_l(f, *])}\{\dist(s, v, H_l\setminus \{f\}) + |\pi[v, t]|\}$ takes time $O(L)$, and thus the total runtime over all $i$ and $f\in E(\alpha_i)$ is bounded by $O(n^2L)$.

\paragraph{Correctness.} Since $\rho_{i, f}$ is canonical and that $z_{i, f}\notin V(\gamma_l)$, $\rho_{i, f}$ diverges from $\pi$ before the failed edge $f$ and then converges with $\pi$ at a vertex $v$ after $f$. If $v\notin V(\gamma_l)$, the the replacement path $\rho_{i, f}$ skips over the entire sub-path $\gamma_l$, and so $|\rho_{i, f}| = \dist(s, t, H_l\setminus \{f\})$. Otherwise, the value of $|\rho_{i, f}| = \dist(s, v, H_l\setminus \{f\}) + |\pi[v, t]|$ must have been calculated in by the minimum formula.

\subsection*{Case 2: $z_{i, f}, w_{i, f}\in V(\gamma_l)$ and $|\rho_{i, f}[z_{i, f}, w_{i, f}]|>L$}
\paragraph{Algorithm.} Take a uniformly random vertex subset $U$ of size $\frac{10n\log n}{L}$, and for each vertex $v\in U$, apply \Cref{ssrp} to compute single-source replacement paths to and from $v$ in graph $H = G\setminus E(\pi)$. Then update the value $\est(s, t, G\setminus \{e_i, f\})$ with the minimum clause:
$$\min_{z\in V(\gamma_l[*, u_i]), w\in V(\gamma_l[u_{i+1}, *]), v\in U}\{|\pi[s, z]| + |\dist(z, v, H\setminus \{f\})| + |\dist(v, w, H\setminus \{f\})| + |\pi[w, t]| \}$$

The correctness of this procedure is straightforward, since the value of $|\rho_{i, f}|$ is witnessed by the above formula when $z = z_{i, f}, w = w_{i, f}$.

\paragraph{Runtime.} Applying \Cref{ssrp} for each $v\in U$ takes runtime $\tilde{O}(\frac{n^{3.5}}{L})$. To implement the computations of the minimum formula  in subcubic time, we need to utilize some data structures. 

For each vertex $v\in U$ and for each edge $f$ belong to the single-source shortest paths tree rooted at $v$ in $H$, store all the values $\{|\pi[s, a]| + \dist(a, v, H\setminus \{f\})\mid \forall a\in V(\pi)\}$ in an array. Using standard data structures, for any sub-path $\gamma$ of $\pi$, we can query the value $\min_{a\in V(\gamma)}\{|\pi[s, a]| + \dist(a, v, H\setminus \{f\})\}$ in $O(\log n)$ time. Symmetrically, we can also query the value $\min_{b\in V(\gamma)}\{|\pi[s, b]| + \dist(v, b, H\setminus \{f\})\}$ in $O(\log n)$ time. The total time of building these arrays is bounded by $\tilde{O}(n^3/L)$.

Now, to calculate the value of $\min_{z\in V(\gamma_l[*, u_i]), w\in V(\gamma_l[u_{i+1}, *]), v\in U}\{|\pi[s, z]| + |\dist(z, v, H\setminus \{f\})| + |\dist(v, w, H\setminus \{f\})| + |\pi[w, t]| \}$, we can go over each vertex $v\in U$ and query the values of $\min_{z\in V(\gamma_l[*, u_i])}\{|\pi[s, z]| + |\dist(z, v, H\setminus \{f\})| \}$ and $\min_{w\in V(\gamma_l[u_{i+1}, *])}\{ |\dist(v, w, H\setminus \{f\})| + |\pi[w, t]| \}$ separately, which takes time $\tilde{O}(n^3/L)$ for all $i$ and $f\in E(\alpha_i)$.

\subsection*{Case 3: $z_{i, f}, w_{i, f}\in V(\gamma_l)$ and $|\rho_{i, f}[z_{i, f}, w_{i, f}]|\leq L$}
\paragraph{Algorithm.} For each index $1\leq l\leq h$, we will build a subgraph $G_l\subseteq G$, such that $G_l$ contains all paths $\rho_{i, f}$ if $z_{i, f}, w_{i, f}\in V(\gamma_l)$ and $|\rho_{i, f}[z_{i, f}, w_{i, f}]|\leq L$. This is done by the following truncated Djikstra procedure.

As before, let $H = G\setminus E(\pi)$. For any offset $1\leq b\leq 5$, initialize a vertex set $A_b \leftarrow V$, then go over the sequence of sub-paths $\gamma_{b}, \gamma_{b+5}, \cdots, \gamma_{b+5k}, k =\ceil{h/5}$. For each sub-path $\gamma_{b+5j}, 0\leq j\leq k$, enumerate all vertices $u\in V(\gamma_{b+5j})$ and perform a Djikstra tree truncated at depth $L$ in the induced subgraph $H[A_b]$. After all vertices in $V(\gamma_{b+5j})$ have been processed, let $P_{b+5j}\subseteq A_b$ be the set of vertices visited by any truncated Djikstra tree rooted at some vertices in $V(\gamma_{b+5j})$, and then define the graph $G_{b+5j} = G[P_{b+5j}]$, and update $A_b\leftarrow A_b\setminus P_{b+5j}$.

After all graphs $G_l$ have been constructed, for each $1\leq l\leq h$, apply \Cref{1fail-special} on $G_l$ with source and terminal being $(s_l, t_l)$ and the same parameter $L$, which computes values $\est(s_l, t_l, e_i, f)$ for each $f\in E(\alpha_i)$. Finally, for each $i$ and $f\in E(\alpha_i)$, update the estimation:
$$\est(s, t, G\setminus \{e_i, f\})\leftarrow \min\{\est(s, t, G\setminus \{e_i, f\}), |\pi[s, s_l]| +  \est(s_l, t_l, e_i, f) + |\pi[t_l, t]|\}$$

\paragraph{Runtime.} For each offset $b$, due to the pruning procedure $A_b\leftarrow A_b\setminus P_{b+5j}$, every vertex in $G$ is explored by the truncated Djikstra for at most $L$ times. Therefore, the total time of truncated Djikstraes is bounded by $O(n^2L)$.

Next, we bound the runtime of applying \Cref{1fail-special} on $G_l$. Suppose each graph $G_l$ has $m_l$ edges and $n_l$ vertices, then by the pruning procedure, we have $\sum_{l=1}^h m_l \leq O(n^2)$, and $\sum_{l=1}^h n_l = O(n)$. By the runtime bound of \Cref{1fail-special}, the runtime of the instance on graph $G_l$ is bounded by:
$$\tilde{O}\left(\frac{m_ln_l^{1.5} + n_l^3}{L} + m_ln_l^{1/2}L/g + n^2_lL/g + m_lL^2/g + m_ln_lL^2/g^3 + m_lLg + L^2g^4 + n^2_lL^2/g^2\right)$$
By convexity, the total sum of the above formula over all $1\leq l \leq h$ is bounded by:
$$\tilde{O}\left(\frac{n^{3.5}}{L} + n^3L^2/g^3 + n^2Lg + L^2g^4\right)$$

Setting $L = \ceil{n^{5/9}}, g = \ceil{n^{7/18}}$, taking a summation over all $1\leq l\leq \ceil{n/L}$, and by looking only at the dominant terms, the total runtime is at most: 
$$\tilde{O}\left(n^{3.5} /L + n^3L^2/g^3 + n^2Lg + L^2g^4\right) = \tilde{O}(n^{3-1/18})$$

\paragraph{Correctness.} Consider any $0\leq i < |\pi|$ and $f\in E(\alpha_i)$ such that $z_{i, f}, w_{i, f}\in V(\gamma_l)$ and $|\rho_{i, f}[z_{i, f}, w_{i, f}]|\leq L$, and assume $l = b+5c, 1\leq b\leq 5$. It suffices to prove that the entire path $\rho_{i, f}[z_{i, f}, w_{i, f}]$ belongs to the induced subgraph $G_l$. Suppose otherwise, then there must be a vertex $v\in V(\rho_{i, f}[z_{i, f}, w_{i, f}])$ which has been visited by some previous truncated Djikstra search at vertices $u\in V(\gamma_{b+5d}), d<c$. Let $\gamma$ be the shortest path from $u$ to $v$ in $G\setminus E(\pi)$ of length at most $L$. Then, consider the path $\gamma\circ \rho_{i, f}[v, w_{i, f}]$. On the one hand, as $|\gamma|, |\rho_{i, f}[z_{i, f}, w_{i, f}]|\leq L$, we know that the length of $\gamma\circ \rho_{i, f}[v, w_{i, f}]$ is at most $2L$. On the other hand, since $u\in V(\gamma_{b+5d}), v\in V(\gamma_{b+5c})$ and $\pi$ is a shortest path, $\dist(u, v, G)$ is at least $L(b+5c) - L(b+5d) - L\geq 4L$, which is a contradiction.

\section{Both failures on the $st$-path}\label{2fail-unweighted}

In this section, we study the case where both edge failures $f_1, f_2$ are lying the shortest path. Let $G = (V, E)$ be a digraph with $n$ vertices, and consider a pair of vertices $s, t$ as well as an $st$-shortest path $\pi = \langle s=u_0, u_1, u_2, \cdots, u_{|\pi|} = t\rangle$.  For convenience, define $H = G\setminus E(\pi)$. The task is to compute for any pairs of edges $f_1, f_2\in V(\pi)$, the value of $\dist(s, t, G\setminus \{f_1, f_2\})$. We will prove the following statement.
\begin{lemma}\label{2fail-exact}
	All values of $\dist(s, t, G\setminus \{f_1, f_2\})$ can be computed in time $\tilde{O}(n^{3-1/7})$ when both edges $f_1, f_2$ are on $\pi$.
\end{lemma}
\begin{proof}[Proof of \Cref{subcubic}]
	This is a direct combination of \Cref{1fail-exact} and \Cref{2fail-exact}.
\end{proof}

Let $g<n^{1/2} < L$ be two parameters to be chosen later; note that $L, g$ do not have to be the same as in the previous section. Divide the shortest path from $s$ to $t$ into sub-paths of length $L$; that is, $\pi = \gamma_1\circ\gamma_2\circ\cdots\circ\gamma_{\ceil{n/L}}$, and let $s_i, t_i$ be the head and tail of sub-path $\gamma_i$. For each sub-path $\gamma_l$, subdivide it into segments of length $g$ as $\gamma_l = \alpha^l_1\circ \alpha^l_2\circ\cdots\circ \alpha^l_{\ceil{L/g}}$, where $\alpha^l_k$ goes from vertex $s^l_k$ to $t^l_k$.

For any $0\leq i<j<|\pi|$, let $\rho_{i, j}$ refer to the shortest replacement path under weight perturbation with respect to dual edge failures $\{e_i, e_j\}$, and assume $e_i\in E(\gamma_{l_i}), e_j\in E(\gamma_{l_j})$; further assume $e_i\in E\brac{\alpha^{l_i}_{h_i}}$ and $e_j\in E\brac{\alpha^{l_j}_{h_j}}$. Since $\rho_{i, j}$ is canonical, we can assume that $\rho_{i, j}$ diverges from $\pi[s, u_i]$ at vertex $a_{i, j}$ and converges with $\pi[u_{j+1}, t]$ at vertex $b_{i, j}$. If $\rho_{i, j}$ intersects with $\pi[u_{i+1}, u_j]$, let $x_{i, j}\in V(\rho_{i, j})\cap V(\pi[u_{i+1}, u_j])$ be the vertex closest to $u_{i+1}$, and $y_{i, j}$ the closest one to $u_j$.

\subsection*{Case 1: $E(\rho_{i, j})\cap E(\pi[u_{i+1}, t_{l_i}]) = \emptyset$ or $E(\rho_{i, j})\cap E(\pi[s_{l_j}, u_j]) = \emptyset$}
\paragraph{Algorithm.} Without loss of generality, assume $E(\rho_{i, j})\cap E(\pi[s_{l_j}, u_j]) = \emptyset$. Then, apply \Cref{ssrp} on $s$ in graph $H_j = G\setminus E(\gamma_{l_j})$. To compute $|\rho_{i, j}|$, update $\est(s, t, G\setminus\{e_i, e_j\})$ with the following minimum clause:
$$\min_{v\in V(\pi[u_{j+1}, t_{l_j}])}\{\dist(s, t, H_j\setminus \{e_i\}), \dist(s, v, H_j\setminus \{e_i\}) + |\pi[v, t]|\}$$
The total runtime of apply \Cref{ssrp} is bounded by $\tilde{O}(n^{3.5}/L)$, and the runtime of calculating all distance estimates is bounded by $O(n^2L)$.

\paragraph{Correctness.} If $\rho_{i, j}$ skips over the entire sub-path $\gamma_{l_j}$, then $|\rho_{i, j}|$ would be equal to $\dist(s, t, H_j\setminus\{e_i\})$. Otherwise, as $E(\rho_{i, j})\cap E(\pi[s_{l_j}, u_j]) = \emptyset$, $\rho_{i, j}$ must contain a vertex $v\in V(\pi[u_{j+1}, t_{l_j}])$, and in this case we have $|\rho_{i, j}| = \dist(s, v, H_j\setminus \{e_i\}) + |\pi[v, t]|$.

\subsection*{Case 2: $E(\rho_{i, j})\cap E(\pi[u_{i+1}, t_{l_i}]) \neq \emptyset$, $E(\rho_{i, j})\cap E(\pi[s_{l_j}, u_j]) \neq \emptyset$, and $l_j - l_i \geq 2$}
\paragraph{Preparation.} As a preparatory step, compute the following shortest paths information.
\begin{itemize}[leftmargin=*]
	\item Uniformly sample a vertex subset $U$ of size $\frac{10n\log n}{L}$. Then, for each vertex $v\in U$ and every segment $\alpha^l_h\subseteq \gamma_l$, compute a single-source shortest paths tree to and from $v$ in graphs $H$, $H\cup E(\pi[s, s^l_h])$ and $H\cup E(\pi[t^l_h, t])$. Then, apply \Cref{ssrp} on each $p\in U$ in graph $G\setminus E(\pi[s, t_l])$ for each $l$.
	
	\item Then, for each index $h$, apply \Cref{ssrp} on $t_h$ in graph $G\setminus E(\pi[s, t_h])$, and on $s_h$ in graph $G\setminus E(\pi[s_h, t])$.
	
	\item Consider any sub-path $\gamma_l$ and one of its segments $\alpha^l_h$. For each $p\in U$, compute single-source shortest paths to and from $p$ in graphs $H\cup E(\pi[s_l, s^l_h])$ and $H\cup E(\pi[t^l_h, t_l])$.
\end{itemize}

\paragraph{Sub-path shortcuts for each $V(\gamma_l)$.} We need to compute all backward distances: $$\mu(w, z) \overset{\text{def}}{=} \dist(w, z, H\cup E(\pi[w, t_l]))$$ for any $w, z\in V(\gamma_l)$ such that $z$ lies between $s$ and $w$ on $\pi$, for all $1\leq l\leq \ceil{n/L}$.

We will perform truncated Dijkstra up to depth $L$ on vertices in $\pi$, which is similar to what we did in previous sections.

For any offset $1\leq b\leq 10$, initialize a vertex subset $A_b\leftarrow V$, and go over all the sub-paths $\gamma_b, \gamma_{b+10}, \cdots, \gamma_{b+10k}, k \leq \ceil{n/10L}$ sequentially. For each sub-paths $\gamma_{b+10j}$, enumerate all vertices $w\in V(\gamma_{b+10j})$ and perform a truncated Dijkstra up to depth $2L$ in the graph $H[A_b]\cup E(\pi[w, t_{b+10j}])$. Then, for each $z\in V(\pi[s_{b+10j}, w])$, if $z$ was visited by the truncated Dijkstra, then store a distance value $\mu(w, z)\leftarrow \dist(w, z, H[A_b]\cup E(\pi[w, t_{b+10j}]))$; otherwise if $z$ was not visited, then go over all vertices $v\in V(\pi[w, t_{b+10j}]), p\in U$, and store a distance value: $$\mu(w, z)\leftarrow \min_{v\in V(\pi[w, t_{b+10j}]), p\in U}\{ |\pi[w, v]| + \dist(v, p, H) + \dist(p, z, H) \}$$
After that, let $P_{b+10j}$ collect all the vertices visited by any truncated Dijkstra of vertices in $\gamma_{b+10j}$, and prune $A_{b}\leftarrow A_{b}\setminus P_{b+10j}$. See \Cref{2fail-case2-shortcut1} for an illustration.

\begin{figure}
	\centering
	\begin{tikzpicture}[thick,scale=0.7]
	\draw (-11, 0) node(1)[circle, draw, fill=black, inner sep=0pt, minimum width=6pt, label=180: {$s$}] {};
	\draw (11, 0) node(2)[circle, draw, fill=black, inner sep=0pt, minimum width=6pt, label=0: {$t$}] {};
	
	\draw (-6, 0) node(3)[circle, draw, fill=black, inner sep=0pt, minimum width=6pt, label=-90: {$s_l$}] {};
	\draw (6, 0) node(4)[circle, draw, fill=black, inner sep=0pt, minimum width=6pt, label=-90: {$t_l$}] {};
	
	\draw (-3, 0) node(5)[circle, draw, fill=black, inner sep=0pt, minimum width=6pt, label=-90: {$z$}] {};
	\draw (0, 0) node(6)[circle, draw, fill=black, inner sep=0pt, minimum width=6pt, label=-90: {$w$}] {};
	\draw (3, 0) node(7)[circle, draw, fill=black, inner sep=0pt, minimum width=6pt, label=-90: {$v$}] {};
	\draw (0, 4) node(8)[circle, draw, fill=black, inner sep=0pt, minimum width=6pt, label=-90: {$p$}] {};
	
	\begin{scope}[on background layer]
		\draw [->, line width = 0.5mm, dotted] (1) to (3);
		\draw [->, line width = 0.5mm, dotted] (4) to (2);
		\draw [->, line width = 0.5mm] (3) to (4);
		\draw [->, line width = 0.5mm, color=orange] (6) to (7);
		\draw [->, line width = 0.5mm, color=orange] (7) to[out=90, in=0] (8);
		\draw [->, line width = 0.5mm, color=orange] (8) to[out=180, in=90] (5);

	\end{scope}
\end{tikzpicture}
	\caption{A typical sub-path shortcut (drawn as the orange curve) from $w$ to $z$ uses some edges on $\pi[w, t_l]$, and then travels to $z$ in $H$.}
	\label{2fail-case2-shortcut1}
\end{figure}
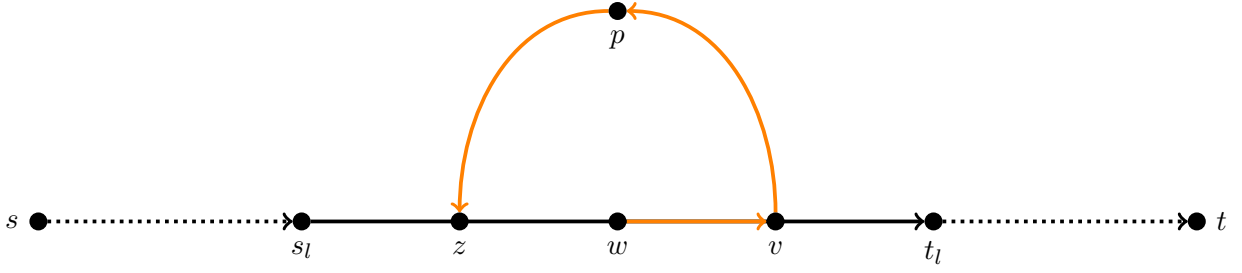

\paragraph{Backward shortcuts.} Next, we need to add some backward shortcuts from each vertex $u\in V(\gamma_l)$ to $v\in V(\gamma_h)$ which are weighted edges $(u, v)$ with weight $\mu_0(u, v)$, where $l-h\geq 2$. Each of the shortcut edge $(u, v)$, if added, will correspond to a path from $u$ to $v$ in graph $H\cup E(\pi[t_{h}, u])$. See \Cref{2fail-case2-shortcut2} for an illustration.

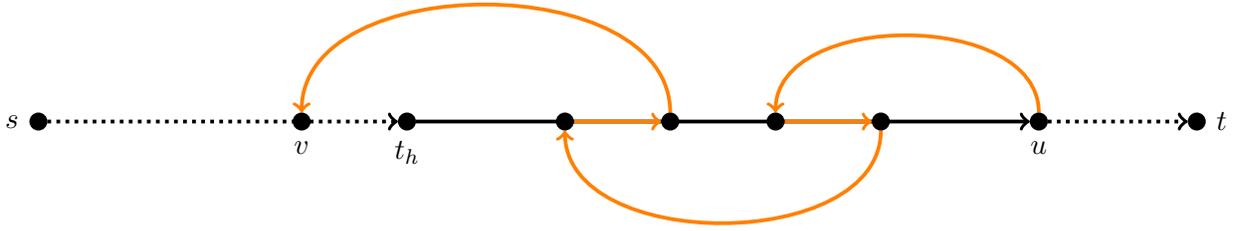
\begin{figure}
	\centering
	\begin{tikzpicture}[thick,scale=0.7]
	\draw (-11, 0) node(1)[circle, draw, fill=black, inner sep=0pt, minimum width=6pt, label=180: {$s$}] {};
	\draw (11, 0) node(2)[circle, draw, fill=black, inner sep=0pt, minimum width=6pt, label=0: {$t$}] {};
	
	\draw (-4, 0) node(3)[circle, draw, fill=black, inner sep=0pt, minimum width=6pt, label=-90: {$t_{h}$}] {};
	\draw (8, 0) node(4)[circle, draw, fill=black, inner sep=0pt, minimum width=6pt, label=-90: {$u$}] {};
	\draw (-6, 0) node(5)[circle, draw, fill=black, inner sep=0pt, minimum width=6pt, label=-90: {$v$}] {};
	
	\draw (-1, 0) node(6)[circle, draw, fill=black, inner sep=0pt, minimum width=6pt] {};
	\draw (1, 0) node(7)[circle, draw, fill=black, inner sep=0pt, minimum width=6pt] {};
	\draw (3, 0) node(8)[circle, draw, fill=black, inner sep=0pt, minimum width=6pt] {};
	\draw (5, 0) node(9)[circle, draw, fill=black, inner sep=0pt, minimum width=6pt] {};
	
	\begin{scope}[on background layer]
		\draw [->, line width = 0.5mm, dotted] (1) to (3);
		\draw [->, line width = 0.5mm, dotted] (4) to (2);
		\draw [->, line width = 0.5mm] (3) to (4);

		\draw [->, line width = 0.5mm, color = orange] (4) to[out=90, in=90] (8);
		\draw [->, line width = 0.5mm, color = orange] (8) to (9);
		\draw [->, line width = 0.5mm, color = orange] (9) to[out=-90, in=-90] (6);
		\draw [->, line width = 0.5mm, color = orange] (6) to (7);
		\draw [->, line width = 0.5mm, color = orange] (7) to[out=90, in=90] (5);
	\end{scope}
\end{tikzpicture}
	\caption{A typical backward shortcut (drawn as the orange curve) corresponds to a path in $H\cup E(\pi[t_{h}, u])$.}
	\label{2fail-case2-shortcut2}
\end{figure}

These shortcuts are computed as following. Initially all values of $\mu_0$ are infinity. Then, for each vertex $p\in U$ and every pair of $h, l$ with $l - h\geq 2$, and every edge $f\in E(\gamma_l)$, let $T^h_{p, f}$ be the shortest path tree rooted at $p$ in graph $G\setminus \brac{E(\pi[s, t_h])\cup \{f\}}$; recall that we have computed all the shortest paths trees by \Cref{ssrp}.

Then, for each vertex $v\in V(\gamma_h)$, find its ancestor $w$ in tree $T^h_{p, f}$ such that $w\in V(\pi)$ and $|\pi[w, t]|$ is minimized; this operation can be done in $O(\log^2n)$ time using classical tree data structures. If $w$ lies between $s_l$ and $f$ on $\pi$, then update a shortcut edge $(w, v)$ with edge weight $\mu_0(w, v)$ equal to the tree distance in $T^h_{p, f}$; namely, assign: $$\mu_0(w, v)\leftarrow\min\left\{\mu_0(w, v), \dist\brac{w, v, T^h_{p, f}} \right\}$$
See \Cref{2fail-case2-shortcut3} for an illustration.

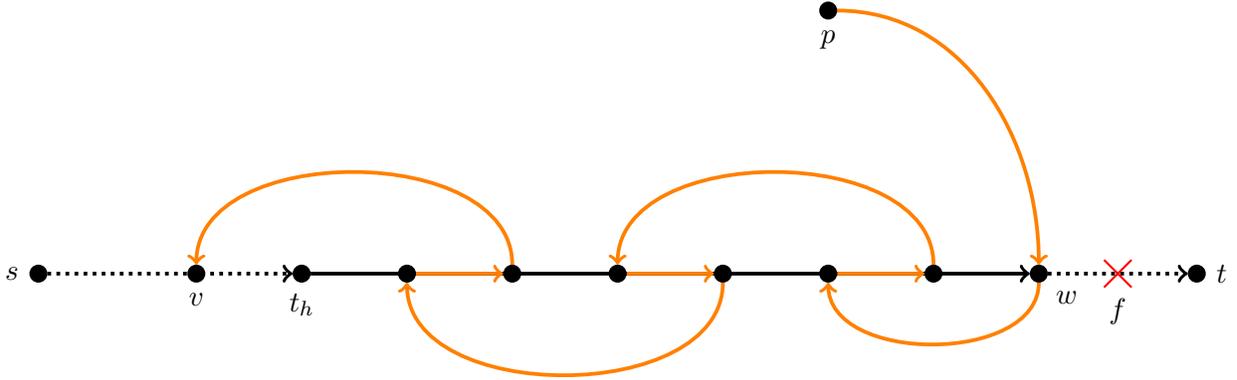
\begin{figure}
	\centering
	\begin{tikzpicture}[thick,scale=0.7]
	\draw (-11, 0) node(1)[circle, draw, fill=black, inner sep=0pt, minimum width=6pt, label=180: {$s$}] {};
	\draw (11, 0) node(2)[circle, draw, fill=black, inner sep=0pt, minimum width=6pt, label=0: {$t$}] {};
	
	\draw (4, 5) node(3)[circle, draw, fill=black, inner sep=0pt, minimum width=6pt, label=-90: {$p$}] {};
	\draw (-8, 0) node(4)[circle, draw, fill=black, inner sep=0pt, minimum width=6pt, label=-90: {$v$}] {};
	\draw (-6, 0) node(5)[circle, draw, fill=black, inner sep=0pt, minimum width=6pt, label=-90: {$t_h$}] {};
	\draw (8, 0) node(6)[circle, draw, fill=black, inner sep=0pt, minimum width=6pt, label=-45: {$w$}] {};

	\draw (-4, 0) node(7)[circle, draw, fill=black, inner sep=0pt, minimum width=6pt] {};	
	\draw (-2, 0) node(8)[circle, draw, fill=black, inner sep=0pt, minimum width=6pt] {};
	\draw (0, 0) node(9)[circle, draw, fill=black, inner sep=0pt, minimum width=6pt] {};
	\draw (2, 0) node(10)[circle, draw, fill=black, inner sep=0pt, minimum width=6pt] {};
	\draw (4, 0) node(11)[circle, draw, fill=black, inner sep=0pt, minimum width=6pt] {};
	\draw (6, 0) node(12)[circle, draw, fill=black, inner sep=0pt, minimum width=6pt] {};
		
	\draw (9.5, 0) node[cross=6, red, label=-90: {$f$}] {};
	
	\begin{scope}[on background layer]
		\draw [->, line width = 0.5mm, dotted] (1) to (5);
		\draw [->, line width = 0.5mm] (5) to (6);
		\draw [->, line width = 0.5mm, dotted] (6) to (2);

		\draw [->, line width = 0.5mm, color = orange] (3) to[out=0, in=90] (6);
		
		\draw [->, line width = 0.5mm, color = orange] (7) to (8);
		\draw [->, line width = 0.5mm, color = orange] (9) to (10);
		\draw [->, line width = 0.5mm, color = orange] (11) to (12);

		\draw [->, line width = 0.5mm, color = orange] (6) to[out=-90, in=-90] (11);
		\draw [->, line width = 0.5mm, color = orange] (12) to[out=90, in=90] (9);
		\draw [->, line width = 0.5mm, color = orange] (10) to[out=-90, in=-90] (7);
		\draw [->, line width = 0.5mm, color = orange] (8) to[out=90, in=90] (4);
		
	\end{scope}
\end{tikzpicture}
	\caption{The value of $\mu_0(w, v)$ corresponds to a tree path of $T^h_{p, f}$ drawn as the orange curve.}
	\label{2fail-case2-shortcut3}
\end{figure}

For every pair of $h, l$ with $l - h\geq 2$, and every edge $f\in E(\gamma_l)$, let $T^h_f$ be the shortest path tree rooted at $t_h$ in graph $G\setminus \brac{E(\pi[s, t_h])\cup \{f\}}$. Then, go over each vertex $v\in V(\gamma_h)$, find its ancestor $w$ in tree $T^h_f$ such that $w\in V(\pi)$ and $|\pi[w, t]|$ is minimized. If $w$ lies between $s$ and $f$ on $\pi$, then for each vertex $u\in V(\pi[w, f))\cap V(\gamma_l)$, store a weight value: $$\mu_0(t_h, v \mid u) \leftarrow \min\left\{\mu_0(t_h, v \mid u), \dist\brac{t_h, v, T^h_f} \right\}$$
Intuitively, $\mu_0(t_h, v \mid u)$ corresponds to a path from $t_h$ to $v$ in $G\setminus E(\pi[s, t_h])$ that does not use any edges on $\pi[u, t]$. See \Cref{2fail-case2-shortcut4} for an illustration.

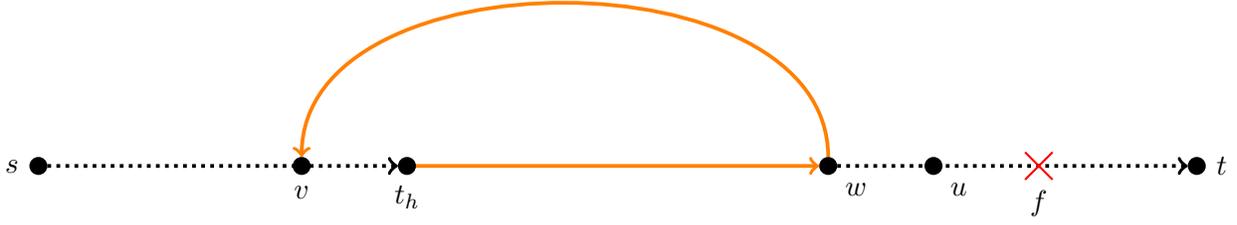
\begin{figure}
	\centering
	\begin{tikzpicture}[thick,scale=0.7]
	\draw (-11, 0) node(1)[circle, draw, fill=black, inner sep=0pt, minimum width=6pt, label=180: {$s$}] {};
	\draw (11, 0) node(2)[circle, draw, fill=black, inner sep=0pt, minimum width=6pt, label=0: {$t$}] {};
	
	\draw (-6, 0) node(4)[circle, draw, fill=black, inner sep=0pt, minimum width=6pt, label=-90: {$v$}] {};
	\draw (-4, 0) node(5)[circle, draw, fill=black, inner sep=0pt, minimum width=6pt, label=-90: {$t_h$}] {};
	\draw (4, 0) node(6)[circle, draw, fill=black, inner sep=0pt, minimum width=6pt, label=-45: {$w$}] {};
	\draw (6, 0) node(7)[circle, draw, fill=black, inner sep=0pt, minimum width=6pt, label=-45: {$u$}] {};
		
	\draw (8, 0) node[cross=6, red, label=-90: {$f$}] {};
	
	\begin{scope}[on background layer]
		\draw [->, line width = 0.5mm, dotted] (1) to (5);
		\draw [->, line width = 0.5mm, dotted] (6) to (2);

		\draw [->, line width = 0.5mm, color=orange] (5) to (6);
		\draw [->, line width = 0.5mm, color=orange] (6) to[out=90, in=90] (4);
	\end{scope}
\end{tikzpicture}
	\caption{The value of $\mu_0(t_h, v\mid u)$ corresponds to a tree path in $T^h_f$ drawn as the orange curve.}
	\label{2fail-case2-shortcut4}
\end{figure}

Once we have computed all the edge weights $\mu_0$, we will augment these shortcuts with sub-path shortcuts which were computed previously. Intuitively, for each pair of $w\in V(\gamma_l), v\in V(\gamma_h)$, we want a shortcut of weight $\mu(w, v)$ from $w$ to $v$ which corresponds to a path in graph $H\cup E(\pi[v, w])$, rather than the more restrictive graph $H\cup E(\pi[t_h, w])$. These values of $\mu(v, w)$ is computed by the following dynamic programming. Fix any vertex $w\in V(\gamma_l)$ and go over all vertices in $v\in V(\gamma_h)$ in the reverse order from $t_h$ to $s_h$, and update the value:
$$\mu(w, v) \leftarrow \min_{z\in V(\pi(v, t_h])}\{\mu_0(w, v), \mu(w, z) + \mu(z, v), \mu(w, z) + |\pi[z, t_h]| + \mu_0(t_h, v \mid w) \}$$

Finally, for each pair of indices $h\leq l-2$ and vertices $u\in V(\gamma_h), v\in V(\gamma_l)$, we need to compute one more kind of backward shortcut value $\mu(s_l, t_h\mid u, v)$ which is the length of a path from $s_l$ to $t_h$ in graph $H\cup \pi[u, v]$ using shortcuts $\mu(*, *)$. This can be calculated by the formula:
$$\mu(s_l, t_h\mid u, v) = \min_{x\in V(\pi[u, t_h], y\in V(\pi[s_l, v]))}\{|\pi[s_l, y]| + \mu(y, x) + |\pi[x, s_l]| \}$$

\paragraph{Sub-case 2(a): $x_{i, j}\in V(\rho_{i, j}[s, y_{i, j}])$.} In this case, $\rho_{i, j}$ contains the entire sub-path $\pi[t_{l_i}, s_{l_j}]$. Apply \Cref{ssrp} on $s$ in graph $H\cup E(\pi[s, t_{l_i}])$ and on $t$ in graph $H\cup E(\pi[s_{l_j}, t])$. Then, update the estimation $\est(s, t, G\setminus\{e_i, e_j\})$ with the following quantity:
$$\min\{\dist\left(s, t_{l_i}, H\cup E(\pi[s, t_{l_i}])\setminus \{e_i\}\right) + \dist(t_{l_i}, s_{l_j}, G) + \dist\left(s_{l_j}, t, H\cup E(\pi[s_{l_j}, t])\setminus \{e_j\}\right) \}$$

For the rest three sub-cases, we assume $x_{i, j}\notin V(\rho_{i, j}[s, y_{i, j}])$; that is, $y_{i, j}$ lies between $s$ and $x_{i, j}$ on $\rho_{i, j}$.

\paragraph{Sub-case 2(b): $x_{i, j}\in V\brac{\alpha^{l_i}_{h_i}}$ and $y_{i, j}\in V\brac{\alpha^{l_j}_{h_j}}$.} To compute $|\rho_{i, j}|$, let us build the following shortcut graph $H_{i, j}$ with edge weight $\wts(*, *)$.
\begin{itemize}[leftmargin=*]
	\item \textbf{Vertices.} Add vertices $\{s, t\}\cup\{t_{l_i}, s_{l_j}\}\cup U\cup V\brac{\alpha^{l_i}_{h_i}}\cup V\brac{\alpha^{l_j}_{h_j}}$ to $H_{i, j}$.
	
	\item \textbf{Edges.} Add the following types of edges.
	\begin{enumerate}[(i),leftmargin=*]
		\item Edges on the edges $E\brac{\alpha^{l_i}_{h_i}}\cup E\brac{\alpha^{l_j}_{h_j}}\setminus \{e_i, e_j\}$, and shortcuts $\brac{s, s^{l_i}_{h_i}}, \brac{t^{l_j}_{h_j}, t}$ with weight $|\pi[s, s^{l_i}_{h_i}]|$ and $|\pi[t^{l_j}_{h_j}, t]|$, respectively.

		\item Add edge $(s, s_{l_j})$ with weight $\wts(s, s_{l_j}) = \dist(s, s_{l_j}, G\setminus (E(\pi[s_{l_j}, t])\cup \{e_i\}))$, and edge $(t_{l_i}, t)$ with weight $\wts(t_{l_i}, t) = \dist(t_{l_i}, t, G\setminus (E(\pi[s, t_{l_i}])\cup \{e_j\}))$.
		
		Also, add edge $\brac{s_{l_j}, s^{l_j}_{h_j}}$ with weight $\wts\brac{s_{l_j}, s^{l_j}_{h_j}} = |\pi[s_{l_j}, s^{l_j}_{h_j}]|$, and edge $\brac{t^{l_i}_{h_i}, t_{l_i}}$ with weight $\wts\brac{t^{l_i}_{h_i}, t_{l_i}} = |\pi[t^{l_i}_{h_i}, t_{l_i}]|$.
		
		Furthermore, for each $p\in U$, add edge $\brac{p, s^{l_j}_{h_j}}$ with weight $$\wts\brac{p, s^{l_j}_{h_j}} = \dist\left(p, s^{l_j}_{h_j}, H\cup E(\pi[s_{l_j}, s^{l_j}_{h_j}])\right)$$
		and edge $\brac{t^{l_i}_{h_i}, p}$ with weight $$\wts\brac{t^{l_i}_{h_i}, p} = \dist\left(t^{l_i}_{h_i}, p, H\cup E(\pi[t^{l_i}_{h_i}, p])\right)$$
		
		\item For each pivot vertex $p\in U$ and vertices $u\in \{s\}\cup V\brac{\alpha^{l_i}_{h_i}}, v\in \{t\}\cup V\brac{\alpha^{l_j}_{h_j}}$, add edge $(u, p)$ with weight $$\wts(u, p) = \dist\left(u, p, H\cup E(\pi[s, s^{l_i}_{h_i}])\right)$$
		and edge $(p, v)$ with weight $$\wts(p, v) = \dist\left(p, v, H\cup E(\pi[t^{l_j}_{h_j}, t])\right)$$
		
		\item For each $u\in V(\pi[u_{i+1}, t^{l_i}_{h_i}])$ and $v\in V(\pi[s^{l_j}_{h_j}, u_j])$, add edge $(v, u)$ with weight $\wts(v, u) = \mu(v, u)$.
	\end{enumerate}	
	
\end{itemize}

Then, apply Dijkstra's algorithm in $H_{i, j}$ to compute the shortest path from $s$ to $t$, and update $$\est(s, t, G\setminus\{e_i, e_j\})\leftarrow \min\{\est(s, t, G\setminus\{e_i, e_j\}), \dist(s, t, H_{i, j}) \}$$
See \Cref{2fail-case2-shortcut5} for an illustration.

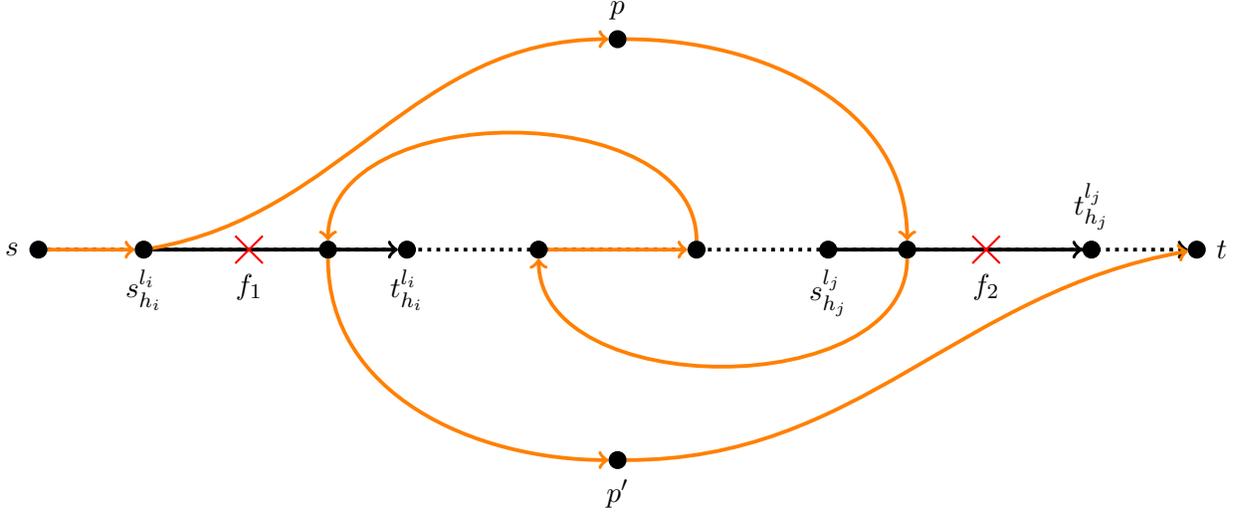
\begin{figure}
	\centering
	\begin{tikzpicture}[thick,scale=0.7]
	\draw (-11, 0) node(1)[circle, draw, fill=black, inner sep=0pt, minimum width=6pt, label=180: {$s$}] {};
	\draw (11, 0) node(2)[circle, draw, fill=black, inner sep=0pt, minimum width=6pt, label=0: {$t$}] {};
	
	\draw (-9, 0) node(3)[circle, draw, fill=black, inner sep=0pt, minimum width=6pt, label=-90: {$s_{h_i}^{l_i}$}] {};
	\draw (-4, 0) node(4)[circle, draw, fill=black, inner sep=0pt, minimum width=6pt, label=-90: {$t^{l_i}_{h_i}$}] {};
	\draw (4, 0) node(5)[circle, draw, fill=black, inner sep=0pt, minimum width=6pt, label=-90: {$s^{l_j}_{h_j}$}] {};
	\draw (9, 0) node(6)[circle, draw, fill=black, inner sep=0pt, minimum width=6pt, label=90: {$t^{l_j}_{h_j}$}] {};
	\draw (-5.5, 0) node(7)[circle, draw, fill=black, inner sep=0pt, minimum width=6pt] {};
	\draw (5.5, 0) node(8)[circle, draw, fill=black, inner sep=0pt, minimum width=6pt] {};
	\draw (-1.5, 0) node(9)[circle, draw, fill=black, inner sep=0pt, minimum width=6pt] {};
	\draw (1.5, 0) node(10)[circle, draw, fill=black, inner sep=0pt, minimum width=6pt] {};
	\draw (0, 4) node(11)[circle, draw, fill=black, inner sep=0pt, minimum width=6pt, label=90: {$p$}] {};
	\draw (0, -4) node(12)[circle, draw, fill=black, inner sep=0pt, minimum width=6pt, label=-90: {$p^\prime$}] {};
	
	\draw (-7, 0) node[cross=6, red, label=-90: {$f_1$}] {};
	\draw (7, 0) node[cross=6, red, label=-90: {$f_2$}] {};
	
	\begin{scope}[on background layer]
		\draw [->, line width = 0.5mm, dotted] (1) to (2);
		\draw [->, line width = 0.5mm] (3) to (4);
		\draw [->, line width = 0.5mm] (5) to (6);

		\draw [->, line width = 0.5mm, color=orange] (1) to (3);
		\draw [->, line width = 0.5mm, color=orange] (3) to[out=10, in=180] (11);
		\draw [->, line width = 0.5mm, color=orange] (11) to[out=0, in=90] (8);
		\draw [->, line width = 0.5mm, color=orange] (8) to[out=-90, in=-90] (9);
		\draw [->, line width = 0.5mm, color=orange] (9) to (10);
		\draw [->, line width = 0.5mm, color=orange] (10) to[out=90, in=90] (7);
		\draw [->, line width = 0.5mm, color=orange] (7) to[out=-90, in=180] (12);
		\draw [->, line width = 0.5mm, color=orange] (12) to[out=0, in=-170] (2);

	\end{scope}
\end{tikzpicture}
	\caption{The shortcut edges in $H_{i, j}$ captures the shortest replacement path $\rho_{i, j}$ drawn as the orange curve.}
	\label{2fail-case2-shortcut5}
\end{figure}

\paragraph{Sub-case 2(c): $x_{i, j}\notin V\brac{\alpha^{l_i}_{h_i}}$ and $y_{i, j}\in V\brac{\alpha^{l_j}_{h_j}}$.} The other symmetric case where $x_{i, j}\in V\brac{\alpha^{l_i}_{h_i}}$ and $y_{i, j}\notin V\brac{\alpha^{l_j}_{h_k}}$ will be handled in a similar manner. In this case, we need to compute some common information that will be useful for all presumable edge failures $e_i$ on segment $\alpha^{l_i}_{h_i}$. Build the following shortcut graph $H^{(l_i, h_i)}_j$ with edge weight $\wts(*, *)$.
\begin{itemize}[leftmargin=*]
	\item \textbf{Vertices.} Add $\{s_{l_j}, t\}\cup U\cup V(\gamma_{l_i})\cup V\brac{\alpha^{l_j}_{h_j}}$ to $H^{(l_i, h_i)}_j$.
	
	\item \textbf{Edges.} Add the following types of edges.
	\begin{enumerate}[(i),leftmargin=*]
		\item Add edges $E\brac{\gamma_{l_i}\setminus \alpha^{l_i}_{h_i}}\cup E\brac{\alpha^{l_j}_{h_j}}\setminus \{e_j\}$ to $H^{(l_i, h_i)}_j$, and shortcuts $\brac{s, s^{l_i}_{h_i}}, \brac{t^{l_j}_{h_j}, t}$ with weight $|\pi[s, s^{l_i}_{h_i}]|$ and $|\pi[t^{l_j}_{h_j}, t]|$, respectively.
		
		\item Add edge $\brac{s_{l_j}, s^{l_j}_{h_i}}$ with weight $\wts\brac{s_{l_j}, s^{l_j}_{h_i}} = |\pi[s_{l_j}, s^{l_j}_{h_i}]|$. Add edge $(t_{l_i}, t)$ with weight $\wts(t_{l_i}, t) = \dist(t_{l_i}, t, G\setminus (E(\pi[s, t_{l_i}])\cup \{e_j\}))$. 
		
		Also, add edge $\brac{s_{l_j}, s^{l_j}_{h_j}}$ with weight $\wts\brac{s_{l_j}, s^{l_j}_{h_j}} = |\pi[s_{l_j}, s^{l_j}_{h_j}]|$.
		
		Furthermore, for each $p\in U$, add edge $\brac{p, s^{l_j}_{h_j}}$ with weight $$\wts\brac{p, s^{l_j}_{h_j}} = \dist\left(p, s^{l_j}_{h_j}, H\cup E(\pi[s_{l_j}, s^{l_j}_{h_j}])\right)$$
		
		\item For each pivot vertex $p\in U$ and vertices $u\in V(\gamma_{l_i}), v\in \{t\}\cup V\brac{\alpha^{l_j}_{h_j}}$, add edge $(u, p)$ with weight $$\wts(u, p) = \dist\left(u, p, H\cup E\brac{\pi[s, s^{l_i}_{h_i}]}\right)$$
		and edge $(p, v)$ with weight $$\wts(p, v) = \dist\left(p, v, H\cup E\brac{\pi[t^{l_j}_{h_j}, t]}\right)$$
		
		\item For each $u\in V\brac{\pi[t^{l_i}_{h_i}, t_{l_i}]}$ and $v\in V\brac{\pi[s^{l_j}_{h_j}, u_j]}$, add edge $(v, u)$ with weight $\wts(v, u) = \mu(v, u)$.
	\end{enumerate}
\end{itemize}

Then, apply Dijkstra's algorithm in $H^{(l_i, h_i)}_j$ to compute shortest paths from $\{s_{l_j}\}\cup U$ to $t$. After that, for each edge $e_i\in E(\alpha^{l_i}_{h_i})$, and update $\est(s, t, G\setminus\{e_i, e_j\})$ with minimum of the following two quantities:
$$\dist\brac{s, s_{l_j}, G\setminus\brac{E(\pi[s_{l_j}, t]\cup\{e_i\})}} + \dist\brac{s_{l_j}, t, H^{(l_i, h_i)}_j}$$
$$\min_{p\in U, v\in V\brac{\pi[s^{l_i}_{h_i}, u_i]}}\left\{\min\left\{|\pi[s, v]| + \dist(v, p, H), \dist(s, p, G\setminus (E(\pi[s^{l_i}_{h_i}, t])))\right\} + \dist\brac{p, t, H^{(l_i, h_i)}_j}\right\}$$
See \Cref{2fail-case2-shortcut6} for an illustration.

\begin{figure}
	\centering
	\begin{tikzpicture}[thick,scale=0.7]
	\draw (-11, 0) node(1)[circle, draw, fill=black, inner sep=0pt, minimum width=6pt, label=180: {$s$}] {};
	\draw (11, 0) node(2)[circle, draw, fill=black, inner sep=0pt, minimum width=6pt, label=0: {$t$}] {};
	
	\draw (-9, 0) node(3)[circle, draw, fill=black, inner sep=0pt, minimum width=6pt, label=-90: {$s_{h_i}^{l_i}$}] {};
	\draw (-6, 0) node(4)[circle, draw, fill=black, inner sep=0pt, minimum width=6pt, label=-90: {$t^{l_i}_{h_i}$}] {};
	\draw (4, 0) node(5)[circle, draw, fill=black, inner sep=0pt, minimum width=6pt, label=-90: {$s^{l_j}_{h_j}$}] {};
	\draw (9, 0) node(6)[circle, draw, fill=black, inner sep=0pt, minimum width=6pt, label=90: {$t^{l_j}_{h_j}$}] {};
	\draw (-4.5, 0) node(7)[circle, draw, fill=black, inner sep=0pt, minimum width=6pt] {};
	\draw (5.5, 0) node(8)[circle, draw, fill=black, inner sep=0pt, minimum width=6pt] {};
	\draw (-1.5, 0) node(9)[circle, draw, fill=black, inner sep=0pt, minimum width=6pt] {};
	\draw (1.5, 0) node(10)[circle, draw, fill=black, inner sep=0pt, minimum width=6pt] {};
	\draw (0, 4) node(11)[circle, draw, fill=black, inner sep=0pt, minimum width=6pt, label=90: {$p$}] {};
	\draw (0, -4) node(12)[circle, draw, fill=black, inner sep=0pt, minimum width=6pt, label=-90: {$p^\prime$}] {};
	
	\draw (-3, 0) node(13)[circle, draw, fill=black, inner sep=0pt, minimum width=6pt, label=-90: {$t_{l_i}$}] {};
	\draw (7, 0) node[cross=6, red, label=-90: {$f_2$}] {};
	
	\begin{scope}[on background layer]
		\draw [->, line width = 0.5mm, dotted] (1) to (2);
		\draw [->, line width = 0.5mm, color=red] (3) to (4);
		\draw [->, line width = 0.5mm] (4) to (13);
		\draw [->, line width = 0.5mm] (5) to (6);

		\draw [->, line width = 0.5mm, color=orange, dotted] (1) to[out=10, in=180] (11);
		\draw [->, line width = 0.5mm, color=orange] (11) to[out=0, in=90] (8);
		\draw [->, line width = 0.5mm, color=orange] (8) to[out=-90, in=-90] (9);
		\draw [->, line width = 0.5mm, color=orange] (9) to (10);
		\draw [->, line width = 0.5mm, color=orange] (10) to[out=90, in=90] (7);
		\draw [->, line width = 0.5mm, color=orange] (7) to[out=-90, in=180] (12);
		\draw [->, line width = 0.5mm, color=orange] (12) to[out=0, in=-170] (2);

	\end{scope}
\end{tikzpicture}
	\caption{As a typical example, the shortcut edges in $H^{(l_i, h_i)}_{j}$ captures the suffix of the shortest replacement path $\rho_{i, j}[p, t]$ drawn as the orange curve. The dotted orange curve will be computed once we are given $e_i$.}
	\label{2fail-case2-shortcut6}
\end{figure}
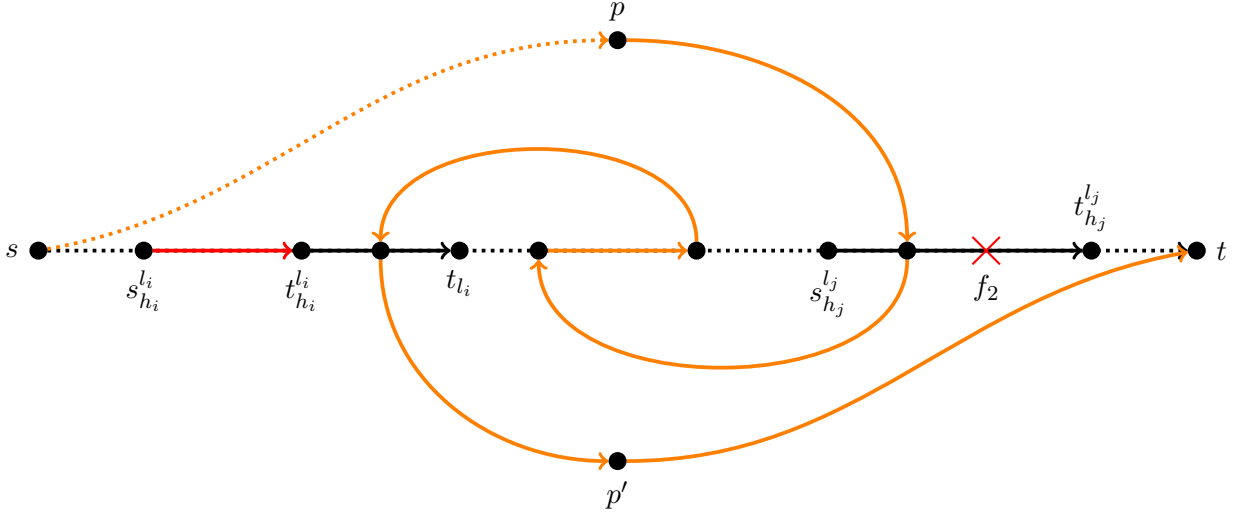

\paragraph{Sub-case 2(d): $x_{i, j}\notin V(\alpha^{l_i}_{h_i})$ and $y_{i, j}\notin V(\alpha^{l_j}_{h_j})$}
Construct the following shortcut graph $H^{(l_i, h_i), (l_j, h_j)}$ with edge weight $\wts(*, *)$.
\begin{itemize}[leftmargin=*]
	\item \textbf{Vertices.} Add $\{s, t\}\cup U\cup V(\gamma_{l_i})\cup V(\gamma_{l_j})$ to $H^{(l_i, h_i), (l_j, h_j)}$.
	\item \textbf{Edges.} Add the following types of edges.
	\begin{enumerate}[(i),leftmargin=*]
		\item Add edges $E\brac{\gamma_{l_i}\setminus \alpha^{l_i}_{h_i}}\cup E\brac{\gamma_{l_j}\setminus \alpha^{l_j}_{h_j}}$ to $H^{(l_i, h_i), (l_j, h_j)}$.
		
		\item For each pivot vertex $p\in U$ and vertices $u\in \{s\}\cup V(\gamma_{l_i}), v\in \{t\}\cup V(\gamma_{l_j})$, add edge $(u, p)$ with weight $$\wts(u, p) = \dist\left(u, p, H\cup E(\pi[s, s^{l_i}_{h_i}])\right)$$
		and edge $(p, v)$ with weight $$\wts(p, v) = \dist\left(p, v, H\cup E(\pi[t^{l_j}_{h_j}, t])\right)$$
		
		\item For each $u\in V\brac{\pi[t^{l_i}_{h_i}, t_{l_i}]}$ and $v\in V\brac{\pi[s_{l_j}, s^{l_j}_{h_j}]}$, add an edge $(v, u)$ with weight $\wts(v, u) = \mu(v, u)$.
	\end{enumerate}
\end{itemize}

For each vertex $p\in \{s, t\}\cup \{t_{l_i}, s_{l_j}\}\cup U$, compute single-source shortest paths to and from $p$ in $H^{(l_i, h_i), (l_j, h_j)}$. See \Cref{2fail-case2-shortcut7} for an illustration.

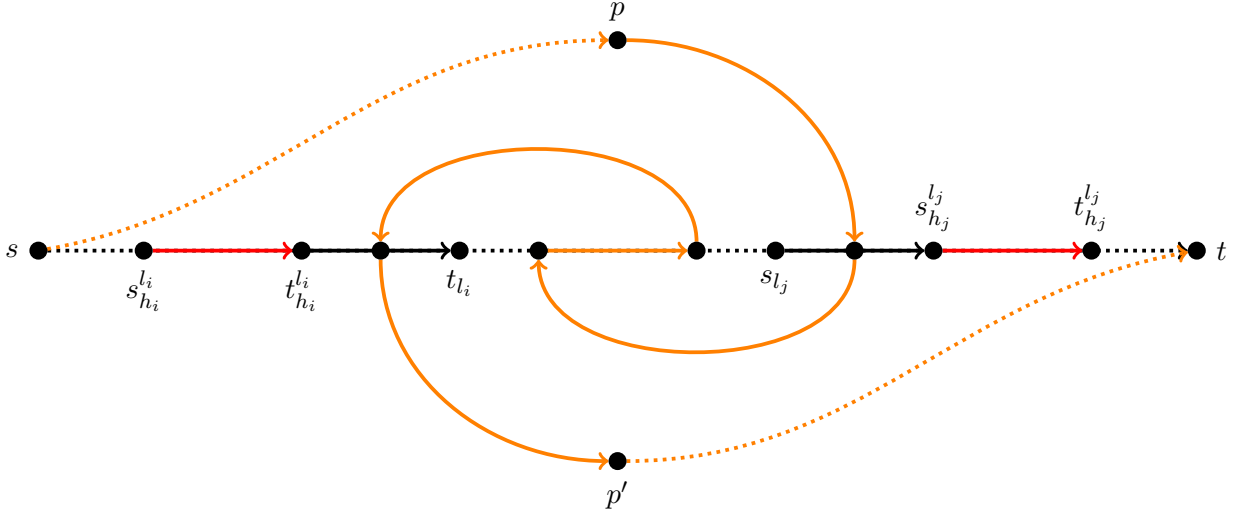
\begin{figure}
	\centering
	\begin{tikzpicture}[thick,scale=0.7]
	\draw (-11, 0) node(1)[circle, draw, fill=black, inner sep=0pt, minimum width=6pt, label=180: {$s$}] {};
	\draw (11, 0) node(2)[circle, draw, fill=black, inner sep=0pt, minimum width=6pt, label=0: {$t$}] {};
	
	\draw (-9, 0) node(3)[circle, draw, fill=black, inner sep=0pt, minimum width=6pt, label=-90: {$s_{h_i}^{l_i}$}] {};
	\draw (-6, 0) node(4)[circle, draw, fill=black, inner sep=0pt, minimum width=6pt, label=-90: {$t^{l_i}_{h_i}$}] {};
	\draw (6, 0) node(5)[circle, draw, fill=black, inner sep=0pt, minimum width=6pt, label=90: {$s^{l_j}_{h_j}$}] {};
	\draw (9, 0) node(6)[circle, draw, fill=black, inner sep=0pt, minimum width=6pt, label=90: {$t^{l_j}_{h_j}$}] {};
	\draw (-4.5, 0) node(7)[circle, draw, fill=black, inner sep=0pt, minimum width=6pt] {};
	\draw (4.5, 0) node(8)[circle, draw, fill=black, inner sep=0pt, minimum width=6pt] {};
	\draw (-1.5, 0) node(9)[circle, draw, fill=black, inner sep=0pt, minimum width=6pt] {};
	\draw (1.5, 0) node(10)[circle, draw, fill=black, inner sep=0pt, minimum width=6pt] {};
	\draw (0, 4) node(11)[circle, draw, fill=black, inner sep=0pt, minimum width=6pt, label=90: {$p$}] {};
	\draw (0, -4) node(12)[circle, draw, fill=black, inner sep=0pt, minimum width=6pt, label=-90: {$p^\prime$}] {};
	
	\draw (-3, 0) node(13)[circle, draw, fill=black, inner sep=0pt, minimum width=6pt, label=-90: {$t_{l_i}$}] {};
	\draw (3, 0) node(14)[circle, draw, fill=black, inner sep=0pt, minimum width=6pt, label=-90: {$s_{l_j}$}] {};
	
	\begin{scope}[on background layer]
		\draw [->, line width = 0.5mm, dotted] (1) to (2);
		\draw [->, line width = 0.5mm, color=red] (3) to (4);
		\draw [->, line width = 0.5mm, color=red] (5) to (6);
		\draw [->, line width = 0.5mm] (4) to (13);
		\draw [->, line width = 0.5mm] (14) to (5);

		\draw [->, line width = 0.5mm, color=orange, dotted] (1) to[out=10, in=180] (11);
		\draw [->, line width = 0.5mm, color=orange] (11) to[out=0, in=90] (8);
		\draw [->, line width = 0.5mm, color=orange] (8) to[out=-90, in=-90] (9);
		\draw [->, line width = 0.5mm, color=orange] (9) to (10);
		\draw [->, line width = 0.5mm, color=orange] (10) to[out=90, in=90] (7);
		\draw [->, line width = 0.5mm, color=orange] (7) to[out=-90, in=180] (12);
		\draw [->, line width = 0.5mm, color=orange, dotted] (12) to[out=0, in=-170] (2);

	\end{scope}
\end{tikzpicture}
	\caption{As a typical example, the shortcut edges in $H^{(l_i, h_i), (l_j, h_j)}$ captures a middle sub-path of  shortest replacement path $\rho_{i, j}[p, p^\prime]$ drawn as the orange curve. The dotted orange curves will be computed using graph $H^{(l_i, h_i), (l_j, h_j)}_{i, j}$ once we are given $e_i, e_j$.}
	\label{2fail-case2-shortcut7}
\end{figure}

After that, for every pair of edge failures $e_i, e_j$, build a shortcut graph based on $H^{(l_i, h_i), (l_j, h_j)}$ with edge weight $\wts(*, *)$.
\begin{itemize}[leftmargin=*]
	\item \textbf{Vertices.} Add vertices $\{s_{l_j}, t_{l_i}\}\cup\{s, t\}\cup U\cup V\brac{\alpha^{l_i}_{h_i}}\cup V\brac{\alpha^{l_j}_{h_j}}$ to $H^{(l_i, h_i), (l_j, h_j)}_{i, j}$.
	
	\item \textbf{Edges.} Add the following edges.
	
	\begin{enumerate}[(i),leftmargin=*]
		\item Add on the paths $E\brac{\pi[s^{l_i}_{h_i}, u_i]}\cup E\brac{\pi[u_{j+1}, t^{l_j}_{h_j}]}$, and shortcuts $\brac{s, s^{l_i}_{h_i}}, \brac{t^{l_j}_{h_j}, t}$ with weight $|\pi[s, s^{l_i}_{h_i}]|$ and $|\pi[t^{l_j}_{h_j}, t]|$, respectively.
		
		\item Add edge $(s, s_{l_j})$ with weight $\wts(s, s_{l_j}) = \dist(s, s_{l_j}, G\setminus (E(\pi[s_{l_j}, t])\cup \{e_i\}))$, and edge $(t_{l_i}, t)$ with weight $\wts(t_{l_i}, t) = \dist(t_{l_i}, t, G\setminus (E(\pi[s, t_{l_i}])\cup \{e_j\}))$. In addition, add a shortcut edge $(s_{l_j}, t_{l_i})$ with weight $\wts(s_{l_j}, t_{l_i}) = \mu(s_{l_j}, t_{l_i}\mid u_{i+1}, u_j)$.

		\item For each pivot vertex $p\in U$ and vertices $u\in \{s\}\cup V\brac{\alpha^{l_i}_{h_i}}, v\in \{t\}\cup V\brac{\alpha^{l_j}_{h_j}}$, add edge $(u, p)$ with weight:
		$$\wts(u, p) = \dist\left(u, p, H\right)$$
		and edge $(p, v)$ with weight:
		$$\wts(p, v) = \dist\left(p, v, H\right)$$
		
		Next, for every pair of vertices $x, y\in \{s_{l_j}, t_{l_i}\}\cup\{s, t\}\cup U$, add a shortcut $(x, y)$ with weight:
		$$\wts(x, y) = \dist\brac{x, y, H^{(l_i, h_i), (l_j, h_j)}}$$
	\end{enumerate} 
\end{itemize}

Then, compute shortest path from $s$ to $t$ in $H^{(l_i, h_i), (l_j, h_j)}_{i, j}$ and update:
$$\est(s, t, G\setminus\{e_i, e_j\})\leftarrow \min\left\{\est(s, t, G\setminus\{e_i, e_j\}), \dist\left(s, t, H^{(l_i, h_i), (l_j, h_j)}_{i, j}\right)\right\}$$

\paragraph{Runtime.} Let us list all the runtime complexities below.
\begin{itemize}[leftmargin=*]
	\item \textbf{Preparation.} By the algorithm description, the runtime is bounded by $\tilde{O}(\frac{n^4}{Lg} + \frac{n^{3.5}}{L})$.
	
	\item \textbf{Sub-path shortcuts.} By the pruning procedure, each vertex is searched by at most $O(L)$ times throughout the truncated Dijkstra procedures. Hence, the total runtime is $\tilde{O}(n^2L)$.
	
	\item \textbf{Backward shortcuts.} Computing all the backward shortcuts $\mu_0(w, v)$ takes time $\tilde{O}(\frac{n^3}{L})$, and computing all the values of $\mu_0(t_h, v\mid u)$ takes time $\tilde{O}(\frac{n^3}{L})$ as well. Computing all entries of $\mu(s_l, t_h\mid u, v)$ can be done in $O(n^2)$ time using dynamic programming.
	
	\item \textbf{Sub-case 2(a).} The runtime of applying \Cref{ssrp} is bounded by $\tilde{O}(\frac{n^{3.5}}{L})$.
	
	\item \textbf{Sub-case 2(b).} For each pair of $e_i, e_j$, the graph $H^{i, j}$ contains at most $\tilde{O}(\frac{ng}{L} + g^2)$ edges, so Dijkstra's algorithm on all $H_{i, j}$ takes time $\tilde{O}(\frac{n^3g}{L} + n^2g^2)$.
	
	\item \textbf{Sub-case 2(c).} Each graph $H^{(l_i, h_i)}_j$ contains at most $\tilde{O}(\frac{ng}{L} + Lg + n)$ edges, so the total runtime of Dijkstra's algorithm is bounded by $\tilde{O}(\frac{n^3}{L} + n^2L + \frac{n^3}{g})$. After that, updating all the entries $\est(s, t, G\setminus\{e_i, e_j\})$ takes time $\tilde{O}(\frac{n^3g}{L})$.
	
	\item \textbf{Sub-case 2(d).} Each graph $H^{(l_i, h_i), (l_j, h_j)}$ has at most $\tilde{O}(\frac{ng}{L} + L^2 + n)$ edges. So the runtime of all Dijkstra instances is bounded by $\tilde{O}(\frac{n^3L}{g^2} + \frac{n^4}{L^2g})$. After that, running Dijkstra's algorithm on $H^{(l_i, h_i), (l_j, h_j)}_{i, j}$ takes total time $\tilde{O}(\frac{n^3g}{L} + \frac{n^4}{L^2})$.
\end{itemize}
To summarize, the overall runtime is bounded by $$\tilde{O}\left(\frac{n^4}{Lg} + \frac{n^{3.5}}{L} + n^2L + \frac{n^3g}{L} + n^2g^2 + \frac{n^3g}{L}\right)$$

\paragraph{Correctness.} First, we we analyze some properties of the shortcuts for each sub-path. 
\begin{claim}
	The algorithm correctly computes the values $$\mu(w, z) = \dist(w, z, H\cup E(\pi[w, t_{b+10j}]))$$ of sub-path shortcuts for each $V(\gamma_{b+10j})$, where $w, z\in V(\gamma_{b+10j})$ and $z$ lies between $s$ and $w$.
\end{claim}
\begin{proof}
	For any $w, z\in V(\gamma_{b+10j})$, first consider the case that $\dist(w, z, H\cup E(\pi[w, t_{b+10j}])) >2L$. Since this shortest path $\beta$ from $w$ to $z$ is unique under weight perturbations, it should first follow the sub-path $\pi[w, t_{b+10j}]$ until a certain vertex $v\in V(\pi[w, t_{b+10j}])$, and then travels to $z$ in graph $H$. As $|\beta| > 2L$, we know that $|\beta[v, z]| > L$. Hence, with high probability, there exists a vertex $p\in U\cap V(\beta[v, z])$. Therefore, we can compute the value of $|\dist(w, z, H\cup E(\pi[w, t_{b+10j}]))|$ via the formula:
	$$\mu(w, z)\leftarrow \min_{v\in V(\pi[w, t_{b+10j}]), p\in U}\{ |\pi[w, v]| + \dist(v, p, H) + \dist(p, z, H) \}$$
	
	Otherwise if $\dist(w, z, H\cup E(\pi[w, t_{b+10j}])) \leq 2L$, then consider the moment when a truncated Dijkstra was executed at $w$ in graph $H[A_b]\cup E(\pi[w, t_{b+10j}])$. If all vertices on the shortest path are in $A_b$, then it would be discovered by the truncated Dijkstra algorithm. Otherwise, there would be a vertex $v\in V(\beta)$ which was visited already by some previous execution of truncated Dijkstra, say a vertex $x$ on sub-path $\gamma_{b+10k}, k<j$. Then, $8L\leq \dist(x, z, G)\leq \dist(x, v, G) + |\beta[v, z]| \leq 4L$, which is a contradiction.
\end{proof}

Next, we analyze some properties of the backward shortcuts.
\begin{claim}
	Let $l - h \geq 2$ be two indices. For any $v\in V(\gamma_h), w\in V(\gamma_l)$, we have $$\mu(w, v)\geq \dist(w, v, H\cup E(\pi[v, w]))$$
\end{claim}
\begin{proof}
	By the algorithm description, each value of $\mu_0(t_h, v \mid u)$ corresponds to a path from $t_h$ to $v$ in $G\setminus \left(E(\pi[s, t_h])\cup E(\pi[u, t])\right)$. Then, by the description of the dynamic programming procedure, $\mu(w, v)$ is equal to some values of  $\mu_0(w, z) + \mu(z, v)$ or $\mu_0(w, t_h) + \mu_0(t_h, v\mid w)$, which always equals the length of some paths from $w$ to $v$ not using any edge on $\pi[w, t]$. Hence, we are guaranteed that $\mu(w, v)\geq \dist(w, v, H\cup E(\pi[v, w]))$.
\end{proof}

\begin{claim}\label{backpath}
	For each replacement path $\rho_{i, j}$ such that $x_{i, j}\notin V(\rho_{i, j}[s, y_{i, j}])$, we have: $$\mu(y_{i, j}, x_{i, j}) = |\rho_{i, j}[y_{i, j}, x_{i, j}]|$$
\end{claim}
\begin{proof}
	Consider the sub-path $\rho_{i, j}[a_{i, j}, y_{i, j}]$. Since $a_{i, j}\in V(\pi[s, u_i])$, $y_{i, j}\in V(\pi[s_{l_j}, u_j])$ and $l_j - l_i \geq 2$, it must be $|\rho_{i, j}[a_{i, j}, y_{i, j}]|\geq L$. Then, with high probability, there exists a vertex $p\in U\cap V(\rho_{i, j}[a_{i, j}, y_{i, j}])\neq \emptyset$.
	
	Let $v$ be the first vertex along $\rho_{i, j}[y_{i, j}, x_{i, j}]$ that touches $\pi[u_{i+1}, t_{l_i}]$. Since $\rho_{i, j}[p, v]$ is the shortest path (under weight perturbation) in the graph $G\setminus \left(E(\pi[s, t_{l_i}])\cup \{e_j\}\right)$, by uniqueness of shortest path under weight perturbation and the fact that \Cref{ssrp} finds the unique shortest paths under weight perturbation as well, it must be that the tree path $T^{l_i}_{p, e_j}[p, v]$ is the same as $\rho_{i, j}[p, v]$. Therefore, by the algorithm for backward shortcuts, we have:
	$$\mu_0(y_{i, j}, v)\leq |\rho_{i, j}[y_{i, j}, v]|$$
	
	We will prove that $\mu(y_{i, j}, x_{i, j}) = |\rho_{i, j}[y_{i, j}, x_{i, j}]|$ by an induction. List all the vertices $(v=)z_1, z_2, \cdots, z_\tau (= x_{i, j})$ in $\rho_{i, j}[y_{i, j}, x_{i, j}]\cap \pi[u_{i+1}, t_{l_i}]$ (in a right-to-left order) such that sub-path $\rho_{i, j}[z_k, z_{k+1}]\cap \pi$ only contains edges in $E(\pi[z_k, y_{i, j}])$. Let us prove by an induction on $k$ such that $\mu(y_{i, j}, z_k) = |\rho_{i, j}[y_{i, j}, z_k]|$.
	\begin{itemize}[leftmargin=*]
		\item \textbf{Basis.} When $k = 1$, $z_1 = v$ is the first vertex on $\rho_{i, j}[y_{i, j}, x_{i, j}]$ that sits on $\gamma_{l_i}$. Therefore, $\mu(y_{i, j}, z_1)\leq \mu_0(y_{i, j}, z_1) = |\rho_{i, j}[y_{i, j}, v]|$.
		
		\item \textbf{Induction.} For a general $k\geq 2$, consider the sub-path $\rho_{i, j}[z_{k-1}, z_k]$. If $\rho_{i, j}[z_{k-1}, z_k]$ does not pass through vertex $t_{l_i}$, then since $\rho_{i, j}$ is canonical, $\rho_{i, j}[z_{k-1}, z_k]\cap \pi$ only contains edges from $E(\pi[z_{k-1}, t_{l_i}])$. Therefore, $|\rho_{i, j}[z_{k-1}, z_k]| = \mu_0(z_{k-1}, z_k)$, and thus by inductive hypothesis we have:
		$$\mu(y_{i, j}, z_k)\leq \mu_0(z_{k-1}, z_k) + \mu(y_{i, j}, z_{k-1})\leq |\rho_{i, j}[y_{i, j}, z_k]|$$
		
		Otherwise if $\rho_{i, j}[z_{k-1}, z_k]$ passes through vertex $t_{l_i}$, then $\rho_{i, j}[t_{l_i}, z_k]$ is a shortest path that does not use any edges in $E(\pi[s, t_{l_i}])\cup E(\pi[y_{i, j}, t])$. Therefore, $\mu_0(t_{l_i}, z_k \mid y_{i, j})\leq |\rho_{i, j}[t_{l_i}, z_k]|$. By the inductive hypothesis, we have:
		$$\mu(y_{i, j}, z_k)\leq \mu(y_{i, j}, z_{k-1}) + |\pi[z_{k-1}, t_{l_i}] + \mu_0(t_{l_i}, z_k\mid y_{i, j})|\leq |\rho_{i, j}[y_{i, j}, z_k]|$$
	\end{itemize}
	
\end{proof}

For the rest, let us analyze the three sub-cases 2(a)(b)(c)(d).
\begin{claim}
	In sub-case 2(a), the algorithm correctly computes the value of $|\rho_{i, j}|$.
\end{claim}
\begin{proof}
	In this case, as $x_{i, j}\in V(\rho_{i, j}[s, y_{i, j}])$ and $\rho_{i, j}$ is canonical, the entire sub-path $\pi[x_{i, j}, y_{i, j}]$ should belong to $\rho_{i, j}$. Therefore, we can decompose $\rho_{i, j}$ as $\rho_{i, j} = \rho_{i, j}[s, t_{l_i}]\circ \pi[t_{l_i}, s_{l_j}]\circ \rho_{i, j}[s_{l_j}, t]$. Since $\rho_{i, j}[s, t_{l_i}]$ does not use any edge on $\pi[t_{l_i}, t]$, $|\rho_{i, j}[s, t_{l_i}]|$ is equal to $\dist\left(s, t_{l_i}, H\cup E(\pi[s, t_{l_i}])\setminus \{e_i\}\right)$; similarly, $|\rho_{i, j}[s_{l_j}, t]|$ is equal to $\dist\left(s_{l_j}, t, H\cup E(\pi[s_{l_j}, t])\setminus \{e_j\}\right)$. Therefore, by the algorithm, $\est(s, t, G\setminus\{e_i, e_j\})$ should be equal to $|\rho_{i, j}|$.
\end{proof}

\begin{claim}\label{2b}
	In sub-case 2(b), the algorithm correctly computes the value of $|\rho_{i, j}|$.
\end{claim}
\begin{proof}
	Decompose the path $\rho_{i, j} = \rho_{i, j}[s, y_{i, j}]\circ \rho_{i, j}[y_{i, j}, x_{i, j}]\circ \rho_{i, j}[x_{i, j}, t]$. By \Cref{backpath} and construction of graph $H_{i, j}$, we have $\wts(y_{i, j}, x_{i, j}) = |\rho_{i, j}[y_{i, j}, x_{i, j}]|$. For the rest, it suffices to show $|\rho_{i, j}[s, y_{i, j}]| = \dist(s, y_{i, j}, H_{i, j})$ and $|\rho[x_{i, j}, t]| = \dist(x_{i, j}, t, H_{i, j})$. By symmetry, let us only focus on $|\rho_{i, j}[s, y_{i, j}]|$. There are two different cases.
	\begin{itemize}[leftmargin=*]
		\item $\rho_{i, j}[a_{i, j}, y_{i, j}]$ passes through $s_{l_j}$.
		
		In this case, $|\rho_{i, j}[s, s_{l_j}]|$ is equal to $\wts(s, s_{l_j}) = \dist(s, s_{l_j}, G\setminus (E(\pi[s_{l_j}, t])\cup \{e_i\}))$. Therefore, the path $(s, s_{l_j})\circ (s_{l_j}, s^{l_j}_{h_j})\circ \pi[s^{l_j}_{h_j}, y_{i, j}]$ in $H_{i, j}$ has weight equal to $|\rho_{i, j}[s, y_{i, j}]|$.
		
		\item $\rho_{i, j}[a_{i, j}, y_{i, j}]$ does not pass through $s_{l_j}$.
		
		Let $z$ be the first vertex of $\rho_{i, j}(a_{i, j}, y_{i, j}]$ that appears on $\pi$. Since $\rho_{i, j}(a_{i, j}, y_{i, j}]$ does not pass through $s_{l_j}$, $z$ belongs to $\pi[s^{l_j}_{h_j}, y_{i, j}]$, and thus $|\rho_{i, j}(a_{i, j}, z]|\geq L$. Therefore, with high probability over $U$, $\rho_{i, j}(a_{i, j}, y_{i, j}]$ should go through a vertex $p\in U$ before it touches on $\pi$.
		
		We first argue that $|\rho_{i, j}[s, p]| = \dist(s, p, H_{i, j})$. In fact, if $a_{i, j}\notin V(\alpha^{l_i}_{h_i})$, then by (\romannumeral3) we know that $\wts(s, p) = |\rho_{i, j}[s, p]|$. Otherwise if $a_{i, j}\in V(\alpha^{l_i}_{h_i})$, the path $(s, s^{l_i}_{h_i})\circ \pi[s^{l_i}_{h_i}, a_{i, j}]\circ (a_{i, j}, p)$ in $H_{i, j}$ would have total weight equal to $|\rho_{i, j}[s, p]|$.
		
		Secondly, let us argue that $|\rho_{i, j}[p, y_{i, j}]| = \dist(p, y_{i, j}, H_{i, j})$. If $\rho_{i, j}[p, y_{i, j}]$ goes through vertex $s^{l_j}_{h_j}$, then the path $(p, s^{l_j}_{h_j})\circ \pi[s^{l_j}_{h_j}, y_{i, j}]$ has weight equal to $|\rho_{i, j}[p, y_{i, j}]|$. Otherwise, the path $(p, z)\circ \pi[z, y_{i, j}]$ has weight equal to $|\rho_{i, j}[p, y_{i, j}]|$.
	\end{itemize}
\end{proof}

\begin{claim}\label{2c}
	In sub-case 2(c), the algorithm correctly computes the value of $|\rho_{i, j}|$.
\end{claim}
\begin{proof}
	Decompose the path $\rho_{i, j} = \rho_{i, j}[s, y_{i, j}]\circ \rho_{i, j}[y_{i, j}, x_{i, j}]\circ \rho_{i, j}[x_{i, j}, t]$. By \Cref{backpath} and construction of graph $H^{(l_i, h_i)}_{j}$, we have $\wts(y_{i, j}, x_{i, j}) = |\rho_{i, j}[y_{i, j}, x_{i, j}]|$. For the rest, let us study the values of $|\rho_{i, j}[x_{i, j}, t]|$ and $|\rho_{i, j}[s, y_{i, j}]|$.
	
	We first show $\dist\left(x_{i, j}, t, H^{(l_i, h_i)}_j\right) = |\rho_{i, j}[x_{i, j}, t]|$. Similar to the proof in \Cref{2b}, there are two different cases.
	\begin{itemize}[leftmargin=*]
		\item $\rho_{i, j}[x_{i, j}, t]$ passes through $t_{l_i}$.
		
		In this case, $|\rho_{i, j}[t_{l_i}, t]|$ is equal to $\wts(t_{l_i}, t) = \dist(t_{l_i}, t, G\setminus (E(\pi[s, t_{l_i}])\cup\{e_j\}))$. There the path $\pi[x_{i, j}, t_{l_i}]\circ (t_{l_i}, t)$ in $H^{(l_i, h_i)}_j$ has weight equal to $|\rho_{i, j}[s, y_{i, j}]|$.
		
		\item $\rho_{i, j}[x_{i, j}, t]$ does not pass through $t_{l_i}$.
		
		Let $z$ be the last vertex of $\rho_{i, j}[x_{i, j}, b_{i, j})$ that appears on $\pi$. Since $\rho_{i, j}[x_{i, j}, b_{i, j})$ does not pass through $s_{l_i}$, $z$ belongs to $\pi[x_{i, j}, t^{l_i}_{h_i}]$, and thus $|\rho_{i, j}[z, b_{i, j})| \geq L$. Hence, with high probability over $U$, $\rho_{i, j}[x_{i, j}, b_{i, j})$ should go through a vertex $p\in U$ before it touches on $\pi$.
		
		We first argue that $|\rho_{i, j}[p, t]| = \dist(p, t, H^{(l_i, h_i)}_j)$. In fact, if $b_{i, j}\notin V(\alpha^{l_j}_{h_j})$, then by (\romannumeral3), we know that $\wts(p, t) = |\rho_{i, j}[p, t]|$. Otherwise if $b_{i, j}\in V(\alpha^{l_j}_{h_j})$, then the path $(p, b_{i, j})\circ \pi[b_{i, j}, t^{l_j}_{h_j}\circ (t^{l_j}_{h_j}, t)]$ in $H^{(l_i, h_i)}_j$ would have total weight equal to $|\rho_{i, j}[p, t]|$.
		
		Secondly, let us argue that $|\rho_{i, j}[x_{i, j}, p]| = \dist\left(x_{i, j}, p , H^{(l_i, h_i)}_j\right)$. If $\rho_{i, j}[x_{i, j}, p]$ goes through vertex $t^{l_i}_{h_i}$, then the path $\pi[x_{i, j}, t^{l_i}_{h_i}]\circ (s^{l_i}_{h_i}, p)$ has weight equal to $|\rho_{i, j}[x_{i, j}, p]|$. Otherwise, the path $\pi[x_{i, j}, z]\circ (p, z)$ has weight equal to $|\rho_{i, j}[x_{i, j}, p]|$.
	\end{itemize}

	Next, let us look at the value of $|\rho_{i, j}[s, y_{i, j}]|$. The proof is almost the same as \Cref{2b} with some changes.
	\begin{itemize}[leftmargin=*]
		\item $\rho_{i, j}[a_{i, j}, y_{i, j}]$ passes through $s_{l_j}$.
		
		In this case, $|\rho_{i, j}[s, s_{l_j}]|$ is equal to $\dist(s, s_{l_j}, G\setminus (E(\pi[s_{l_j}, t])\cup \{e_i\}))$. Therefore, the path $(s, s_{l_j})\circ \pi[s_{l_j}, t]$ in $H^{(l_i, h_i)}_j$ has weight equal to $ \dist(s, s_{l_j}, G\setminus (E(\pi[s_{l_j}, t])\cup \{e_i\})) + \dist(s_{l_j}, t, H^{(l_i, h_i)}_j)$ which is captured by the algorithm.
		
		\item $\rho_{i, j}[a_{i, j}, y_{i, j}]$ does not pass through $s_{l_j}$.
		
		Let $z$ be the first vertex of $\rho_{i, j}(a_{i, j}, y_{i, j}]$ that appears on $\pi$. Since $\rho_{i, j}(a_{i, j}, y_{i, j}]$ does not pass through $s_{l_j}$, $z$ belongs to $\pi[s^{l_j}_{h_j}, y_{i, j}]$, and thus $|\rho_{i, j}(a_{i, j}, z]|\geq L$. Therefore, with high probability over $U$, $\rho_{i, j}(a_{i, j}, y_{i, j}]$ should go through a vertex $p\in U$ before it touches on $\pi$.
		
		First, let us argue that $|\rho_{i, j}[p, y_{i, j}]| = \dist\left(p, y_{i, j}, H^{(l_i, h_i)}_j\right)$. If $\rho_{i, j}[p, y_{i, j}]$ goes through vertex $s^{l_j}_{h_j}$, then the path $(p, s^{l_j}_{h_j})\circ \pi[s^{l_j}_{h_j}, y_{i, j}]$ in $H^{(l_i, h_i)}_j$ has weight equal to $|\rho_{i, j}[p, y_{i, j}]|$. Otherwise, the path $(p, z)\circ \pi[z, y_{i, j}]$ in $H^{(l_i, h_i)}_j$ has weight equal to $|\rho_{i, j}[p, y_{i, j}]|$.
		
		Next, let us look at the value of $|\rho_{i, j}[s, p]|$. If $a_{i, j}\notin V(\alpha^{l_i}_{h_i})$, then by (\romannumeral3) we know that $\wts(s, p) = |\rho_{i, j}[s, p]|$. Otherwise if $a_{i, j}\in V(\alpha^{l_i}_{h_i})$, then we can decompose the path as following
		$$\begin{aligned}
			|\rho_{i, j}[s, p]| = |\pi[s, a_{i, j}]| + |\rho_{i, j}[a_{i, j}, p]| 
			= |\pi[s, a_{i, j}]| + \dist(a_{i, j}, p, H) 
		\end{aligned}$$
		Therefore, 
		$$\begin{aligned}
			|\rho_{i, j}[s, t]| &= |\rho_{i, j}[s, a_{i, j}]| + |\rho_{i, j}[a_{i, j}, p]|+ |\rho_{i, j}[p, y_{i, j}]|+ |\rho_{i, j}[y_{i, j}, x_{i, j}]| + |\rho_{i, j}[x_{i, j}, t]|\\			
			&= |\rho_{i, j}[s, a_{i, j}]| + \dist(a_{i, j}, p, H) + \dist\left(p, y_{i, j}, H^{(l_i, h_i)}_j\right)\\&+ \dist\left(y_{i, j}, x_{i, j}, H^{(l_i, h_i)}_j\right) + \dist\left(x_{i, j}, t, H^{(l_i, h_i)}_j\right)\\
			&= |\rho_{i, j}[s, a_{i, j}]| + \dist(a_{i, j}, p, H)  + \dist\left(p, t, H^{(l_i, h_i)}_j\right)
		\end{aligned}$$
		which is captured by the algorithm.
	\end{itemize}
\end{proof}

\begin{claim}
	In sub-case 2(d), the algorithm correctly computes the value of $|\rho_{i, j}|$.
\end{claim}
\begin{proof}
	Consider three different cases of $\rho_{i, j}$.
	\begin{itemize}[leftmargin=*]
		\item $s_{l_j}\notin V\brac{\rho_{i, j}[s, y_{i, j}]}$ and $t_{l_i}\notin V\brac{\rho_{i, j}[x_{i, j}, t]}$.
		
		Similar to the proofs in \Cref{2b} and \Cref{2c}, $\rho_{i, j}(a_{i, j}, y_{i, j}]$ passes through a vertex $p\in U$ before it touches on $\pi$, and $\rho[x_{i, j}, b_{i, j})$ passes through a vertex $q\in U$ after its last visit on any vertices on $\pi$. As $\rho_{i, j}[p, y_{i, j}]$ does not contain any edge in $E\brac{\pi[s^{l_i}_{h_i}, t]}$, due to edges and shortcuts in (\romannumeral1)(\romannumeral2), we have $|\rho_{i, j}[p, y_{i, j}]| = \dist\left(p, y_{i, j}, H^{(l_i, h_i), (l_j, h_j)}\right)$. Symmetrically, we can also prove $|\rho_{i, j}[x_{i, j}, q]| = \dist\left(x_{i, j}, q, H^{(l_i, h_i), (l_j, h_j)}\right)$. Since $\wts(y_{i, j}, x_{i, j}) = |\rho_{i, j}[y_{i, j}, x_{i, j}]|$ by \Cref{backpath}, we have:
		$$|\rho_{i, j}[p, q]| = \dist\brac{p, q, H^{(l_i, h_i), (l_j, h_j)}}$$
		
		If $a_{i, j}\in V\brac{\alpha^{l_i}_{h_i}}$, then there is a path $\brac{s, s^{l_i}_{h_i}}\circ \pi[s^{l_i}_{h_i}, a_{i, j}]\circ (a_{i, j}, p)$ in $H^{(l_i, h_i), (l_j, h_j)}_{i, j}$ with weight equal to $|\rho_{i, j}[s, p]|$; otherwise, the edge $(s, p)$ in $H^{(l_i, h_i), (l_j, h_j)}_{i, j}$ has weight equal to $|\rho_{i, j}[s, p]|$. Hence, we always have $|\rho_{i, j}[s, p]| = \dist\brac{s, p, H^{(l_i, h_i), (l_j, h_j)}_{i, j}}$. Symmetrically, we also have $|\rho_{i, j}[q, t]| = \dist\brac{q, t, H^{(l_i, h_i), (l_j, h_j)}_{i, j}}$. As the shortcut edge $(p, q)$ in $H^{(l_i, h_i), (l_j, h_j)}_{i, j}$ has weight $$\wts(p, q) = \dist\brac{p, q, H^{(l_i, h_i), (l_j, h_j)}_{i, j}}$$
		we can conclude that $|\rho_{i, j}| = \dist\brac{s, t, H^{(l_i, h_i), (l_j, h_j)}_{i, j}}$.
		
		\item $s_{l_j}\in V\brac{\rho_{i, j}[s, y_{i, j}]}$ and $t_{l_i}\notin V\brac{\rho_{i, j}[x_{i, j}, t]}$.
		
		Similar to the previous case, we can still find the vertex $q\in U$ on $\rho_{i, j}[x_{i, j}, t]$. Notice that we have added all edges $E\brac{\pi[s^{l_j}_{h_j}], y_{i, j}}$ to $H^{(l_i, h_i), (l_j, h_j)}$, we have $$|\rho_{i, j}[s^{l_j}_{h_j}]| = \dist\brac{s^{l_j}_{h_j}, t, H^{(l_i, h_i), (l_j, h_j)}_{i, j}}$$
		Therefore, we have added a shortcut $\brac{s^{l_j}_{h_j}, t}$ with weight $|\rho_{i, j}[s^{l_j}_{h_j}, t]|$ to $H^{(l_i, h_i), (l_j, h_j)}_{i, j}$.
		
		Then, in graph $H^{(l_i, h_i), (l_j, h_j)}_{i, j}$, as we have a shortcut $\brac{s, s^{l_j}_{h_j}}$ with weight equal to $\wts\brac{s, s^{l_j}_{h_j}} = |\rho_{i, j}[s, s^{l_j}_{h_j}]|$, we know that $|\rho_{i, j}| = \dist\brac{s, t, H^{(l_i, h_i), (l_j, h_j)}_{i, j}}$.
		
		\item $s_{l_j}\in V\brac{\rho_{i, j}[s, y_{i, j}]}$ and $t_{l_i}\in V\brac{\rho_{i, j}[x_{i, j}, t]}$.
		
		In graph $H^{(l_i, h_i), (l_j, h_j)}_{i, j}$, as we have a shortcut $(s, s^{l_j}_{h_j})$ with weight equal to $\wts\brac{s, s^{l_j}_{h_j}} = |\rho_{i, j}[s, s^{l_j}_{h_j}]|$, a shortcut $\brac{t^{l_i}_{h_i}, t}$ with weight equal to $\wts\brac{t^{l_i}_{h_i}, t} = |\rho_{i, j}[t^{l_i}_{h_i}, t]|$, and a backward shortcut $\brac{s^{l_j}_{h_j}, t^{l_i}_{h_i}}$ with weight $|\rho_{i, j}[s^{l_j}_{h_j}, t^{l_i}_{h_i}]|$ we know that $|\rho_{i, j}| = \dist\brac{s, t, H^{(l_i, h_i), (l_j, h_j)}_{i, j}}$.
	\end{itemize}
\end{proof}

\subsection*{Case 3: $0\leq l_j - l_i \leq 1$}
\paragraph{Preparation.} 
For notational convenience, for each $l$, define $\beta_l = \gamma_l\circ \gamma_{l+1}$ which is the concatenation of two consecutive sub-paths. Then, as $l_j - l_i\leq 1$, both edges $e_i, e_j$ are on sub-path $\beta_{l_i}$. Similar to the previous case, we also need single source shortest paths from a random pivot set. For clarity let us rewrite some of these steps.
\begin{itemize}[leftmargin=*]
	\item Uniformly sample a vertex subset $U$ of size $\frac{10n\log n}{L}$. Then, for each vertex $v\in U$, compute a single-source shortest paths tree to and from $v$ in graph $H = G\setminus E(\pi)$.
	
	\item For each pair of indices $l, h$, compute single-source shortest paths from $s$ and to $t$ in graphs $G\setminus E(\beta_l)$, $H\cup E(\pi[s, s^l_h])$ and $H\cup E(\pi[t^l_h, t])$.
\end{itemize}

\paragraph{Sub-path shortcuts for each $V(\beta_l)$.} We need to compute all distances: 
$$\mu_1(w, z) \overset{\text{def}}{=} \dist(w, z, H)$$
for any $w, z\in V(\beta_l)$; different from the previous case, $z$ does not have to lie between $s$ and $w$ on $\pi$. We will perform truncated Dijkstra up to depth $L$ on vertices in $\pi$, which is similar to what we did in previous sections.

For any offset $1\leq b\leq 10$, initialize a vertex subset $A_b\leftarrow V$, and go over all the sub-paths $\beta_b, \beta_{b+10}, \cdots, \beta_{b+10k}, k \leq \ceil{n/10L}$ sequentially. For each sub-paths $\beta_{b+10j}$, enumerate all vertices $w\in V(\gamma_{b+10j})$ and perform a truncated Dijkstra up to depth $2L$ in the graph $H[A_b]$. Then, for each $z\in V(\beta_{b+10j})$, if $z$ was visited by the truncated Dijkstra, then store a distance value $\mu(w, z)\leftarrow \dist(w, z, H[A_b])$; otherwise if $z$ was not visited, then go over all vertices $p\in U$, and store a distance value: $$\mu_1(w, z)\leftarrow \min_{p\in U}\{\dist(w, p, H) + \dist(p, z, H) \}$$
After that, let $P_{b+10j}$ collect all the vertices visited by any truncated Dijkstra of vertices in $\beta_{b+10j}$, and prune $A_{b}\leftarrow A_{b}\setminus P_{b+10j}$. 

\paragraph{Backward shortcuts for each $V(\beta_l)$.} Using the information of $\mu_1(\cdot, \cdot)$, we are able to compute the following distances:
$$\mu_2(w, z) = \dist(w, z, H\cup E(\pi[z, w]))$$
for any $w, z\in V(\beta_l)$ such that $z$ lies between $s$ and $w$ on $\pi$. To compute these values, for each $w$, apply a dynamic programming procedure which enumerates all vertices $z\in V(\pi[s_l, w])$ in the reverse order and updates the table:$$\mu_2(w, z)\leftarrow \min_{v_1, v_2\in V(\pi(z, w)), v_2\in V(\pi[v_1, w))}\{\mu_2(w, v_1) + |\pi[v_1, v_2]| + \mu_1(v_2, z)\}$$
To calculate this minimum formula efficiently, note that minimizing the sum $|\pi[v_1, v_2]| + \mu_1(v_2, z)$ is equivalent to minimizing $|\pi[s, v_2]| + \mu(v_2, z)$ which is independent of $v_1$. Therefore, the minimizer of $|\pi[v_1, v_2]| + \mu_1(v_2, z)$ can be computed in $O(\log n)$ time using interval data structures. Hence, the minimum formula can be calculated in $O(L\log n)$ time.

\paragraph{Sub-case 3(a): $\alpha^{l_i}_{h_i} = \alpha^{l_j}_{h_j}$.} In this case, build the following weighted graph $Z_{i, j}$ with edge weight $\wts(*, *)$.
\begin{itemize}[leftmargin=*]
	\item \textbf{Vertices.} Add all vertices $\{s, t\}\cup V(\beta_{l_i})$ to $Z_{i, j}$.
	
	\item \textbf{Edges.} Add edges $E(\beta_{l_i})\setminus \{e_i, e_j\}$ to $Z_{i, j}$. Then,for any $w, z\in V(\beta_{l_i})$, add a shortcut edge $(w, z)$ with edge weight $\wts(w, z) = \mu_1(w, z)$. Also, for each $v\in\{s, t\}\cup V(\beta_{l_i})$, add edges $(s, v), (v, t)$ with weights $\wts(s, v) = \dist\brac{s, v, G\setminus E(\beta_{l_i})}$ and $\wts(v, t) = \dist\brac{v, t, G\setminus E(\beta_{l_i})}$.
\end{itemize}
After that, assign $\est(s, t, G\setminus\{e_i, e_j\})\leftarrow \min\{\est(s, t, G\setminus\{e_i, e_j\}), \dist(s, t, Z_{i, j})\}$.

\paragraph{Sub-case 3(b): $\alpha^{l_i}_{h_i} \neq \alpha^{l_j}_{h_j}$, and $x_{i, j}\in V\brac{\alpha^{l_i}_{h_i}}$ and $y_{i, j}\in V\brac{\alpha^{l_j}_{h_j}}$.} In this case, build the following weighted graph $Z_{i, j}$ with edge weight $\wts(*, *)$.
\begin{itemize}[leftmargin=*]
	\item \textbf{Vertices.} Add vertices $\{s, t\}\cup V\brac{\alpha^{l_i}_{h_i}}\cup V\brac{\alpha^{l_j}_{h_j}}$ to $Z_{i, j}$.
	
	\item \textbf{Edges.} Add the following types of edges to $Z_{i, j}$.
	\begin{enumerate}[(i),leftmargin=*]
		\item Add all edges in $E\brac{\alpha^{l_i}_{h_i}\cup \alpha^{l_j}_{h_j}}\setminus \{e_i, e_j\}$.
		
		\item For each $u\in V\brac{\pi[u_{i+1}, t^{l_i}_{h_i}]}$ and $v\in V\brac{\pi[s^{l_j}_{h_j}, u_j]}$, add shortcut edge $(u, v)$ with weight $\wts(u, v) = |\pi[u, v]|$, and shortcut edge $(v, u)$ with weight $\wts(v, u) = \mu_2(v, u)$.
		
		\item For each $u\in V\brac{\pi[s^{l_i}_{h_i}, u_i]}$, add a shortcut $\brac{u, s^{l_j}_{h_j}}$ with weight
		$$\wts\brac{u, s^{l_j}_{h_j}} = \min_{z\in V\brac{\pi[t^{l_i}_{h_i}, s^{l_j}_{h_j}]}}\left\{ \mu_1(u, z) + |\pi[z, s^{l_j}_{h_j}]| \right\}$$
		
		Symmetrically, for each $v\in V\brac{\pi[u_{j+1}, t^{l_j}_{h_j}]}$, add a shortcut $\brac{t^{l_i}_{h_i}, v}$ with weight 
		$$\wts\brac{t^{l_i}_{h_i}, v} = \min_{z\in V\brac{\pi[t^{l_i}_{h_i}, s^{l_j}_{h_j}]}}\left\{ |\pi[t^{l_i}_{h_i}, z]| + \mu_1(z, v) \right\}$$
		
		Also, add a shortcut $\brac{s, s^{l_j}_{h_j}}$ with weight
		$$\wts\brac{s, s^{l_j}_{h_j}} = \min_{z\in V\brac{\pi[t^{l_i}_{h_i}, s^{l_j}_{h_j}]}}\left\{ \dist(s, z, H\cup E(\pi[s, s_{l_i}])) + |\pi[z, s^{l_j}_{h_j}]| \right\}$$
		
		As well as a shortcut $\brac{t^{l_i}_{h_i}, t}$ with weight
		$$\wts\brac{t^{l_i}_{h_i}, t} = \min_{z\in V\brac{\pi[t^{l_i}_{h_i}, s^{l_j}_{h_j}]}}\left\{ |\pi[t^{l_i}_{h_i}, z]| + \dist(z, t, H\cup E(\pi[t_{l_j}, t])) \right\}$$
		
		These quantities can be evaluated in $O(\log n)$ time using standard interval data structures.
		
		\item For each $v\in V\brac{\alpha^{l_i}_{h_i}}\cup V\brac{\alpha^{l_j}_{h_j}}$, add shortcut edge $(s, v)$ with weight $$\wts(s, v) = \dist\brac{s, v, H\cup E\brac{\pi[s, s^{l_i}_{h_i}]}}$$
		as well as shortcut edge $(v, t)$ with weight
		$$\wts(v, t) = \dist\brac{v, t, H\cup E\brac{\pi[t^{l_j}_{h_k}, t]}}$$
	\end{enumerate}
\end{itemize}
After that, assign $\est(s, t, G\setminus\{e_i, e_j\})\leftarrow \min\{\est(s, t, G\setminus\{e_i, e_j\}), \dist(s, t, Z_{i, j})\}$. See \Cref{2fail-case2-shortcut8} for an illustration.

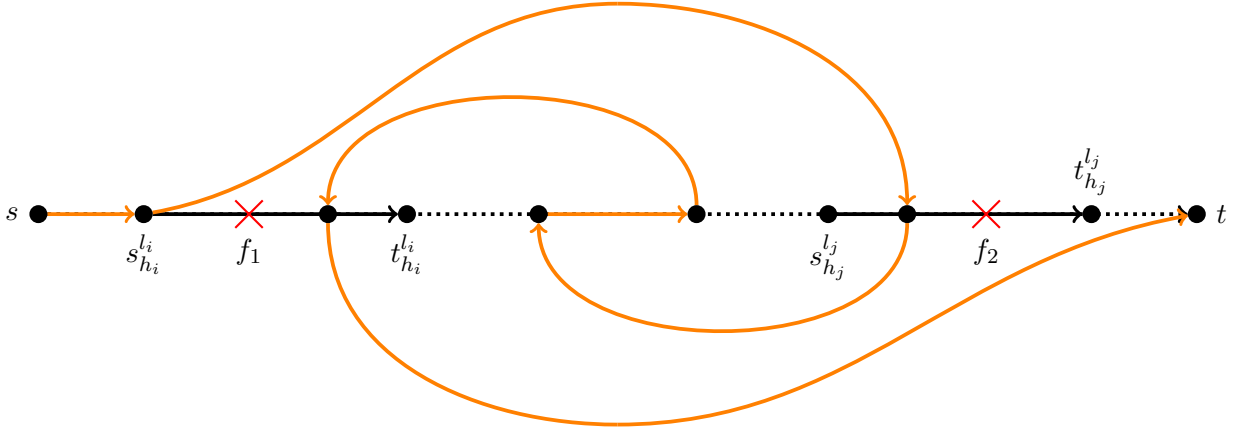
\begin{figure}
	\centering
	\begin{tikzpicture}[thick,scale=0.7]
	\draw (-11, 0) node(1)[circle, draw, fill=black, inner sep=0pt, minimum width=6pt, label=180: {$s$}] {};
	\draw (11, 0) node(2)[circle, draw, fill=black, inner sep=0pt, minimum width=6pt, label=0: {$t$}] {};
	
	\draw (-9, 0) node(3)[circle, draw, fill=black, inner sep=0pt, minimum width=6pt, label=-90: {$s_{h_i}^{l_i}$}] {};
	\draw (-4, 0) node(4)[circle, draw, fill=black, inner sep=0pt, minimum width=6pt, label=-90: {$t^{l_i}_{h_i}$}] {};
	\draw (4, 0) node(5)[circle, draw, fill=black, inner sep=0pt, minimum width=6pt, label=-90: {$s^{l_j}_{h_j}$}] {};
	\draw (9, 0) node(6)[circle, draw, fill=black, inner sep=0pt, minimum width=6pt, label=90: {$t^{l_j}_{h_j}$}] {};
	\draw (-5.5, 0) node(7)[circle, draw, fill=black, inner sep=0pt, minimum width=6pt] {};
	\draw (5.5, 0) node(8)[circle, draw, fill=black, inner sep=0pt, minimum width=6pt] {};
	\draw (-1.5, 0) node(9)[circle, draw, fill=black, inner sep=0pt, minimum width=6pt] {};
	\draw (1.5, 0) node(10)[circle, draw, fill=black, inner sep=0pt, minimum width=6pt] {};	
	\draw (-7, 0) node[cross=6, red, label=-90: {$f_1$}] {};
	\draw (7, 0) node[cross=6, red, label=-90: {$f_2$}] {};
	
	\begin{scope}[on background layer]
		\draw [->, line width = 0.5mm, dotted] (1) to (2);
		\draw [->, line width = 0.5mm] (3) to (4);
		\draw [->, line width = 0.5mm] (5) to (6);

		\draw [->, line width = 0.5mm, color=orange] (1) to (3);
		\draw [line width = 0.5mm, color=orange] (3) to[out=10, in=180] (0, 4);
		\draw [->, line width = 0.5mm, color=orange] (0, 4) to[out=0, in=90] (8);
		\draw [->, line width = 0.5mm, color=orange] (8) to[out=-90, in=-90] (9);
		\draw [->, line width = 0.5mm, color=orange] (9) to (10);
		\draw [->, line width = 0.5mm, color=orange] (10) to[out=90, in=90] (7);
		\draw [line width = 0.5mm, color=orange] (7) to[out=-90, in=180] (0, -4);
		\draw [->, line width = 0.5mm, color=orange] (0, -4) to[out=0, in=-170] (2);

	\end{scope}
\end{tikzpicture}
	\caption{The shortcut edges in $Z_{i, j}$ captures the shortest replacement path $\rho_{i, j}$ drawn as the orange curve.}
	\label{2fail-case2-shortcut8}
\end{figure}

\paragraph{Sub-case 3(c): $\alpha^{l_i}_{h_i} \neq \alpha^{l_j}_{h_j}$, and $x_{i, j}\notin V\brac{\alpha^{l_i}_{h_i}}$ or $y_{i, j}\notin V\brac{\alpha^{l_j}_{h_j}}$.} Without loss of generality, assume $x_{i, j}\notin V\brac{\alpha^{l_i}_{h_i}}$; this also includes the case where $x_{i, j} = \bot$ does not exist. Build the following weighted graph $Z^{(l_i, h_i)}_j$ with edge weight $\wts(*, *)$.
\begin{itemize}[leftmargin=*]
	\item \textbf{Vertices.} Add vertices $\{s, t\}\cup V(\beta_{l_i})$ to $Z_{i, j}$.
	
	\item \textbf{Edges.} Add the following types of edges to $Z_{i, j}$.
	
	\begin{enumerate}[(i),leftmargin=*]
		\item Add all edges in $E(\beta_{l_i})\setminus \brac{V\brac{\alpha^{l_i}_{h_i}}\cup \{e_j\}}$.
		
		\item Then,for any $w, z\in V(\beta_{l_i})$, add a shortcut edge $(w, z)$ with edge weight $\wts(w, z) = \mu_1(w, z)$.
		
		\item For each $v\in V(\beta_{l_i})$, add shortcut edge $(s, v)$ with weight $$\wts(s, v) = \dist\brac{s, v, H\cup E\brac{\pi[s, s^{l_i}_{h_i}]}}$$
		as well as shortcut edge $(v, t)$ with weight
		$$\wts(v, t) = \dist\brac{v, t, H\cup E\brac{\pi[t^{l_j}_{h_k}, t]}}$$
	\end{enumerate}
\end{itemize}
Computing single-source shortest paths to $t$ in graph $Z^{(l_i, h_i)}_j$. After that, update $\est(s, t, G\setminus\{e_i, e_j\})$ with the minimum clause:
$$\min_{v\in V\brac{\pi[s^{l_i}_{h_i}, u_i]}}\left\{\dist\brac{s, t, Z^{(l_i, h_i)}_j}, |\pi[s, v]| + \dist\brac{v, t, Z^{(l_i, h_i)}_j}\right\}$$
See \Cref{2fail-case2-shortcut9} for an illustration.

\begin{figure}
	\centering
	\begin{tikzpicture}[thick,scale=0.7]
	\draw (-11, 0) node(1)[circle, draw, fill=black, inner sep=0pt, minimum width=6pt, label=180: {$s$}] {};
	\draw (11, 0) node(2)[circle, draw, fill=black, inner sep=0pt, minimum width=6pt, label=0: {$t$}] {};
	
	\draw (-9, 0) node(3)[circle, draw, fill=black, inner sep=0pt, minimum width=6pt, label=-90: {$s_{h_i}^{l_i}$}] {};
	\draw (-6, 0) node(4)[circle, draw, fill=black, inner sep=0pt, minimum width=6pt, label=-90: {$t^{l_i}_{h_i}$}] {};
	\draw (4, 0) node(5)[circle, draw, fill=black, inner sep=0pt, minimum width=6pt, label=-90: {$s^{l_j}_{h_j}$}] {};
	\draw (9, 0) node(6)[circle, draw, fill=black, inner sep=0pt, minimum width=6pt, label=90: {$t^{l_j}_{h_j}$}] {};
	\draw (-4.5, 0) node(7)[circle, draw, fill=black, inner sep=0pt, minimum width=6pt] {};
	\draw (5.5, 0) node(8)[circle, draw, fill=black, inner sep=0pt, minimum width=6pt] {};
	\draw (-1.5, 0) node(9)[circle, draw, fill=black, inner sep=0pt, minimum width=6pt] {};
	\draw (1.5, 0) node(10)[circle, draw, fill=black, inner sep=0pt, minimum width=6pt] {};
	\draw (-7.5, 0) node(11)[circle, draw, fill=black, inner sep=0pt, minimum width=6pt, label=-90: {$v$}] {};
	
	\draw (-3, 0) node(13)[circle, draw, fill=black, inner sep=0pt, minimum width=6pt, label=-90: {$t_{l_i}$}] {};
	\draw (7, 0) node[cross=6, red, label=-90: {$f_2$}] {};
	
	\begin{scope}[on background layer]
		\draw [->, line width = 0.5mm, dotted] (1) to (2);
		\draw [->, line width = 0.5mm, color=red] (11) to (4);
		\draw [->, line width = 0.5mm] (4) to (13);
		\draw [->, line width = 0.5mm] (5) to (6);

		\draw [->, line width = 0.5mm, dotted, color=orange] (1) to (11);
		\draw [line width = 0.5mm, color=orange] (11) to[out=30, in=180] (0, 4);
		\draw [->, line width = 0.5mm, color=orange] (0, 4) to[out=0, in=90] (8);
		\draw [->, line width = 0.5mm, color=orange] (8) to[out=-90, in=-90] (9);
		\draw [->, line width = 0.5mm, color=orange] (9) to (10);
		\draw [->, line width = 0.5mm, color=orange] (10) to[out=90, in=90] (7);
		\draw [line width = 0.5mm, color=orange] (7) to[out=-90, in=180] (0, -4);
		\draw [->, line width = 0.5mm, color=orange] (0, -4) to[out=0, in=-170] (2);

	\end{scope}
\end{tikzpicture}
	\caption{As a typical example, the shortcut edges in $Z^{(l_i, h_i)}_{j}$ captures the suffix of the shortest replacement path $\rho_{i, j}[v, t], v\in V\brac{\alpha^{l_i}_{h_i}}$ drawn as the orange curve. The dotted orange segment will be computed once we are given $e_i$.}
	\label{2fail-case2-shortcut9}
\end{figure}

\paragraph{Runtime.} Let us list all the runtime complexities below.
\begin{itemize}[leftmargin=*]
	\item \textbf{Preparation.} The runtime is bounded by $\tilde{O}(\frac{n^3}{g})$.
	
	\item \textbf{Sub-path shortcuts.} By the pruning procedure, each vertex is searched by at most $O(L)$ times throughout the truncated Dijkstra procedures. Hence, the total runtime is $\tilde{O}(n^2L)$.
	
	\item \textbf{Backward shortcuts.} According to the algorithm, the dynamic programming takes time $\tilde{O}(nL^2)$.
	
	\item \textbf{Sub-case 3(a).} The graph $Z_{i, j}$ has size $O(L^2)$, so the runtime is bounded by $\tilde{O}(ngL^2)$.
	
	\item \textbf{Sub-case 3(b).} The graph $Z_{i, j}$ has size $O(g^2)$, so the runtime is bounded by $\tilde{nLg^2}$.
	
	\item \textbf{Sub-case 3(c).} The graph $Z^{(l_i, h_i)}_j$ has size $O(L^2)$, so applying Dijkstra's algorithm on all graphs $Z^{(l_i, h_i)}_j$ takes time $\tilde{O}(nL^3/g)$. Then, updating the values of $\est(s, t, G\setminus\{e_i, e_j\})$ takes time $O(g)$ for each pair of $(e_i, e_j)$, and thus $O(nLg)$ in total.
\end{itemize}
To summarize, the overall runtime is bounded by $$\tilde{O}\brac{\frac{n^3}{g} + n^2L + ngL^2 + \frac{nL^3}{g}}$$

\paragraph{Correctness.} Similar to the previous case, we analyze some properties of the shortcuts for each sub-path. 
\begin{claim}
	The algorithm correctly computes the values $$\mu_1(w, z) = \dist(w, z, H)$$ of sub-path shortcuts for each $V(\beta_{b+10j})$, where $w, z\in V(\beta_{b+10j})$.
\end{claim}
\begin{proof}
	For any $w, z\in V(\beta_{b+10j})$, let $\theta$ be the shortest path from $w$ to $z$ in $H$. First consider the case that $\dist(w, z, H) >2L$. In this case, with high probability, there exists a vertex $p\in U$ on this path. Therefore, we can compute the value of $|\dist(w, z, H)|$ via the formula:
	$$\mu(w, z)\leftarrow \min_{p\in U}\{ \dist(w, p, H) + \dist(p, z, H) \}$$
	
	Otherwise if $\dist(w, z, H) \leq 2L$, then consider the moment when a truncated Dijkstra was executed at $w$ in graph $H[A_b]$. If all vertices on the shortest path are in $A_b$, then it would be discovered by the truncated Dijkstra algorithm. Otherwise, there would be a vertex $v\in V(\theta)$ which was visited already by some previous execution of truncated Dijkstra, say a vertex $x$ on sub-path $\beta_{b+10h}, h<j$. Then, $7L\leq \dist(x, z, G)\leq \dist(x, v, G) + |\theta[v, z]| \leq 4L$, which is a contradiction.
\end{proof}

\begin{claim}
	The algorithm correctly computes the values $$\mu_2(w, z) = \dist(w, z, H\cup E(\pi[z, w]))$$ of sub-path shortcuts for each $V(\beta_{l})$, for any $w, z\in V(\beta_l)$ such that $z$ lies between $s$ and $w$ on $\pi$.
\end{claim}
\begin{proof}
	This is proved by a straightforward induction. Basically, for any $z$, assume the shortest path $\theta$ from $w$ to $z$ in $H\cup E(\pi[z, w])$ passes through a vertex $v_1\in V(\pi(z, w))$ which minimizes $|\pi[z, v_1]|$. Then, the suffix sub-path $\theta[v_1, z]$ only uses edges in $H\cup E(\pi[v_1, w])$, which is captured by the minimum formula.
\end{proof}

\begin{claim}
	In sub-case 3(a), the algorithm correctly computes the value of $|\rho_{i, j}|$.
\end{claim}
\begin{proof}
	If $\rho_{i, j}$ skips over the entire sub-path $\beta_{l_i}$, then graph $Z_{i, j}$ contains a direct edge $(s, t)$ with weight $\dist(s, t, G\setminus E(\beta_{l_i})) = |\rho_{i, j}|$. Otherwise, suppose the first and the last vertices of $\rho_{i, j}$ in $V(\beta_{l_i})$ are $u, v$ respectively. Then by construction of $Z_{i, j}$, it contains all edges in $E(\beta_{l_i})\setminus \{e_i, e_j\}$ and all shortcuts in $H$. Therefore, $\dist(u, v, Z_{i, j}) = |\rho_{i, j}[u, v]|$ and $\dist(s, t, Z_{i, j}) = |\rho_{i, j}|$.
\end{proof}

\begin{claim}
	In sub-case 3(b), the algorithm correctly computes the value of $|\rho_{i, j}|$.
\end{claim}
\begin{proof}
		First, consider the case where $x_{i, j}$ lies between $s$ and $y_{i, j}$ on $\rho_{i, j}$. Then, since $\rho_{i, j}$ is canonical, $\rho_{i, j}[x_{i, j}, y_{i, j}] = \pi[x_{i, j}, y_{i, j}]$. If $a_{i, j}\notin V(\alpha^{l_i}_{h_i})$, then $\rho_{i, j}[s, x_{i, j}]$ is a path in $H\cup E(\pi[s, s^{l_i}_{h_i}])$. Therefore, the edge $(s, x_{i, j})$ in $Z_{i, j}$ has weight $$\wts(s, x_{i, j}) = \dist\brac{s, x_{i, j}, H\cup E(\pi[s, s^{l_i}_{h_i}])} = |\rho_{i, j}[s, x_{i, j}]|$$
		Otherwise, in graph $Z_{i, j}$ the path $(s, s^{l_i}_{h_i})\circ \pi[s^{l_i}_{h_i}, a_{i, j}]\circ (x_{i, j}, a_{i, j})$ has weight equal to $|\rho_{i, j}[s, x_{i, j}]|$. In both cases, we have $|\rho_{i, j}[s, x_{i, j}]| = \dist(s, x_{i, j}, Z_{i, j})$. Symmetrically, we can also argue $|\rho_{i, j}[y_{i, j}, t]| = \dist(y_{i, j}, t, Z_{i, j})$. This proves $|\rho_{i, j}| = \dist(s, t, Z_{i, j})$.
		
		For the rest, assume $y_{i, j}$ lies between $s$ and $x_{i, j}$ on $\rho_{i, j}$. Let us argue that $|\rho_{i, j}[s, y_{i, j}]| = \dist(s, y_{i, j}, Z_{i, j})$. There are two cases below.
		\begin{itemize}[leftmargin=*]
			\item $\rho_{i, j}[s, y_{i, j}]$ does not pass through $s^{l_j}_{h_j}$.
			
			In this case, $\rho_{i, j}(a_{i, j}, y_{i, j}]$ first lands on $\pi$ at a vertex $v$ in $\alpha^{l_j}_{h_j}$. If $a_{i, j}\notin V\brac{\alpha^{l_i}_{h_i}}$, Then, in graph $Z_{i, j}$ the path $(s, v)\circ \pi[v, y_{i, j}]$ has weight $|\rho_{i, j}[s, y_{i, j}]|$. Otherwise, if $a_{i, j}\in V\brac{\alpha^{l_i}_{h_i}}$, then in graph $Z_{i, j}$ the path $(s, s^{l_i}_{h_i})\circ \pi[s^{l_i}_{h_i}, a_{i, j}]\circ (a_{i, j}, v)\circ \pi[v, y_{i, j}]$ has weight $|\rho_{i, j}[s, y_{i, j}]|$.
			
			\item $\rho_{i, j}[s, y_{i, j}]$ passes through $s^{l_j}_{h_j}$.
			
			If $a_{i, j}\notin V\brac{\alpha^{l_i}_{h_i}}$, Then, in graph $Z_{i, j}$ the path $(s, s^{l_j}_{h_j})\circ \pi[s^{l_j}_{h_j}, y_{i, j}]$ has weight $|\rho_{i, j}[s, y_{i, j}]|$. Otherwise, if $a_{i, j}\in V\brac{\alpha^{l_i}_{h_i}}$, then in graph $Z_{i, j}$ the path $(s, s^{l_i}_{h_i})\circ \pi[s^{l_i}_{h_i}, a_{i, j}]\circ (a_{i, j}, s^{l_j}_{h_j})\circ \pi[s^{l_j}_{h_j}, y_{i, j}]$ has weight $|\rho_{i, j}[s, y_{i, j}]|$.
		\end{itemize}
	
		Symmetrically, we can also argue $|\rho_{i, j}[x_{i, j}, t]| = \dist(x_{i,j}, t, Z_{i,j})$. Finally, as there is a shortcut edge $(y_{i, j}, x_{i, j})$ with weigh $\wts(y_{i, j}, x_{i, j}) = \dist(y_{i, j}, x_{i, j}, H\cup \pi[x_{i, j}, y_{i, j}]) = |\rho_{i, j}[y_{i, j}, x_{i, j}]|$, we can conclude that $|\rho_{i, j}| = \dist(s, t, Z_{i, j})$.
\end{proof}

\begin{claim}
	In sub-case 3(c), the algorithm correctly computes the value of $|\rho_{i, j}|$.
\end{claim}
\begin{proof}
	If $a_{i,j}\notin V\brac{\alpha^{l_i}_{h_i}}$, then $\rho_{i, j}$ skips over the entire sub-path $\alpha^{l_i}_{h_i}$. As $Z^{(l_i, h_i)}_j$ contains all shortcuts among vertices in $V(\beta_{l_i})$ as well as all edges in $E(\beta_{l_i})\setminus \brac{V(\alpha^{l_i}_{h_i})\cup \{e_j\}}$, we can argue that $|\rho_{i, j}| = \dist\brac{s, t, Z^{(l_i, h_i)}_j}$.
	
	Otherwise, let us assume $a_{i,j}\in V\brac{\alpha^{l_i}_{h_i}}$. Then, $\rho_{i, j}[a_{i, j}, t]$ skips over the entire sub-path $\alpha^{l_i}_{h_i}$. For the same reason, as $Z^{(l_i, h_i)}_j$ contains all shortcuts among vertices in $V(\beta_{l_i})$ as well as all edges in $E(\beta_{l_i})\setminus \brac{V(\alpha^{l_i}_{h_i})\cup \{e_j\}}$, we can argue that $|\rho_{i, j}[a_{i, j}, t]| = \dist\brac{a_{i, j}, t, Z^{(l_i, h_i)}_j}$. Finally, by enumerating all possible choices of $a_{i, j}\in V\brac{\pi[s^{l_i}_{h_i}, u_i]}$, the value of $|\rho_{i, j}|$ can be retrieved by the minimum formula
	$$|\rho_{i, j}| = \min_{v\in V\brac{\pi[s^{l_i}_{h_i}, u_i]}}\left\{|\pi[s, v]| + \dist\brac{v, t, Z^{(l_i, h_i)}_j}\right\}$$
\end{proof}

\paragraph{Proof of \Cref{2fail-exact}}
To summarize, choosing $L = \ceil{n^{5/7}}$ and $g = \ceil{n^{3/7}}$ yields the overall runtime of $\tilde{O}(n^{3-1/7})$.

\section{Algebraic algorithm for weighted digraphs}\label{alge}
For any two vertices $u, v\in V(\pi)$ such that $u\in V(\pi[s, v])$, the backward distance is defined as:
$$\bw(v, u) = \dist(v, u, G\setminus (E(\pi[s, u])\cup E(\pi[v, t])))$$
It was shown in \cite{williams2022algorithms} that all backward distances can be computed in subcubic time using fast matrix multiplication.
\begin{lemma}[\cite{williams2022algorithms}]\label{bw-alge}
	There is a randomized algorithm that computes the values of all $\bw(*, *)$ in time $\tilde{O}(M^{\frac{1}{3}}n^{2+\frac{\omega}{3}})$.
\end{lemma}

Let $g<L$ be two parameters to be chosen later; note that $L, g$ do not have to be the same as in the previous section. Divide the shortest path from $s$ to $t$ into sub-paths of hops (number of edges) $L$; that is, $\pi = \gamma_1\circ\gamma_2\circ\cdots\circ\gamma_{\ceil{n/L}}$, and let $s_i, t_i$ be the head and tail of sub-path $\gamma_i$. For each sub-path $\gamma_l$, subdivide it into segments of hops $g$ as $\gamma_l = \alpha^l_1\circ \alpha^l_2\circ\cdots\circ \alpha^l_{\ceil{L/g}}$, where $\alpha^l_k$ goes from vertex $s^l_k$ to $t^l_k$. As preparation, perform the following steps.
\begin{itemize}[leftmargin=*]
	\item For each $\alpha^l_h$, compute single-source shortest paths from $s$ in $G\setminus E\brac{\pi[s^l_h, t]}$ which takes runtime $\tilde{O}(n^3/g)$ using \Cref{sssp-neg}; symmetrically, compute single-source shortest paths to $t$ in $G\setminus E\brac{\pi[s, t^l_h]}$ as well.
	\item Apply \Cref{dso} in graph $G\setminus E(\pi)$ which takes runtime $\tilde{O}(Mn^{2.8719})$.
	\item For each sub-path $\gamma_l$, compute the single-source replacement paths algorithm from \Cref{ssrp-alge} on graph $G\setminus E(\gamma_l)$ from source $s$ and terminal $V(\gamma_l)\cup \{t\}$. Symmetrically, apply \Cref{ssrp-alge} on (the reversed version of) $G\setminus E(\gamma_l)$ to source $t$ and terminal $V(\gamma_l)\cup \{s\}$. This step takes time $\tilde{O}(Mn^{\omega+1}/L + M^{\frac{1}{4-\omega}}n^{2+\frac{1}{4-\omega}})$.
	\item Compute all-pairs shortest paths in graph $G\setminus E(\pi)$ which takes time $\tilde{O}(M^{\frac{1}{4-\omega}}n^{2+\frac{1}{4-\omega}})$ by \Cref{apsp-neg}.
\end{itemize}

\subsection{One failure on the $st$-path}
In this subsection, we study the case where only one edge failure lies on the shortest path $\pi$ from $s$ to $t$. We prove the following statement.
\begin{lemma}\label{1fail-alge-alg}
	All values of $\dist(s, t, G\setminus \{f_1, f_2\})$ can be computed in the following amount of time, where $f_1\in E(\pi), f_2\notin E(\pi)$:
	$$\tilde{O}\brac{Mn^{\omega+0.5} + Mn^{\omega+1}/L + M^{\frac{1}{4-\omega}}n^{2+\frac{1}{4-\omega}} + n^2g^2 + n^2L}$$
\end{lemma}

Given a pair of edge failures $\{f_1, f_2\}$ where $f_1\in E(\pi), f_2\notin E(\pi)$, assume $f_1\in E(\alpha^l_h)$. Let $\rho$ be an optimal canonical replacement path from $s$ to $t$ avoiding $\{f_1, f_2\}$, and let $x, y\in V(\pi)$ be the vertices where $\rho$ diverges and converges with $\pi$. For the rest, let us consider several cases and deal with them separately.

\subsubsection*{Case 1: $x, y\in V(\alpha^l_h)$}
For each possible choice of edge failure $f_1$, build the following sketch graph $H_{f_1, f_2}$ with edge weight function $\mu$.
\begin{itemize}[leftmargin=*]
	\item \textbf{Vertices.} Add vertices $\{s, t\}\cup V(\alpha^l_h)$ to $H_{f_1, f_2}$.
	\item \textbf{Edges.} Add edges $(s, s^l_h)$ and $(t^l_h, t)$ with weights $\mu(s, s^l_h) = \dist(s, s^l_h, G)$ and $\mu(t^l_h, t) = \dist(t^l_h, t, G)$ respectively. Then, add edges $E(\alpha^l_h)\setminus \{f_1\}$ with the same weights under $\wts(*,*)$. After that, for each pair of vertices $u, v\in V(\alpha^l_h)$, add a shortcut edge $(u, v)$ with edge weight $\mu(u, v) = \dist(u, v, G\setminus (E(\pi)\cup \{f_2\}))$ which can be queried in $\tilde{O}(1)$ time using the distance sensitivity oracle on $G\setminus E(\pi)$.
\end{itemize}
By definition of $H_{f_1, f_2}$, we know that this sketch graph encodes the replacement path $\rho$. Therefore, we can apply \Cref{sssp-neg} on $H_{f_1, f_2}$ to compute the value of $\wts(\rho) = \dist(s, t, H_{f_1, f_2})$, and we update $\est(s, t, G\setminus \{f_1, f_2\})$ with this value. The runtime of this algorithm is bounded by $\tilde{O}(g^2)$ for each choice of $f_1$, and overall it takes time $\tilde{O}(n^2g^2)$.

\subsubsection*{Case 2: $x, y\in V(\gamma_l)$, and one of $x, y$ is not in $V(\alpha^l_h)$}
Without loss of generality, assume $x\notin V(\alpha^l_h)$. Fix any vertex $v\in V(\pi[s^l_h, t_l])$ and any choice of $f_2$, calculate the value which takes $\tilde{O}(L)$ time:
$$\mu(s, v) = \min_{u\in V\brac{\gamma_l[*, s^l_h]}}\{\dist(s, u, G) + \dist(u, v, G\setminus(E(\pi)\cup \{f_2\})) \}$$
After that, to compute $\dist(s, t, G\setminus \{f_1, f_2\})$, update the estimation $\est(s, t, G\setminus \{f_1, f_2\})$ with the following quantity which takes $\tilde{O}(L)$ time:
$$\min_{v\in V(\pi(f_1, t_l])}\{\mu(s, v) + \dist(v, t, G)\}$$
Overall, the runtime would be bounded by $\tilde{O}(n^2L)$. 

For the proof of correctness, clearly by the definition of $\mu(s, v)$, it encodes the weighted length of a path from $s$ to $v$ that does not use edges from $E(\gamma_l[*, f_1))\cup \{f_2\}$, and so $\est(s, t, G\setminus \{f_1, f_2\})$ never gives an underestimation. On the other hand, if we take $u = x, v = y$, then as $\rho$ is canonical, we know that $\rho = \rho[s, y] + \dist(y, t, G)\leq \mu(s, y) + \dist(y, t, G)$. Therefore, our algorithm gives the correct answer for this case.

\subsubsection*{Case 3: one of $x, y$ is not in $V(\gamma_l)$}
Without loss of generality, assume that $x\notin V(\gamma_l)$. If $y\notin V(\gamma_l)$, then $\rho$ skips over the entire sub-path $\gamma_l$. In this case, we have already computed the value of $\dist(s, t, G\setminus (E(\gamma_l)\cup \{f_2\}))$ during the preparatory steps. Otherwise, if $y\in V(\gamma_l)$, we can compute the value of $\wts(\rho)$ by the following minimum clause:
$$\min_{v\in V\brac{\pi(f_1, t_l]}}\{\dist(s, v, G\setminus (E(\gamma_l)\cup \{f_2\})) + \dist(v, t, G)\}$$
Overall, the runtime of this case is bounded by $\tilde{O}(n^2L)$.


\subsection{Both failures on the $st$-path}
In this subsection, we study the case where both edge failures lie on the shortest path $\pi$ from $s$ to $t$. We prove the following statement.
\begin{lemma}\label{2fail-alge-alg}
	All values of $\dist(s, t, G\setminus \{f_1, f_2\})$ can be computed in the following amount of time, where $f_1, f_2\in E(\pi)$:
	$$\tilde{O}\brac{Mn^{\omega+0.5} + Mn^{\omega+1}/L + M^{\frac{1}{4-\omega}}n^{2+\frac{1}{4-\omega}} + n^2g^2 + n^3/g + n^2L}$$
\end{lemma}

Given a pair of edge failures $\{f_1, f_2\}$ where $f_1, f_2\in E(\pi)$ ($f_1$ precedes $f_2$ on $\pi$.), assume $f_1\in E(\alpha^{l_1}_{h_1}), f_2\in E(\alpha^{l_2}_{h_2})$. Let $\rho$ be an optimal canonical replacement path from $s$ to $t$ avoiding $\{f_1, f_2\}$. Define $x\in V(\pi[s, f_1))$ to be the vertex where $\rho$ diverges from $\pi$, and define $y\in V(\pi(f_2, t])$ to be the vertex where $\rho$ converges with $\pi$. Take the vertex $a\in V(\rho)\cap V(\pi(f_1, t_{l_1}])$ which is the closest to $f_1$ on $\pi$, and symmetrically the vertex $b\in V(\rho)\cap V(\pi[s_{l_2}, f_2))$ which is the closest to $f_2$ on $\pi$; we will discuss the case where $a, b$ might not exist.

We can first assume that $l_1 < l_2$; otherwise if $l_1 = l_2$, this case was already studied in \cite{williams2022algorithms} (Section 5.4), and the total runtime was bounded by at most $\tilde{O}(nL^3)$. For the rest, let us assume $l_1 < l_2$.

\subsubsection*{Case 1: $a$ or $b$ does not exist}
Without loss of generality, assume $b$ does not exist; in other words, $\rho$ does not intersect the path $\pi[s_{l_2}, f_2)$. If $\rho$ does not intersect $\pi(f_2, t_{l_2}]$ as well, then $\rho$ avoids the entire path $\gamma_{l_2}$. Hence, $\wts(\rho)$ would be equal to $\dist(s, t, G\setminus(E(\gamma_{l_2})\cup \{f_1\}))$, which was already computed during our preparation steps. Otherwise, if $\rho$ intersects with $\pi(f_2, t_{l_2}]$, then since $\rho$ is canonical, its weighted length is equal to the following minimum clause:
$$\min_{v\in V\brac{\pi(f_2, t_{l_2}]}}\{\dist(s, v, G\setminus(E(\gamma_{l_2})\cup \{f_1\})) + \dist(v, t, G)\}$$
This formula can be computed in $O(L)$ time as all terms $\dist(s, v, G\setminus(E(\gamma_{l_2})\cup \{f_1\}))$ are already computed during the preparation phase. Overall, the runtime is bounded by $O(n^2L)$.

\subsubsection*{Case 2: $a\in V(\rho[s, b])$}
In this case, $\rho = \pi[s, x]\circ \rho[x, a]\circ \pi[a, b]\circ \rho[b, y]\circ \pi[y, t]$, where $\rho[x, a], \rho[b, y]$ are paths in $G\setminus E(\pi)$. Therefore, $\pi$ can also be decomposed as $\rho = \rho[s, t_{l_1}]\circ \pi[t_{l_1}, s_{l_2}]\circ \rho[s_{l_2}, t]$. As $\rho[s, t_{l_1}]$ avoids the path $\gamma_{l_2}$, $\wts(\rho[s, t_{l_1}])$ is equal to $\dist(s, t, \dist(s, v, G\setminus(E(\gamma_{l_2})\cup \{f_1\})))$; similarly, as $\rho[s_{l_2}, t]$ avoids the path $\gamma_{l_1}$, $\wts(\rho[s_{l_2}, t])$ is equal to $\dist(s, v, G\setminus(E(\gamma_{l_1})\cup \{f_2\}))$. Hence, $\wts(\rho)$ can be calculated as:
$$\dist(s, t, \dist(s, v, G\setminus(E(\gamma_{l_2})\cup \{f_1\}))) + \dist(t_{l_1}, s_{l_2}, G) + \dist(s, v, G\setminus(E(\gamma_{l_1})\cup \{f_2\}))$$
The overall runtime would be $O(n^2)$.

\subsubsection*{Case 3: $b\in V(\rho[s, a])$, and $a\in V\brac{\alpha^{l_1}_{h_1}}, b\in V\brac{\alpha^{l_2}_{h_2}}$}
Build the sketch graph $H_{f_1, f_2}$ with edge weights $\mu(*, *)$ as following.
\begin{itemize}[leftmargin=*]
	\item \textbf{Vertices.} Add $\{s, t\}\cup V\brac{\alpha^{l_1}_{h_1}}\cup V\brac{\alpha^{l_2}_{h_2}}$ to $H_{f_1, f_2}$.
	
	\item \textbf{Edges.} Add the following types of edges to $H_{f_1, f_2}$.
	
	\begin{enumerate}[(i),leftmargin=*]
		\item Add edges $E\brac{\alpha^{l_1}_{h_1}}\cup E\brac{\alpha^{l_2}_{h_2}}\setminus \{f_1, f_2\}$ to $H_{f_1, f_2}$. Add to $H_{f_1, f_2}$ the edge $\brac{s, s^{l_1}_{h_1}}$ with weight $\mu\brac{s, s^{l_1}_{h_1}} = \dist\brac{s, s^{l_1}_{h_1}, G}$, and the edge $\brac{t^{l_2}_{h_2}, t}$ with weight $\mu\brac{t^{l_2}_{h_2}, t} = \dist\brac{t^{l_2}_{h_2}, t, G}$.
		
		\item Secondly, for any pair of vertices $u\in V\brac{\alpha^{l_1}_{h_1}}, v\in V\brac{\alpha^{l_2}_{h_2}}$, add an edge $(u, v)$ with weight $\mu(u, v) = \dist\brac{u, v, G\setminus E(\pi)}$, an edge $(s, v)$ with weight 
		$$\mu(s, v) = \dist\brac{s, v, G\setminus E\brac{\pi[s^{l_1}_{h_1}, t]}}$$
		and an edge $(u, t)$ with weight $$\mu(u, t) = \dist\brac{u, t, G\setminus E\brac{\pi[s, t^{l_2}_{h_2}]}}$$
		
		\item For each $u\in V\brac{\alpha^{l_1}_{h_1}[*, f_1)}$, add an edge $\brac{u, s^{l_2}_{h_2}}$ with edge weight
		$$\mu\brac{u, s^{l_2}_{h_2}} = \min_{z\in V\brac{\pi(f_1, s^{l_2}_{h_2}]}}\left\{\dist(u, z, G\setminus E(\pi)) + \dist\brac{z, s^{l_2}_{h_2}, G} \right\}$$
		Add edge $\brac{s, s^{l_2}_{h_2}}$ with edge weight
		$$\mu\brac{s, s^{l_2}_{h_2}} = \min_{z\in V\brac{\pi(f_1, s^{l_2}_{h_2}]}}\left\{\dist\brac{s, z, G\setminus E\brac{\pi[s^{l_1}_{h_1}, t]}} + \dist\brac{z, s^{l_2}_{h_2}, G} \right\}$$
		
		Symmetrically, for each $v\in V\brac{\alpha^{l_2}_{h_2}(f_2, *]}$, add an edge $\brac{t^{l_1}_{h_1}, v}$ with edge weight
		$$\mu\brac{t^{l_1}_{h_1}, v} = \min_{z\in V\brac{\pi[t^{l_1}_{h_1}, f_2)}}\left\{\dist(z, v, G\setminus E(\pi)) + \dist\brac{t^{l_1}_{h_1}, z, G} \right\}$$
		Add edge $\brac{t^{l_1}_{h_1}, t}$ with edge weight
		$$\mu\brac{t^{l_1}_{h_1}, t} = \min_{z\in V\brac{\pi[t^{l_1}_{h_1}, f_2)}}\left\{\dist\brac{z, t, G\setminus E\brac{\pi[s, t^{l_2}_{h_2}]}} + \dist\brac{t^{l_1}_{h_1}, z, G} \right\}$$
		
		Each of the minimum clauses can be computed in $\tilde{O}(1)$ time using standard interval data structures.
		
		\item Finally, add an edge $(v, u)$ with weight $\mu(v, u) = \bw(v, u)$ if $u\in V\brac{\pi(f_1, t^{l_1}_{h_1}]}, v\in V\brac{\pi[s^{l_2}_{h_2}, f_2)}$.
	\end{enumerate}
\end{itemize}
After that, apply \Cref{sssp-neg} on $H_{f_1, f_2}$ to compute $\dist(s, t, H_{f_1, f_2})$ and update $\est(s, t, G\setminus \{f_1, f_2\})$ with it, which takes runtime $\tilde{O}(g^2)$.

\begin{claim}
	$\wts(\rho) = \dist(s, t, H_{f_1, f_2})$.
\end{claim}
\begin{proof}
	It is clear that $\dist(s, t, H_{f_1, f_2})$ is not an underestimation as all shortcut edges in $H_{f_1, f_2}$ correspond to paths in $G$ avoiding $\{f_1, f_2\}$.
	
	Decompose $\rho$ as $\rho = \rho[s, b]\circ \rho[b, a]\circ\rho[a, t]$. By definition of $a, b$, we know that $\rho[b, a]$ is a shortest backward path in $G\setminus (E(\pi[s, a])\cup E(\pi[b, t]))$, and so $\wts(\rho[b, a]) = \bw(b, a)$. Now it suffices to argue that both $\rho[s, b]$ and $\rho[a, t]$ are captured in the sketch graph $H_{f_1, f_2}$. Let us only focus on the path $\rho[s, b]$; the analysis for $\rho[a, t]$ is symmetric.
	
	Let $z$ be the first vertex on $V(\rho[s, b])$ that lands on  $\pi(f_1, f_2)$. There are several cases depending on the position of $x$ and $z$.
	\begin{itemize}[leftmargin=*]
		\item $z\in V(\alpha^{l_2}_{h_2})$ and $x\in V\brac{\alpha^{l_1}_{h_1}[*, f_1)}$.
		
		In this case, $\rho[x, z]$ is encoded by the shortcut edge $(x, z)$ in $H_{f_1, f_2}$ with weight $\mu(x, z) = \dist(x, z, G\setminus E(\pi))$, and thus $\wts(\rho[s, b]) = \mu(s, s^{l_1}_{h_1}) + \wts\brac{\pi[s^{l_1}_{h_1}, x]} + \mu(x, z) + \dist(z, b, G)$ which corresponds to a path in $H_{f_1, f_2}$.
		
		\item $z\in V(\alpha^{l_2}_{h_2})$ and $x\notin V\brac{\alpha^{l_1}_{h_1}[*, f_1)}$.
		
		For a similar reason, $\wts(\rho[s, b])$ would be equal to $\mu(s, b) = \dist\brac{s, b, G\setminus E\brac{\pi[s^{l_1}_{h_1}, t]}}$. Hence, the prefix $\rho[s, b]$ is still captured in $H_{f_1, f_2}$.
		
		\item $z\notin V(\alpha^{l_2}_{h_2})$ and $x\in V\brac{\alpha^{l_1}_{h_1}[*, f_1)}$.
		
		In this case, as $\rho$ is canonical, we can further decompose $\rho[x, b]$ as:
		$$\rho[x, b] = \rho[x, z]\circ \pi[z, b] = \rho[x, z]\circ \pi[z, s^{l_2}_{h_2}]\circ \pi[s^{l_2}_{h_2}, b]$$
		By the construction of $H_{f_1, f_2}$, there exists a shortcut edge $\brac{x, s^{l_2}_{h_2}}$ with weight $\mu\brac{x, s^{l_2}_{h_2}} = \wts(\rho[x, z]) + \dist\brac{z, s^{l_2}_{h_2}, G}$. As for path $\pi[s^{l_2}_{h_2}, b]$, it belongs to $H_{f_1, f_2}$ by construction. Hence, $\wts(\rho[x, b]) = \dist(x, b, H_{f_1, f_2})$.
		
		Recall that when $x\in V\brac{\alpha^{l_1}_{h_1}[*, f_1)}$, there is a path in $H_{f_1, f_2}$ from $s$ to $x$ with weight $\dist(s, x, G)$. So we know $\wts(\rho[s, b]) = \dist(s, b, H_{f_1, f_2})$.
		
		\item $z\notin V(\alpha^{l_2}_{h_2})$ and $x\notin V\brac{\alpha^{l_1}_{h_1}[*, f_1)}$.
		
		In this case, as $\rho$ is canonical, we can further decompose $\rho[s, b]$ as:
		$$\rho[s, b] = \rho[s, z]\circ \pi[z, b] = \rho[s, z]\circ \pi[z, s^{l_2}_{h_2}]\circ \pi[s^{l_2}_{h_2}, b]$$
		By the construction of $H_{f_1, f_2}$, since $\rho[s, z]$ avoids $\pi[s^{l_1}_{h_1}, t]$ entirely, there exists a shortcut edge $\brac{s, s^{l_2}_{h_2}}$ with weight $\mu\brac{s, s^{l_2}_{h_2}} = \wts(\rho[s, z]) + \dist\brac{z, s^{l_2}_{h_2}, G}$. As for path $\pi[s^{l_2}_{h_2}, b]$, it belongs to $H_{f_1, f_2}$ by construction. Hence, $\wts(\rho[s, b]) = \dist(s, b, H_{f_1, f_2})$.
		
	\end{itemize}
	
\end{proof}

\subsubsection*{Case 4: $b\in V(\rho[s, a])$, plus that $a\notin V\brac{\alpha^{l_1}_{h_1}}$ or $b\notin V\brac{\alpha^{l_2}_{h_2}}$}
Without loss of generality, assume $b\notin V\brac{\alpha^{l_2}_{h_2}}$. Build the sketch graph $H_{f_1, (l_2, h_2)}$ with edge weights $\mu(*, *)$ as following.
\begin{itemize}[leftmargin=*]
	\item \textbf{Vertices.} Add $\{s, t\}\cup V\brac{\gamma_{l_1}}\cup V\brac{\gamma_{l_2}}$ to $H_{f_1, (l_2, h_2)}$.
	
	\item \textbf{Edges.} Add the following types of edges to $H_{f_1, (l_2, h_2)}$.
	
	\begin{enumerate}[(i),leftmargin=*]
		\item Add edges $E\brac{\gamma_{l_1}}\cup E\brac{\gamma_{l_2}}\setminus \brac{\{f_1\}\cup E\brac{\alpha^{l_2}_{h_2}}}$ to $H_{f_1, (l_2, h_2)}$. Add edge $(s, s_{l_1})$ with weight $\mu(s, s_{l_1}) = \dist(s, s_{l_1}, G)$, and edge $(t_{l_2}, t)$ with weight $\mu(t_{l_2}, t) = \dist(t_{l_2}, t, G)$.
		
		\item Secondly, for any $u\in V(\gamma_{l_1}), v\in V(\gamma_{l_2})$, add edge $(u, v)$ with weight $\mu(u, v) = \dist(s, t, G\setminus E(\pi))$, and edge $(s, v)$ with weight $\mu(s, v) = \dist\brac{s, v, G\setminus E\brac{\pi[s^{l_1}_{h_1}, t]}}$, and edge $(u, t)$ with weight $\mu(u, t) = \dist\brac{u, t, G\setminus E\brac{\pi[s, t^{l_2}_{h_2}]}}$.
		
		\item For each $u\in V(\gamma_{l_1}[*, f_1))$, add an edge $(u, s_{l_2})$ with edge weight:
		$$\mu(u, s_{l_2}) = \min_{z\in V(\pi(f_1, s_{l_2}])}\left\{\dist(u, z, G\setminus E(\pi)) + \dist(z, s_{l_2}, G) \right\}$$
		Besides, add an edge $(s, s_{l_2})$ with edge weight:
		$$\mu(s, s_{l_2}) = \min_{z\in V(\pi(f_1, s_{l_2}])}\left\{\dist\brac{s, z, G\setminus E(\pi[s^{l_1}_{h_1}, t])} + \dist(z, s_{l_2}, G) \right\}$$
		
		Similarly, for each $v\in V\brac{\gamma_{l_2}[s^{l_2}_{h_2}, *]}$, add an edge $(t_{l_1}, v)$ with edge weight:
		$$\mu(t_{l_1}, v) = \min_{z\in V\brac{\pi[t_{l_1}, s^{l_2}_{h_2}]}}\left\{\dist\brac{z, v, G\setminus E(\pi)} + \dist(t_{l_1}, z, G)\right\}$$
		
		Besides, add an edge $(t_{l_1}, t)$ with edge weight:
		$$\mu(t_{l_1}, t) = \min_{z\in V\brac{\pi[t_{l_1}, s^{l_2}_{h_2}]}}\left\{ \dist\brac{z, t, G\setminus E(\pi[s, t^{l_2}_{h_2}])} + \dist(t_{l_1}, z, G) \right\}$$
		
		Each of the minimum clauses can be computed in $\tilde{O}(1)$ time using standard interval data structures.
		
		\item Finally, add an edge $(v, u)$ with weight $\mu(v, u) = \bw(v, u)$ if $u\in V\brac{\pi(f_1, t_{l_1}]}, v\in V\brac{\pi[s_{l_2}, f_2)}$.
	\end{enumerate}	
\end{itemize}

After that, apply \Cref{sssp-neg} on $H_{f_1, (l_2, h_2)}$ to compute $\dist(s, v, H_{f_1, (l_2, h_2)})$ for each $v\in \{t\}\cup V(\pi(f_2, t_{l_2}])$. Then, update the value of $\est(s, t, G\setminus \{f_1, f_2\})$ with the following quantity:
$$\min_{v\in V\brac{\pi(f_2, t_{l_2}]}}\{\dist(s, v, H_{f_1, (l_2, h_2)}) + \dist(v, t, G), \dist(s, t, H_{f_1, (l_2, h_2)})\}$$
The total runtime would be $\tilde{O}(n^2L^2/g + n^2L)$.

\begin{claim}
	$\wts(\rho) = \min_{v\in V\brac{\pi(f_2, t_{l_2}]}} \left\{\dist\brac{s, v, H_{f_1, (l_2, h_2)}} + \dist(v, t, G), \dist\brac{s, t, H_{f_1, (l_2, h_2)}}\right\}$.
\end{claim}
\begin{proof}
	Let $z$ be the first vertex on $V(\rho[s, b])$ that lands on  $\pi(f_1, f_2)$. We first show that $\wts(\rho[s, b]) = \dist(s, b, H_{f_1, (l_2, h_2)})$. There are several cases depending on the position of $x$ and $z$.
	
	\begin{itemize}[leftmargin=*]
		\item $z\in V(\gamma_{l_2})$ and $x\in V\brac{\gamma_{l_1}}$.
		
		In this case, $\rho[x, b]$ can be decomposed as $\rho[x, b] = \rho[x, z]\circ \pi[z, b]$. As $\rho[x, z]$ lies in $G\setminus E(\pi)$, we know that $\wts(\rho[x, b]) = \mu(x, z) + \dist(z, b, G)$ which corresponds to a path in $H_{f_1, (l_2, h_2)}$.
		
		As for $\rho[s, x]$, as $x\in V\brac{\gamma_{l_1}}$, $\rho[s, x] = \pi[s, s_{l_1}]\circ \pi[s_{l_1}, x]$, it corresponds to a path in $H_{f_1, (l_2, h_2)}$. Therefore, we have $\wts(\rho[s, b]) = \dist(s, b, H_{f_1, (l_2, h_2)})$.
		
		\item $z\in V(\gamma_{l_2})$ and $x\notin V\brac{\gamma_{l_1}}$.
		
		In this case, $\rho[s, b]$ can be decomposed as $\rho[s, b] = \rho[s, z]\circ \pi[z, b]$. As $\rho[s, z]$ lies in $G\setminus E(\pi[s^{l_1}_{h_1}, t])$, we know that $\wts(\rho[s, b]) = \mu(s, z) + \dist(z, b, G)$ which corresponds to a path in $H_{f_1, (l_2, h_2)}$.
		
		\item $z\notin V(\gamma_{l_2})$ and $x\in V\brac{\gamma_{l_1}}$.
		
		In this case, as $\rho$ is canonical, we can decompose $\rho[x, b]$ as $\rho[x, b] = \rho[x, z]\circ \pi[z, s_{l_2}]\circ \pi[s_{l_2}, b]$. By construction of $H_{f_1, (l_2, h_2)}$, it contains a shortcut edge $(x, s_{l_2})$ with weight $\mu(x, s_{l_2}) \leq \wts(\rho[x, z]) + \dist(z, s_{l_2}, G)$.
		
		As for $\rho[s, x]$, as $x\in V\brac{\gamma_{l_1}}$, $\rho[s, x] = \pi[s, s_{l_1}]\circ \pi[s_{l_1}, x]$, it corresponds to a path in $H_{f_1, (l_2, h_2)}$. Therefore, we have $\wts(\rho[s, b]) = \dist(s, b, H_{f_1, (l_2, h_2)})$.
		
		\item $z\notin V(\gamma_{l_2})$ and $x\notin V\brac{\gamma_{l_1}}$.
		
		In this case, as $\rho$ is canonical, we can decompose $\rho[x, b]$ as $\rho[s, b] = \rho[s, z]\circ \pi[z, s_{l_2}]\circ \pi[s_{l_2}, b]$. As $\rho[s, z]$ lies in $G\setminus E(\pi[s^{l_1}_{h_1}, t])$, we know that $\wts(\rho[s, b]) = \mu(s, z) + \dist(z, s_{l_2}, G) +\dist(s_{l_2}, b, G)$ which corresponds to a path in $H_{f_1, (l_2, h_2)}$ for a similar reason as in the previous case.
	\end{itemize}
	With symmetric analysis we can show that $\wts(a, t) = \dist(a, t, H_{f_1, (l_2, h_2)})$ if $y\notin V(\gamma_{l_2})$, and $\wts(a, y) = \dist(a, y, H_{f_1, (l_2, h_2)})$ if $y\in V(\gamma_{l_2})$.
	
	Finally, by construction of $H_{f_1, (l_2, h_2)}$, $(b, a)$ is a shortcut edge with weight $\bw(b, a)$ which is equal to $\rho[b, a]$. Therefore, $\wts(\rho[s, a]) = \dist(s, a, H_{f_1, (l_2, h_2)})$. Now, if $y\notin V(\gamma_{l_2})$, then by a path concatenation we have $\wts(\rho) = \dist(s, t, H_{f_1, (l_2, h_2)})$; otherwise if $y\in V(\gamma_{l_2})$, we can recover the value of $\wts(\rho)$ by the minimum clause (minimum achieved by $v = y$):
	$$\wts(\rho) = \min_{v\in V\brac{\pi(f_2, t_{l_2}]}}\{\dist(s, v, H_{f_1, (l_2, h_2)}) + \dist(v, t, G)\}$$
\end{proof}

\subsection{Proof of \Cref{faster-alge}}
Setting $L = \ceil{n^{\frac{3(\omega-1)}{7}}}$ and $g = \ceil{L^{2/3}}$ (under $\omega = 2.371552$ \cite{williams2024new,duan2023faster,alman2021refined,le2014powers,williams2012multiplying}) and applying \Cref{1fail-alge-alg} \Cref{2fail-alge-alg} concludes the proof.

\section*{Acknowledgment}
This work is part of a project that has received funding from the European Research Council (ERC) under the European Union’s Horizon 2020 research and innovation programme (grant agreement No 803118 UncertainENV).

\vspace{5mm}
\bibliographystyle{alpha}
\bibliography{ref}

\begin{thebibliography}{GPVWX21}

\bibitem[AW21]{alman2021refined}
Josh Alman and Virginia~Vassilevska Williams.
\newblock A refined laser method and faster matrix multiplication.
\newblock In {\em Proceedings of the 2021 ACM-SIAM Symposium on Discrete
  Algorithms (SODA)}, pages 522--539. SIAM, 2021.

\bibitem[Ber10]{bernstein2010nearly}
Aaron Bernstein.
\newblock A nearly optimal algorithm for approximating replacement paths and k
  shortest simple paths in general graphs.
\newblock In {\em Proceedings of the twenty-first annual ACM-SIAM symposium on
  Discrete Algorithms}, pages 742--755. SIAM, 2010.

\bibitem[BG04]{bhosle2004replacement}
Amit~M Bhosle and Teofilo~F Gonzalez.
\newblock Replacement paths for pairs of shortest path edges in directed
  graphs.
\newblock In {\em Proceedings of the 16th IASTED International Conference on
  Parallel and Distributed Computing and Systems}. Citeseer, 2004.

\bibitem[BNWN22]{bernstein2022negative}
Aaron Bernstein, Danupon Nanongkai, and Christian Wulff-Nilsen.
\newblock Negative-weight single-source shortest paths in near-linear time.
\newblock In {\em 2022 IEEE 63rd Annual Symposium on Foundations of Computer
  Science (FOCS)}, pages 600--611. IEEE, 2022.

\bibitem[CC20]{chechik2020distance}
Shiri Chechik and Sarel Cohen.
\newblock Distance sensitivity oracles with subcubic preprocessing time and
  fast query time.
\newblock In {\em Proceedings of the 52nd Annual ACM SIGACT Symposium on Theory
  of Computing}, pages 1375--1388, 2020.

\bibitem[CM20]{chechik2020near}
Shiri Chechik and Ofer Magen.
\newblock Near optimal algorithm for the directed single source replacement
  paths problem.
\newblock In {\em 47th International Colloquium on Automata, Languages, and
  Programming (ICALP 2020)}. Schloss Dagstuhl-Leibniz-Zentrum f{\"u}r
  Informatik, 2020.

\bibitem[CN20]{chechik2020simplifying}
Shiri Chechik and Moran Nechushtan.
\newblock Simplifying and unifying replacement paths algorithms in weighted
  directed graphs.
\newblock In {\em 47th International Colloquium on Automata, Languages, and
  Programming (ICALP 2020)}. Schloss Dagstuhl-Leibniz-Zentrum f{\"u}r
  Informatik, 2020.

\bibitem[CZ24]{chechik2024nearly}
Shiri Chechik and Tianyi Zhang.
\newblock Nearly optimal approximate dual-failure replacement paths.
\newblock In {\em Proceedings of the 2024 Annual ACM-SIAM Symposium on Discrete
  Algorithms (SODA)}, pages 2568--2596. SIAM, 2024.

\bibitem[Dij22]{dijkstra2022note}
Edsger~W Dijkstra.
\newblock A note on two problems in connexion with graphs.
\newblock In {\em Edsger Wybe Dijkstra: His Life, Work, and Legacy}, pages
  287--290. 2022.

\bibitem[DWZ23]{duan2023faster}
Ran Duan, Hongxun Wu, and Renfei Zhou.
\newblock Faster matrix multiplication via asymmetric hashing.
\newblock In {\em 2023 IEEE 64th Annual Symposium on Foundations of Computer
  Science (FOCS)}, pages 2129--2138. IEEE, 2023.

\bibitem[GPVWX21]{gu2021faster}
Yuzhou Gu, Adam Polak, Virginia Vassilevska~Williams, and Yinzhan Xu.
\newblock Faster monotone min-plus product, range mode, and single source
  replacement paths.
\newblock In {\em 48th International Colloquium on Automata, Languages, and
  Programming (ICALP 2021)}. Schloss Dagstuhl-Leibniz-Zentrum f{\"u}r
  Informatik, 2021.

\bibitem[GW12]{grandoni2012improved}
Fabrizio Grandoni and Virginia~Vassilevska Williams.
\newblock Improved distance sensitivity oracles via fast single-source
  replacement paths.
\newblock In {\em 2012 IEEE 53rd Annual Symposium on Foundations of Computer
  Science}, pages 748--757. IEEE, 2012.

\bibitem[GW19]{grandoni2019faster}
Fabrizio Grandoni and Virginia~Vassilevska Williams.
\newblock Faster replacement paths and distance sensitivity oracles.
\newblock {\em ACM Transactions on Algorithms (TALG)}, 16(1):1--25, 2019.

\bibitem[LG14]{le2014powers}
Fran{\c{c}}ois Le~Gall.
\newblock Powers of tensors and fast matrix multiplication.
\newblock In {\em Proceedings of the 39th international symposium on symbolic
  and algebraic computation}, pages 296--303, 2014.

\bibitem[Rod07]{roditty2007k}
Liam Roditty.
\newblock On the k-simple shortest paths problem in weighted directed graphs.
\newblock In {\em Proceedings of the eighteenth annual ACM-SIAM symposium on
  Discrete algorithms}, pages 920--928. Citeseer, 2007.

\bibitem[RZ05]{roditty2005replacement}
Liam Roditty and Uri Zwick.
\newblock Replacement paths and k simple shortest paths in unweighted directed
  graphs.
\newblock In {\em Automata, Languages and Programming: 32nd International
  Colloquium, ICALP 2005, Lisbon, Portugal, July 11-15, 2005. Proceedings 32},
  pages 249--260. Springer, 2005.

\bibitem[Wil11]{williams2011faster}
Virginia~Vassilevska Williams.
\newblock Faster replacement paths.
\newblock In {\em Proceedings of the twenty-second annual ACM-SIAM symposium on
  Discrete Algorithms}, pages 1337--1346. SIAM, 2011.

\bibitem[Wil12]{williams2012multiplying}
Virginia~Vassilevska Williams.
\newblock Multiplying matrices faster than coppersmith-winograd.
\newblock In {\em Proceedings of the forty-fourth annual ACM symposium on
  Theory of computing}, pages 887--898, 2012.

\bibitem[WW10]{williams2010subcubic}
Virginia~Vassilevska Williams and Ryan Williams.
\newblock Subcubic equivalences between path, matrix and triangle problems.
\newblock In {\em 2010 IEEE 51st Annual Symposium on Foundations of Computer
  Science}, pages 645--654. IEEE, 2010.

\bibitem[WWX22]{williams2022algorithms}
Virginia~Vassilevska Williams, Eyob Woldeghebriel, and Yinzhan Xu.
\newblock Algorithms and lower bounds for replacement paths under multiple edge
  failure.
\newblock In {\em 2022 IEEE 63rd Annual Symposium on Foundations of Computer
  Science (FOCS)}, pages 907--918. IEEE, 2022.

\bibitem[WXXZ24]{williams2024new}
Virginia~Vassilevska Williams, Yinzhan Xu, Zixuan Xu, and Renfei Zhou.
\newblock New bounds for matrix multiplication: from alpha to omega.
\newblock In {\em Proceedings of the 2024 Annual ACM-SIAM Symposium on Discrete
  Algorithms (SODA)}, pages 3792--3835. SIAM, 2024.

\bibitem[Zwi02]{zwick2002all}
Uri Zwick.
\newblock All pairs shortest paths using bridging sets and rectangular matrix
  multiplication.
\newblock {\em Journal of the ACM (JACM)}, 49(3):289--317, 2002.

\end{thebibliography}


\end{document}